\documentclass[11pt,journal,onecolumn]{IEEEtran}

\usepackage[T1]{fontenc}
\usepackage{url}
\usepackage{ifthen}
\usepackage{cite}
\usepackage[cmex10]{amsmath} 
\usepackage{amsmath}
\usepackage{amsfonts}
\usepackage{amssymb}

\usepackage{psfrag}
\usepackage{enumerate}
\usepackage{url}

\usepackage{graphicx}
\usepackage{color}

\newtheorem{theorem}{Theorem}
\newtheorem{lemma}[theorem]{Lemma}

\newtheorem{definition}{Definition}
\newtheorem{remark}{Remark}
\newtheorem{example}{Example}

\newcommand{\comment}[1]{}

\def\bX{{\bf X}}
\def\bY{{\bf Y}}
\def\bZ{{\bf Z}}
\def\bS{{\bf S}}
\def\bs{{\bf s}}
\def\bD{{\bf D}}
\def\bd{{\bf d}}

\def\bq{{\bf q}}

\def\bx{{\bf x}}
\def\by{{\bf y}}

\def\bW{{\bf W}}
\def\bhW{{\bf \hat{W}}}

\def\CR{\dot{R}}
\def\CC{\dot{C}}
\def\CCP{\dot{C}_{\bot} }

\def\CCPP{\dot{C}_{\bot \bot}}

\def\by{{\bf y}}

\def\bc{{\bf c}}
\def\bw{{\bf w}}

\def\bv{{\bf v}}
\def\bu{{\bf u}}

\def\Pm{P_m(k'_n)}

\def\bH{{\bf H}}

\def\CR{\dot{R}}

\def\CC{\dot{C}}
\def\CCe{\dot{C}_{\epsilon}}
\def\PR{\dot{R^{A}}}

\def\CCP{\dot{C}_{\bot}}

\def\el{\ell}

\def\cA{\mbox{$\cal{A}$}}

\def\cC{\mbox{$\cal{C}$}}
\def\cD{\mbox{$\cal{D }$}}

\def\cE{\mbox{$\cal{E }$}}

\def\cM{\mbox{$\cal{M }$}}

\def\cB{\mbox{$\cal{B}$}}

\def\cY{\mbox{$\cal{Y}$}}

\def\cW{\mbox{$\cal{W}$}}

\def\cK{\mbox{$\cal{K}$}}
\def\cT{\mbox{$\cal{T}$}}
\def\cI{\mbox{$\cal{I}$}}

\def\Pr{\text{Pr} }
\def\MS{\cT_{t}}

\def\cS{\mbox{$\cal{S}$}}

\def\cW{\mbox{$\cal{W}$}}
\def\cE{\mbox{$\cal{E}$}}

\def\cK{\mbox{$\cal{K}$}}

\def\cU{\mbox{$\cal{U}$}}

\def\cN{\mbox{$\cal{N}$}}

\newcommand{\edit}[1]{{\color{black}#1}}

\newcommand{\I}[1]{\mathbf{1} (#1)}

\begin{document}
	\IEEEoverridecommandlockouts
	\title{Scaling Laws for Gaussian  Random  \\Many-Access  Channels 
	}
	\author{
		\IEEEauthorblockN{Jithin~Ravi  and Tobias~Koch}
		\thanks{J.~Ravi and T.~Koch have received funding from the European Research Council (ERC) under the European Union's Horizon 2020 research and innovation programme (Grant No.~714161). T.~Koch has further received funding from the Spanish Ministerio de Econom\'ia y Competitividad under Grants RYC-2014-16332 and TEC2016-78434-C3-3-R (AEI/FEDER, EU).  The material in this paper was presented in part at the IEEE International Symposium on Information Theory, Paris, France, July 2019, at the International Zurich Seminar
			on Information and Communication, Zurich, Switzerland, February 2020, and at the  IEEE International Symposium on Information Theory, Los Angeles, CA, USA, June 2020.
			
 J. Ravi and	T. Koch are with the Signal Theory and Communications Department, Universidad Carlos III de Madrid, Spain, and with the Gregorio Mara\~n\'on  Health Research Institute, Madrid, Spain (emails: rjithin@tsc.uc3m.es, koch@tsc.uc3m.es).}
	}
	
	\maketitle
	\begin{abstract}
	This paper considers a Gaussian multiple-access channel  with random user activity where the total number of users $\ell_n$ and the average number of active users $k_n$ may grow with the blocklength $n$. For this channel, it studies the maximum number of bits that can be transmitted reliably per unit-energy as a function of $\ell_n$ and $k_n$. When all users are active with probability one, i.e., $\ell_n = k_n$, it is demonstrated that if $k_n$ is of an order strictly below $n/\log n$, then each user can achieve the single-user capacity per unit-energy $(\log e)/N_0$ (where $N_0/ 2$ is the noise power) by using an orthogonal-access scheme.  In contrast, if $k_n$ is of an order strictly above $n/\log n$, then the users cannot achieve any positive rate per unit-energy.  Consequently, there is a sharp transition between orders of growth where interference-free communication is feasible and orders of growth where reliable communication at a positive rate per unit-energy is infeasible. It is further demonstrated that \edit{orthogonal-access schemes in combination with orthogonal codebooks}, which achieve the capacity per unit-energy when the number of users is bounded, can be strictly suboptimal. 
	
	When the user activity is random, i.e., when $\ell_n$ and $k_n$ are different, it is demonstrated that if $k_n\log \ell_n$ is sublinear in $n$, then each user can achieve the single-user capacity per unit-energy $(\log e)/N_0$.  Conversely, if $k_n\log \ell_n$ is superlinear in $n$, then the users cannot achieve any positive rate per unit-energy. Consequently, there is again a sharp transition between orders of growth where interference-free communication is feasible and orders of growth where reliable communication at a positive rate is infeasible that depends on the asymptotic behaviours of both $\ell_n$ and $k_n$.  It is further demonstrated that orthogonal-access schemes, which are optimal when all users are active with probability one, can be strictly suboptimal in general.
\end{abstract}


\section{Introduction}

Chen~\emph{et al.} \cite{ChenCG17} introduced the many-access channel (MnAC) as a multiple-access channel (MAC) where the number of users grows with the blocklength and each user is active with a given probability. This model is motivated by systems consisting of a single receiver and many transmitters, the number of which is comparable or even larger than the blocklength. This situation may occur, \emph{e.g.}, in a machine-to-machine communication scenario with many thousands of devices in a given cell that are active only sporadically. In \cite{ChenCG17}, Chen~\emph{et al.} considered a Gaussian MnAC with $\el_n$ users, each of which is active with probability $\alpha_n$, and determined the number of messages $M_n$ each user can transmit reliably with a codebook of average power not exceeding $P$. Since then, MnACs have been studied in various papers under different settings.  

An example of a MAC is the uplink connection in a cellular network. Current cellular networks follow a \emph{grant-based} access protocols, i.e., an active device has to obtain permission from the base station to transmit data. In MnACs, this will lead to a large signalling overhead. \emph{Grant-free} access protocols, where active devices can access the network without a permission, were proposed to overcome this~\cite{LiuPopovski18}. 
The synchronization issues \edit{arising} in such scenarios have been studied by Shahi \emph{et al.}~\cite{ShahiTD18}.
 In some of the MnACs, such as sensor networks, detecting the identity of a device that sent a particular message may not be important. This scenario was studied under the name of \emph{unsourced} massive access by Polyanskiy~\cite{Polyanskiy17}, who further introduced the notion of \emph{per-user probability of error}.
  Specifically, \cite{Polyanskiy17} analyzed the minimum \emph{energy-per-bit} required to reliably transmit a message over an unsourced Gaussian MnAC where the number of active users grows linearly in the blocklength and each user's payload is fixed.  Low-complexity schemes for this setting were studied in many works~\cite{OrdentlichP17, Vem19, Amalladinne20, Fengler19, Fengler20}. Generalizations to quasi-static fading MnACs can be found in~\cite{KowshikPISIT19,KowshikP19,KowshikTCOM20}. Zadik \emph{et al.}~\cite{ZadikPT19} presented improved bounds on the tradeoff between user density and energy-per-bit for the many-access channel introduced in~\cite{Polyanskiy17}.
  

Related to energy per-bit is the \emph{capacity per unit-energy} $\CC$ which is defined as the largest number of bits per unit-energy that can be transmitted reliably over a channel. Verd\'u \cite{Verdu90} showed that $\CC$ can be obtained from the capacity-cost function $C(P)$, defined as the largest number of bits per channel use that can be transmitted reliably with average power per symbol not exceeding $P$, as 
\begin{align*}
\CC=\sup_{P>0} \frac{C(P)}{P}.
\end{align*}
 For the Gaussian channel with noise power $N_0/2$, this is equal to $\frac{\log e}{N_0}$. Verd\'u further showed that the capacity per unit-energy can be achieved by a codebook that is orthogonal in the sense that the nonzero components of different codewords do not overlap. \edit{Such a codebook corresponds to pulse position modulation (PPM) either in time or frequency domains.} In general, we shall say that a codebook is orthogonal if the inner product between different codewords is zero. The two-user Gaussian multiple access channel (MAC) was also studied in~\cite{Verdu90}, and it was demonstrated that both users can achieve the single-user capacity per unit-energy by timesharing the channel between the users, i.e., while one user transmits the other user remains silent. This is an \emph{orthogonal-access scheme} in the sense that the inner product between codewords of different users is zero.\footnote{Note, however, that in an orthogonal-access scheme the codebooks are not required to be orthogonal. That is, codewords of different codebooks are orthogonal to each other, but codewords of the same codebook need not be.} To summarize, in a two-user Gaussian MAC, both users can achieve the rate per unit-energy $\frac{\log e}{N_0} $ by combining an orthogonal-access scheme with orthogonal codebooks. This result can be directly generalized to any finite number of users.

 The picture changes when the number of users grows without bound with the blocklength $n$.
 In this paper,  we consider a setting where the total number of users $\el_n$ may grow as an arbitrary function of the blocklength and the probability $\alpha_n$ that a user is active may be a function of the blocklength, too. Contributions of this paper are as follows.

\begin{enumerate}
	\item First, we consider the capacity per unit-energy of the Gaussian MnAC as a function of the order of growth of users when all users are active with probability one. In Theorem~\ref{Thm_nonrandom}, we show that, if the order of growth is above $n/ \log n$, then the capacity per unit-energy is zero, and if the order of growth is below $n/ \log n$, then each user can achieve the singe-user capacity per unit-energy $\frac{\log e}{N_0}$. Thus, there is a sharp transition between orders of growth where interference-free communication is feasible and orders of growth where reliable communication at a positive rate is infeasible. \edit{
	We further show that, if the order of growth is proportional to $n/\log n$, then the capacity per unit-energy is strictly between zero and $\frac{\log e}{N_0}$. Finally, we  show that the capacity per unit-energy can be achieved by an orthogonal-access scheme.}
	
	\item Since an orthogonal-access scheme in combination with orthogonal codebooks is optimal in achieving the capacity per unit-energy for a finite number of users,  we study the performance of such a scheme for an unbounded number of users. In particular,  we characterize in Theorem~\ref{Thm_ortho_code} the largest rate per unit-energy that can be achieved with an orthogonal-access scheme and orthogonal codebooks. Our characterization shows that this scheme is only optimal if the number of users grows more slowly than any positive power of $n$.
	
	\item We then analyze the behaviour of the capacity per unit-energy of the Gaussian MnAC as a function of \edit{the} order of growth of the number of users for the per-user probability of error which, in this paper,  we shall refer to as \emph{average probability of error} (APE). In contrast, we refer to the classical probability of error as \emph{joint probability of error} (JPE). We demonstrate that, if the order of growth  is sublinear, then each user can achieve the capacity per unit-energy $\frac{\log e}{N_0}$ of the single-user Gaussian channel. Conversely, if the growth is linear or above, then the capacity per unit-energy is zero (Theorem~\ref{Thm_capac_APE}).  Comparing with the results in Theorem~\ref{Thm_nonrandom}, we observe that relaxing the error probability from JPE to APE shifts the transition threshold separating the two regimes of interference-free communication and no reliable communication from $n/\log n$ to $n$.
	
	\item We next consider MnACs with random user activity. As before, we consider a setting where the total number of users $\el_n$ may grow as an arbitrary function of the blocklength. Furthermore, the probability $\alpha_n$ that a user is active may be a function of the blocklength, too. Let $k_n = \alpha_n \el_n$ denote the average number of active users. We demonstrate in Theorem~\ref{Thm_random_JPE} that, if $k_n \log \el_n$ is sublinear in $n$, then each user can achieve the single-user capacity per unit-energy. Conversely, if $k_n \log \el_n$ is superlinear in $n$, then the capacity per unit-energy is zero. \edit{We also demonstrate that, if $k_n \log \ell_n$ is linear in $n$, then the capacity per unit-energy is strictly between zero and $\frac{\log e}{N_0}$.}
	
	\item We further show in Theorem~\ref{Thm_ortho_accs} that orthogonal-access schemes, which are optimal when $\alpha_n=1$, are strictly suboptimal when $\alpha_n \to 0$.  In Theorem~\ref{Thm_capac_PUPE}, we then characterize the behaviour of \edit{the} random MnAC under APE.
	
	\item	\edit{We conclude the paper with a comparison of the setting considered in this paper and the setting proposed by Polyanskiy in \cite{Polyanskiy17}. Since an important aspect of the Polyanskiy setting is that the probability of error does not vanish as the blocklength tends to infinity, we briefly discuss the behavior of the $\epsilon$-capacity per unit-energy, i.e., the largest rate per unit-energy for which the error probability does not exceed a given $\epsilon$.} For the case where the users are active with probability one, we show that, for JPE and an unbounded number of users, the $\epsilon$-capacity per unit energy coincides with $\CC$. In other words, the strong converse holds in this case.  In contrast, for APE, the $\epsilon$-capacity per unit-energy can be strictly larger than $\CC$, so the strong converse does not hold.
\end{enumerate}

The rest of the paper is organized as follows. Section~\ref{Sec_model} introduces the system model and the different notions of probability of error.  Section~\ref{Sec_nonrandom} presents our results for
the case where all users are active with probability one (``non-random MnAC"). Section~\ref{sec_random_MnAC} presents our results for the case where the user activity is random (``random MnAC"). 
\edit{
Section~\ref{Sec_discuss}
 briefly discusses our results with that obtained in~\cite{Polyanskiy17}  under non-vanishing probability of error.} Section~\ref{Sec_conclusion} concludes the paper with a summary and discussion of our results.

\section{Problem Formulation and Preliminaries}
\label{Sec_model}
\subsection{Model and Definitions}
\label{Sec_Def}

Consider a network with $\el$ users that, if they are active, wish to transmit their messages $W_i, i=1, \ldots, \el$ to one common receiver, see Fig.~\ref{Fig_many_acc}. The messages are assumed to be independent and uniformly distributed on $ \cM_n^{(i)} \triangleq \{1,\ldots,M_n^{(i)}\}$. To transmit their messages, the users send a codeword of $n$ symbols over the channel, where $n$ is referred to as the \emph{blocklength}. We consider a many-access scenario where the number of users $\el$ may grow with $n$, hence, we denote it as $\el_n$. 
We  assume that a user is active with probability $\alpha_n$. 
\edit{
We denote the average number of active users at blocklength $n$ by $k_n$, i.e., $k_n = \alpha_n \ell_n$.} 

Let $\cU_n$ denote the set of active users at blocklength $n$, defined as
\begin{align*}
\cU_n \triangleq \{\edit{i=1,\ldots, \ell_n} : \mbox{user } i \mbox{ is active} \}.
\end{align*}
We  consider a Gaussian channel model where the received vector $\bY$ is given by
\begin{align}
\bY & = \sum_{i\in \cU_n } \bx_i(W_i) + \bZ. \label{Eq_model}
\end{align}
Here $ \bx_i(W_i)$ is the length-$n$ transmitted codeword from user $i$ for message $W_i$, and $\bZ$ is 
a vector of $n$ i.i.d. Gaussian components $Z_j \sim \cN(0, N_0/2)$ (where $\cN(\mu, \sigma^2)$ denotes the Gaussian distribution with mean $\mu$ and variance $\sigma^2$) independent of $\bX_i\triangleq \bx_i(W_i)$. The decoder produces 
\edit{an estimate of the users that are active and estimates of their transmitted messages.}

\begin{figure}[htb]
	\centering
	\includegraphics[scale=0.5]{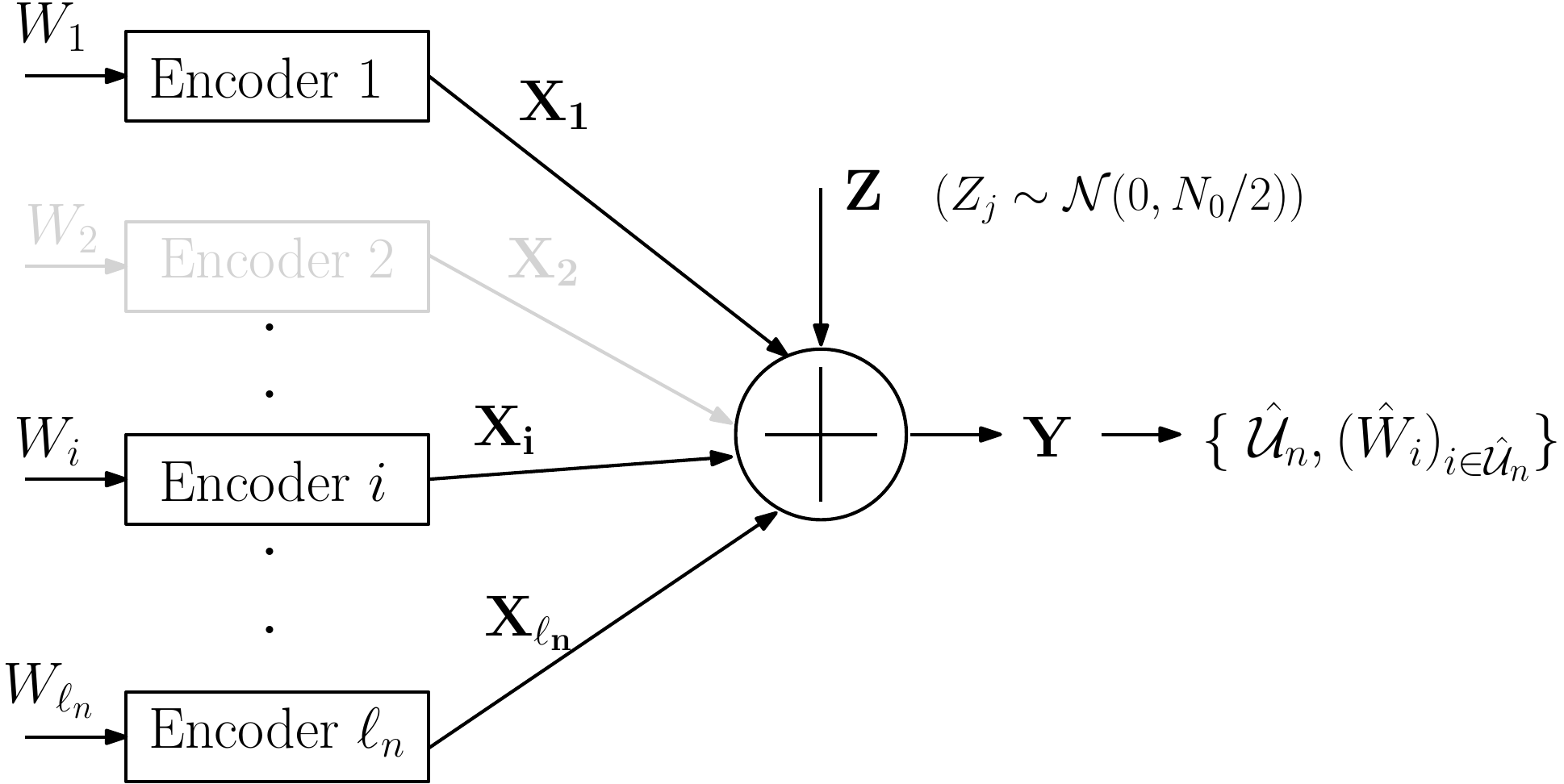}
	\caption{Many-access channel with $\ell_n$ users at blocklength $n$.}
	\label{Fig_many_acc}
\end{figure}

There are different ways to define the overall probability of error.
Let $\hat{\mathcal{U}}_n$ and $\hat{W}_i$ denote the decoder's \edit{estimates} of the set of active users and the message transmitted by user $i$, respectively. Further let $\mathcal{D}_E$ denote the event \edit{that} the set of active users was detected erroneously, i.e., that $\hat{\mathcal{U}}_n \neq \mathcal{U}_n$, and let $\mathcal{M}_E$ denote the event 
that $\hat{W}_i\neq W_i$ for some $i\in\mathcal{U}_n$ (where we set $\hat{W}_i=0$ for every $i\notin \hat{\mathcal{U}}_n$).  One possibility to measure the likelihood of these two events is via the probability of error
\begin{equation}
P_e^{(n)} = \Pr(\mathcal{D}_E \cup\mathcal{M}_E). \label{Eq_prob_err_union}
\end{equation}
Another possibility is to consider
\begin{equation}
P_{\textnormal{max}}^{(n)} = \max(\Pr(\cD_E),\Pr(\cM_E)).
\end{equation}
We have \edit{$P_{\textnormal{max}}^{(n)} \leq P_e^{(n)} \leq 2 P_{\textnormal{max}}^{(n)}$}. Indeed, the left inequality follows since $\Pr(\cD_E\cup \mathcal{M}_E)\geq \Pr(\cD_E ) $ and $\Pr(\cM_E\cup \cM_E)\geq \Pr(\cM_E)$. The right-most inequality follows because, by the union bound, $P_e^{(n)}\leq \Pr(\mathcal{D}_E)$ + $\Pr(\mathcal{M}_E) \leq 2 P_{\textnormal{max}}^{(n)}$. So, in general, $P_{\textnormal{max}}^{(n)}$ is more optimistic than $P_{e}^{(n)}$. However, if we wish the probability of error to vanish as $n\to\infty$, then the two definitions are equivalent since $P_{e}^{(n)}$ vanishes if, and only if, $P_{\textnormal{max}}^{(n)}$ vanishes.

In this paper, we will mainly consider the more pessimistic definition of probability of error $P_e^{(n)}$. For this definition, one can model an inactive user by an active user that transmits message $W_i=0$ and an encoder that maps the zero message to the all-zero codeword. The decoder then simply guesses the transmitted message, and the error probability $P_e^{(n)}$ is given by the probability that the decoder's guess $\hat{W}_i$ is different from $W_i$. Mathematically, this can be described as follows. We enhance the message set to
\begin{align*}
\overline{\cM}_n^{(i)}\triangleq \cM_n^{(i)} \cup \{0\}
\end{align*} 
and define the distribution of the $i$-th user's message as
\begin{align}
\Pr\{W_i = w\}  =
\begin{cases}
1 - \alpha_n,  & \quad w=0 \\
\frac{\alpha_n}{M_n^{(i)}}, & \quad w \in \{1,\ldots,M_n^{(i)}\}.
\end{cases}
\label{Eq_messge_def}
\end{align}
We assume that the codebook is such that message $0$ is mapped to the all-zero codeword. Then, the channel model~\eqref{Eq_model} can be written as 
 \begin{align*}
 \bY & = \sum_{i=1}^{\el_n} \bx_i(W_i) + \bZ. 
 \end{align*}
 We next introduce the notion of an  $(n,\bigl\{M_n^{(\cdot)}\bigr\},\bigl\{E_n^{(\cdot)}\bigr\}, \epsilon)$ code.
 \begin{definition}
 	\label{Def_nMCode}
 	For $0 \leq \epsilon \leq 1$, an  $(n,\bigl\{M_n^{(\cdot)}\bigr\},\bigl\{E_n^{(\cdot)}\bigr\}, \epsilon)$ code for the Gaussian MnAC consists of:
 	\begin{enumerate}
 		\item Encoding functions $f_i: \{0, 1,\ldots,M_n^{(i)}\} \rightarrow \mathbb{R}^n$, \mbox{$i =1,\ldots, \el_n$} which map user $i$'s message to the codeword $\bx_i(W_i)$, satisfying the energy constraint
 		\begin{align}
 		\label{Eq_energy_consrnt}
 		\sum_{j=1}^{n} x_{ij}^2(W_i) \leq E_n^{(i)}, \quad \textnormal{ with probability one}
 		\end{align}
 		where $x_{ij}(W_i)$ is the $j$-th symbol of the transmitted codeword. We set $x_{ij}(0) = 0$, $j=1, \ldots, n$ for all users $i=1,\ldots, \ell_n$.
 		\item Decoding function $g: \mathbb{R}^n \rightarrow \{ 0,1,\ldots,M_n^{(1)}\} \times \ldots \times \{ 0,1,\ldots,M_n^{(\el_n)}\} $ which maps the received vector $\bY$ to the messages of all users and whose probability of error $P_{e}^{(n)}$ satisfies
 	\end{enumerate} 
 	\begin{align}
 	\label{Eq_prob_err}
 	P_{e}^{(n)} \triangleq  \Pr\{ g(\bY) \neq (W_1,\ldots,W_{\el_n}) \} \leq \epsilon.
 	\end{align}
 \end{definition}
 
 \edit{The probability of error in~\eqref{Eq_prob_err} is equal to $P_e^{(n)}$ defined in~\eqref{Eq_prob_err_union}. Indeed, the event $g(\bY) \neq (W_1,\ldots,W_{\el_n})$ occurs if, and only if, there exists at least one index $i=1,\ldots, \ell_n$ for which $\hat{W_i}\neq W_i$. This in turn implies that either event $\cD_E$ occurs (if $W_i=0$) or event $\cM_E$ occurs (if $W_i \neq 0$). Conversely, if the event $\cD_E \cup \cM_E$ occurs, then there exists either some $i \notin \cU_n$ for which $\hat{W}_i \neq 0$ or some $i \in \cU_n$ for which $\hat{W}_i \neq W_i$. Consequently, there exists at least one index $i=1,\ldots, \ell_n$ for which $\hat{W_i}\neq W_i$. It  follows that the events $ g(\bY) \neq (W_1,\ldots,W_{\el_n})$ and $\cD_E \cup \cM_E$ are equivalent.}
 
 \edit{We shall say that the \emph{codebook of user $i$ is orthogonal} if the inner product between $\bx_i(w)$ and $\bx_i(w')$ is zero for every $w\neq w'$, where $w,w'=1,\ldots,M_n^{(i)}$. Similarly, we shall say that an \emph{access scheme is orthogonal} if, for any two users $i$ and $j$, the inner product between $\bx_i(w)$ and $\bx_j(w')$ is zero for every $w=1,\ldots,M_n^{(i)}$ and $w'=1,\ldots, M_n^{(j)}$.}
An $(n,\{M_n^{(\cdot)}\},\{E_n^{(\cdot)}\}, \epsilon)$ code is said to be \emph{symmetric} if $M_n^{(i)} = M_n$ and $E_n^{(i)} = E_n$ for all $i=1, \ldots, \el_n$. For compactness, we denote such a code by $(n, M_n, E_n, \epsilon)$. In this paper, we restrict ourselves to symmetric codes.
\begin{definition}
	\label{Def_Sym_Rate_Cost}
	For a symmetric code, the rate per unit-energy $\CR$ is said to be $\epsilon$-achievable if for every $\delta > 0$ there exists an $n_0$ such that, if $n \geq n_0$, then an $(n,M_n,E_n, \epsilon)$ code can be found whose rate per unit-energy satisfies $\frac{\log M_n}{ E_n} > \CR - \delta$. Furthermore, $\CR$ is said to be achievable if it is $\epsilon$-achievable for all $0 < \epsilon < 1$.  The capacity per unit-energy $\CC$ is the supremum of all achievable rates per unit-energy. The $\epsilon$-capacity per unit-energy $\CCe$ is the supremum of all $\epsilon$-achievable rates per unit-energy.
\end{definition}

\begin{remark}
	\label{remark}
	In \cite[Def.~2]{Verdu90}, a rate per unit-energy $\CR$ is said to be $\epsilon$-achievable if for every $\alpha>0$ there exists an $E_0$ such that, if $E \geq E_0$, then an $(n,M,E,\epsilon)$ code can be found whose rate per unit-energy satisfies $\frac{\log M}{E} > \CR - \alpha$. \edit{Thus, in contrast to Definition~\ref{Def_Sym_Rate_Cost}, the energy $E$ is required to be larger than some threshold, rather than the blocklength $n$.} For the MnAC, where the number of users grows with the blocklength, we believe it is more natural to impose \edit{a threshold on $n$.} Definition~\ref{Def_Sym_Rate_Cost} is also consistent with the definition of energy-per-bit in \cite{Polyanskiy17, ZadikPT19}. Further note that, for the capacity per unit-energy, where a vanishing error probability is required,
	Definition~\ref{Def_Sym_Rate_Cost} is in fact equivalent to 
	\cite[Def.~2]{Verdu90}, since $P_{e}^{(n)}\to 0$ only if $E_n \to \infty$ (see Lemma~\ref{Lem_energy_infty} ahead).
\end{remark}

\begin{remark}
	Many works in the literature on many-access channels, including \cite{Polyanskiy17, OrdentlichP17,ZadikPT19,KowshikPISIT19,KowshikP19,KowshikTCOM20}, consider a \emph{per-user probability of error}
	\begin{equation}
	\label{eq_Pe_A}
	P_{e,A}^{(n)} \triangleq \frac{1}{\el_n} \sum_{i=1}^{\el_n} \textnormal{Pr}\{\hat{W_i} \neq W_i\}
	\end{equation}
	rather than the probability of error in~\eqref{Eq_prob_err}.  In this paper, we shall refer to \eqref{Eq_prob_err} as \emph{joint error probability (JPE)} and to \eqref{eq_Pe_A} as \emph{average error probability (APE)}. While we mainly consider the JPE, we also discuss the capacity per unit-energy for APE. To this end, we define an $(n,\{M_n^{(\cdot)}\},\{E_n^{(\cdot)}\}, \epsilon)$ code for  APE with the same encoding and decoding functions described in Definition~\ref{Def_nMCode}, but with the probability of error \eqref{Eq_prob_err} replaced with~\eqref{eq_Pe_A}. The capacity per unit-energy \edit{and the $\epsilon$-capacity per unit-energy} for APE, denoted by $\CC^A$ \edit{and  $\CCe^A$ respectively, are} then defined as in Definition~\ref{Def_Sym_Rate_Cost}.
\end{remark}

\subsection{Order Notation}
Let $\{a_n\}$ and $\{b_n\}$ be two sequences of nonnegative real numbers.
We write $a_n = O(b_n)$  if $\limsup_{n \to \infty} \frac{a_n}{b_n}	< \infty$. Similarly, we write $a_n = o(b_n)$ if $ \lim_{n\rightarrow \infty} \frac{a_n}{b_n} = 0$, and $a_n = \Omega(b_n)$ if $\liminf\limits_{n \rightarrow \infty} \frac{a_n}{b_n} >0$. 
The notation $a_n = \Theta (b_n)$ indicates that $a_n = O(b_n)$ and $a_n = \Omega(b_n)$.
Finally, we write
$a_n = \omega (b_n)$ if $\lim\limits_{n\rightarrow \infty}  \frac{a_n}{b_n} = \infty$.

 \section{Capacity per Unit-Energy of Non-Random Many-Access Channels}
 \label{Sec_nonrandom}

In this section, we study the Gaussian MnAC \edit{under the assumption that all users are active, i.e., $\alpha_n =1$ and $k_n = \el_n$.} 
We present our results on the non-random MnACs in Subsection~\ref{Sec_feasble_infeasble}. In particular, in Theorem~\ref{Thm_nonrandom}, we  characterize the capacity per unit-energy as a function of the order of $k_n$. Then,  in Theorem~\ref{Thm_ortho_code}, we give a characterization of  the largest rate per unit-energy achievable using an orthogonal-access scheme with orthogonal codebooks. Finally, in Theorem~\ref{Thm_capac_APE}, we present our results on the capacity per unit-energy for APE. The proofs of Theorems~\ref{Thm_nonrandom}--\ref{Thm_capac_APE} are given in Subsections~\ref{Sec_proof_Thm1}, \ref{sec_ortho}, and~\ref{sec_APE}, respectively.

\subsection{Main Results}
\label{Sec_feasble_infeasble}

\begin{theorem}
	\label{Thm_nonrandom}
	The capacity per unit-energy of the \edit{Gaussian} non-random \edit{MnAC} has the following behaviour:
	\begin{enumerate}
		\item 	If $k_n = o(n/\log n)$, then any rate per unit-energy satisfying $\CR < \frac{\log e}{N_0}$ is achievable. Moreover, this rate can be achieved by an orthogonal-access scheme.	\label{Thm_Infeasble_achv}
		\item 	 If $k_n =\omega(n / \log n)$, then $\CC =0$. In words, if the order of $k_n$ is strictly above $n/\log n$, then no coding scheme achieves a positive rate per unit-energy. 	\label{Thm_Infeasble_convrs}
		\edit{\item 	If $k_n = \Theta( \frac{n}{\log n})$, then $0< \CC < \frac{\log e}{N_0} $. In words, if the order of $k_n$ is exactly $n/\log n$, then a positive rate per unit-energy, but strictly less than $\frac{\log e}{N_0}$ is achievable.
		\label{Thm_exact_order}}
	\end{enumerate}
\end{theorem}
\begin{proof}
See Subsection~\ref{Sec_proof_Thm1}.	
\end{proof}

Theorem~\ref{Thm_nonrandom} demonstrates that there is a sharp transition between orders of growth of $k_n$ where
each user can achieve the single-user capacity per unit-energy $\frac{\log e}{N_0}$, i.e., \edit{where} users can communicate as if free of interference, and orders of growth where no positive rate per unit-energy is feasible. The transition threshold separating these two regimes is at the order of growth $n/ \log n$.
 The capacity per unit-energy can be achieved using an orthogonal-access scheme where each user is assigned an exclusive time slot.
 As we shall show in Section~\ref{sec_random_MnAC}, such an access scheme is wasteful in terms of resources and strictly suboptimal when users are active only sporadically.
 \edit{The theorem also demonstrates that, when the order of growth of $k_n$ is exactly equal to $n/\log n$, the rate per unit-energy is strictly positive, but also strictly less than $\frac{\log e}{N_0}$.}

As mentioned in the introduction, when the number of users is finite, all users can achieve the single-user capacity per unit-energy $\frac{\log e}{N_0}$ by an orthogonal-access scheme where each user uses an orthogonal codebook. In the following theorem, we show that this is not necessarily the case \edit{anymore} when the number of users grows with the blocklength.

\begin{theorem} 
	\label{Thm_ortho_code}
	The largest rate per unit-energy $\CCPP$ achievable with an orthogonal-access scheme and orthogonal codebooks has the following behaviour: 
	\begin{enumerate}[1)]
		\item 	 If $k_n = o(n^{c})$ for every $c>0$, then  $\CCPP = \frac{\log e}{N_0}$. \label{Thm_ortho_part1}
		\item If $k_n=\Theta\left({n^c}\right)$, then 
		\begin{equation*}
		\CCPP = \begin{cases} \frac{\log e}{N_0}  \frac{1}{\left(1+\sqrt{\frac{c}{1-c}}\right)^2}, \quad & \textnormal{if $0<c\leq 1/2$} \\ \frac{\log e}{2 N_0} (1-c), \quad & \textnormal{if $1/2<c<1$}.\end{cases}
		\end{equation*}
		\label{Thm_ortho_part2}
	\end{enumerate}
\end{theorem}
\begin{proof}
	See Subsection~\ref{sec_ortho}.	
\end{proof}

Theorem~\ref{Thm_ortho_code} shows that an orthogonal-access scheme in combination with orthogonal codebooks is  optimal  only if $k_n$ grows more slowly than any positive power of $n$. Part~\ref{Thm_ortho_part2}) of Theorem~\ref{Thm_ortho_code} gives the largest rate per unit-energy achievable when the order of $k_n$ is a positive power of $n$. 

\edit{
\begin{remark}
	\label{Rem_orth_code}
	Observe that the behavior of $\CCPP$ as a function of $c$ can be divided into two regimes: if $1/2<c<1$, then $\CCPP$ decays linearly in $c$; if $0<c\leq 1/2$, then the dependence of $\CCPP$ on $c$ is nonlinear. This is a consequence of the behavior of the error exponent achievable with orthogonal codebooks. More specifically, Theorem~\ref{Thm_ortho_code} follows from lower and upper bounds on the probability of error that become asymptotically tight as $E\to\infty$; see Lemma~\ref{Lem_ortho_code}. The lower bound follows from the sphere-packing bound \cite{ShannonGB67}. The upper bound is obtained by applying Gallager's $\rho$-trick to improve upon the union bound \cite[Sec.~2.5]{ViterbiO79}, followed by an optimization over the parameter $0\leq\rho\leq 1$. When the rate per unit-energy is smaller than $\frac{1}{4} \frac{\log e}{N_0}$, the optimal value of $\rho$ is $1$, and the exponent of the upper bound depends linearly on the rate per unit-energy. For rates per unit-energy above $\frac{1}{4} \frac{\log e}{N_0}$, the optimal value of $\rho$ depends on the rate per unit-energy, which results in a nonlinear dependence of the exponent on the rate per unit-energy. This behavior of the error exponent as a function of the rate per unit-energy translates to the two regimes of $\CCPP$ observed in Theorem~\ref{Thm_ortho_code}.
\end{remark}
}

 Next we discuss the behaviour of the capacity per unit-energy for APE. We
 show that, if the order of growth of $k_n$ is sublinear,
 then each user can achieve the single-user capacity per unit-energy $\frac{\log e}{N_0}$. Conversely, if the growth of $k_n$ is linear or above, then the capacity per unit-energy is zero. We have the following theorem.

\begin{theorem}
	\label{Thm_capac_APE}
	The capacity per unit-energy $\CC^A$ for APE has the following behavior:
	\begin{enumerate}
		\item  If  $k_n  = o(n)$, then $\CC^A = \frac{\log e}{N_0}$. Furthermore, the capacity per unit-energy can be achieved by an orthogonal-access scheme where each user uses an   orthogonal codebook.	\label{Thm_APE_achv_part}
		\item If $k_n  = \Omega(n)$, then $\CC^A =0$.	\label{Thm_APE_conv_part}
	\end{enumerate}
\end{theorem}
\begin{proof}
	See Subsection~\ref{sec_APE}.	
\end{proof}

Theorem~\ref{Thm_capac_APE} demonstrates that under APE the capacity per unit-energy has a similar behaviour as under JPE. Again, there is a sharp transition between orders of growth of $k_n$ where interference-free communication is possible and orders of growth where no positive rate per unit-energy is feasible. The main difference is that the transition threshold is shifted from $n/\log n$ to $n$. Such an improvement on the order of growth is possible because, for the probability of error to vanish as $n\to \infty$, the energy $E_n$ needs to satisfy different  necessary constraints under JPE and APE. Indeed, we show in the proof of Theorem~\ref{Thm_nonrandom} that 
the JPE vanishes only if the energy $E_n$ scales logarithmically in the  number of users (Lemma~\ref{Lem_convrs_err_prob}), and a positive rate per unit-energy is feasible only if the total power $k_n E_n/n$ is bounded in $n$. No sequence $\{E_n\}$ can satisfy both these conditions if $k_n = \omega(n/\log n)$. In contrast, for the APE to vanish asymptotically,  the energy $E_n$ does not need to grow logarithmically in the number of users, it suffices that it tends to infinity as $n \to \infty$. We can then find sequences $\{E_n\}$ that tend to infinity and for which $k_nE_n/n$ is bounded if, and only if, $k_n$ is sublinear in $n$. Also note that, for APE, an orthogonal-access scheme with orthogonal codebooks is optimal for all orders of $k_n$, whereas for JPE it is only optimal if the order of $k_n$ is not a positive power of $n$.

\subsection{Proof of Theorem~\ref{Thm_nonrandom}}
\label{Sec_proof_Thm1}

We first give an outline of the proof of Theorem~\ref{Thm_nonrandom}. To prove Part~\ref{Thm_Infeasble_achv}), we use an orthogonal-access scheme where the total number of channel uses is divided equally among all the users. Each user
uses the same single-user code in the assigned channel uses. The receiver decodes the message of each user separately, which is possible because the access scheme is orthogonal. We \edit{next} express the overall probability of error in terms of the number of users $k_n$ and the
probability of error achieved by the single-user code in an AWGN channel, which we then show vanishes as $n \to \infty$ if $k_n = o(n/\log n)$. The proof of Part~\ref{Thm_Infeasble_convrs}) hinges mainly on two facts. The first one is that the probability of error vanishes only if the energy $E_n$ scales at least
\edit{logarithmically} in the number of users, i.e., $E_n = \Omega(\log k_n)$.\footnote{A similar bound was presented in \cite[p.~82]{Polyanskiy18} for the case where $M_n=2$.} The second one is that we have $\CR>0$ only if the total power $k_nE_n/n$ is bounded as $n \to \infty$, which
is a direct consequence of Fano's inequality. If $k_n = \omega(n/\log n)$, then there is no sequence $\{E_n\}$ that simultaneously satisfies these two conditions. \edit{Part~\ref{Thm_exact_order}) follows by revisiting the proofs of Parts~\ref{Thm_Infeasble_achv}) and \ref{Thm_Infeasble_convrs}) for the case where $k_n=\Theta(n/\log n)$.}

\subsubsection{Proof of Part~\ref{Thm_Infeasble_achv})} The achievability uses  an orthogonal-access scheme \edit{where, in each time step,} only one user transmits, \edit{the} other users remain silent. We first note that the probability of correct decoding of any orthogonal-access scheme is given by 
	\begin{align*}
	P_c^{(n)} = \prod_{i=1}^{k_n}\left(1-P_{e,i}\right)
	\end{align*}
	where $P_{e,i} = \Pr(\hat{W}_i \neq W_i)$ denotes the probability of error in decoding  user $i$'s message. 
	In addition, if each user follows the same coding scheme, then the probability of correct decoding is given by 
	\begin{align}
	P_c^{(n)} & = \left(1-P_{e,1}\right)^{k_n}. \label{Eq_ortho_prob_err}
	\end{align}
	
		For a Gaussian point-to-point channel with \edit{ blocklength $N$ and  power constraint $P$, i.e., $\frac{E_N}{N} \leq P$,} there exists an encoding and decoding scheme whose
	average probability of error is upper-bounded by
	\begin{align}
	P(\cE) & \leq  M_N^{ \rho}  \exp[-NE_0(\rho, P)],  \; \mbox{ for every } 0< \rho \leq 1 \label{Eq_upp_prob_AWGN} 
	\end{align}
	where 
	\begin{align}
	E_0(\rho, P) & \triangleq \frac{\rho}{2} \ln \left(1+\frac{2P}{(1+\rho)N_0}\right). \notag
	\end{align} 
	This bound is due to Gallager and can be found in~\cite[Sec.~7.4]{Gallager68}.
	
	Now let us consider an orthogonal-access scheme in which each user gets $n/k_n$ channel uses, and we timeshare between users. Each user follows
	the coding scheme that achieves~\eqref{Eq_upp_prob_AWGN} with  power constraint $P_n = \frac{E_n}{n/k_n}$. Note that this coding scheme satisfies also the energy constraint~\eqref{Eq_energy_consrnt}. 
	\edit{
	Then, we obtain the following upper bound $P_{e,1}$ for a fixed rate per unit-energy $\CR = \frac{\log M_n}{E_n}$, 	by substituting in~\eqref{Eq_upp_prob_AWGN}
	$N$ by $n/k_n$ and  $P$ by $P_n = \frac{E_n}{n/k_n}$:}
	\begin{align}
	P_{e,1} & \leq  M_n^{ \rho}  \exp\left[-\frac{ n}{k_n}E_0(\rho, P_n)\right] \nonumber \\
	& = \exp\left[   \rho  \ln M_n -   \frac{ n}{k_n} \frac{\rho}{2} \ln \left(1+\frac{ 2E_nk_n/n}{(1+\rho)N_0}\right) \right] \nonumber \\
	& = \exp\left[ -E_n \rho   \left(  \frac{\ln (1+\frac{ 2E_nk_n/n}{(1+\rho)N_0})}{2E_nk_n/n } -\frac{\CR}{ \log e} \right)\right]. \label{Eq_err_uppr1}
	\end{align}
	Combining \eqref{Eq_err_uppr1} with \eqref{Eq_ortho_prob_err}, we obtain that the probability of correct decoding can be lower-bounded as 
	\begin{align}
	&  1 -P_e^{(n)}   \geq \Biggl(1 -   \exp\Biggl[ -E_n \rho   \Biggl(  \frac{\ln (1+\frac{ 2E_nk_n/n}{(1+\rho)N_0})}{2E_nk_n/n } -\frac{\CR}{ \log e} \Biggr)\Biggr]\Biggr)^{k_n}. \label{Eq_ortho_prob_lower}
	\end{align}
	We next choose $E_n = c_n \ln n$ with $c_n \triangleq \ln\bigl(\frac{n}{k_n\ln n}\bigr)$. Since, by assumption, $k_n = o(n / \log n)$, this implies that $\frac{k_nE_n}{n} \to 0$ as $n \to \infty$. Consequently, the first term in the inner-most bracket in \eqref{Eq_ortho_prob_lower} tends to $1/((1+\rho)N_0)$ as $n \to \infty$. It follows that for $\CR < \frac{\log e}{N_0}$, there exists a sufficiently large $n_0$, a $0<\rho \leq 1$, and a $\delta>0$ such that, for all $n\geq n_0$, the right-hand side (RHS) of \eqref{Eq_ortho_prob_lower} is lower-bounded by $\left(1-\exp[-E_n \rho \delta]\right)^{k_n}$. Since $c_n\delta \rho \to \infty$ as $n\to\infty$, we have
	\begin{align}
	\left(1-\exp[-E_n \rho \delta]\right)^{k_n}& \geq \left(1-\frac{1}{n^{2}}\right)^{k_n} \notag \\
	& \geq	\left(1-\frac{1}{n^{2}}\right)^{\frac{n}{\log n}} \notag \\
	& =  \left[\left(1 - \frac{1}{n^{2}}\right)^{n^{2}}\right]^{\frac{1}{n\log n}} \label{Eq_prob_corrct}
	\end{align}
	for $n \geq n_0$ and sufficiently large $n_0$, such that $c_n\delta \rho\geq2$ and  $k_n \leq \frac{n}{\log n}$. Noting that $(1 - \frac{1}{n^{2}})^{n^{2}} \to 1/e$ and $\frac{1}{n\log n} \to 0$ as $n\to\infty$, we obtain that the	RHS of~\eqref{Eq_prob_corrct} tends to one as $n \to \infty$. This implies that, if $k_n = o(n/\log n)$, then any rate per unit-energy  $\CR < \frac{\log e}{ N_0} $ is achievable.
	
	\subsubsection{Proof of Part~\ref{Thm_Infeasble_convrs})}
	Let $\bW$ and $\bhW$ denote the vectors $(W_1,\ldots, W_{k_n})$ and $(\hat{W_1},\ldots,\hat{W}_{k_n})$, respectively. Then 
	\begin{align}
	k_n \log M_n& = H(\bW) \nonumber\\
	& = H(\bW|\bhW)+I(\bW;\bhW)\nonumber\\
	& \leq 1+P_e^{(n)}k_n \log M_n + I(\bX;\bY)  \nonumber 
	\end{align}
	by Fano's inequality and the data processing inequality. By following~\cite[Sec.~15.3]{CoverJ06},
	it can be shown that 
	\mbox{$I(\bX;\bY) \leq  \frac{n}{2} \log \left(1+\frac{2 k_nE_n}{nN_0}\right)$}. Consequently,
	\begin{equation}
	 \frac{\log M_n }{E_n} \leq \frac{1}{k_nE_n}+ \frac{ P_e^{(n)} \log M_n}{E_n}  + \frac{n}{2 k_nE_n} \log \left(\!1+\frac{ 2k_nE_n}{nN_0}\!\right)\!. \notag
	\end{equation}
	This implies that the rate per unit-energy $\CR=(\log M_n)/E_n$ is upper-bounded by 
	\begin{align}
	\CR\leq \frac{ \frac{1}{k_nE_n} + \frac{n}{2 k_nE_n}\log(1+\frac{ 2k_nE_n}{nN_0})}{1 -P_e^{(n)}}.\label{Eq_R_avg1}
	\end{align}
	
	We next show by contradiction that, if $k_n =\omega(n / \log n)$, then $P_e^{(n)} \to 0$ as $n \to \infty$ only if $\CC=0$. Thus, assume that $k_n =\omega(n / \log n)$ and that there exists a code with rate per unit-energy $\CR >0$ such that $P_e^{(n)} \to 0$ as $n \to \infty$.
	To prove that there is a contradiction we need the following lemma. 
	\begin{lemma}
		\label{Lem_energy_infty}
		If $M_n \geq 2$, then  $P_{e}^{(n)}  \to 0$ only if $E_n \to \infty$.
	\end{lemma}
\begin{proof}
	See Appendix~\ref{Sec_energy_infty}.
\end{proof}

By the assumption $\CR > 0$, we have that $M_n \geq 2$.
Since we further assumed that 
 $P_{e}^{(n)}  \to 0$, Lemma~\ref{Lem_energy_infty} implies that $E_n \to \infty$. 
 Together with \eqref{Eq_R_avg1},  this in turn implies that  $\CR > 0$ is only possible if
 $k_nE_n/n$ is bounded in $n$. Thus,
	\begin{align}
	E_n = O(n/k_n). \label{Eq_energy_bnd}
	\end{align}
The next lemma presents another necessary condition on the order of $E_n$ which contradicts \eqref{Eq_energy_bnd}.
	\begin{lemma}
		\label{Lem_convrs_err_prob}
		 If $\CR > 0$ and $k_n\geq 5$, then $P_{e}^{(n)} \to 0$ only if $E_n = \Omega(\log k_n)$.
	\end{lemma}
\begin{proof}
	This lemma is a special case of Lemma~\ref{Lem_energy_bound} stated in the proof of Theorem~\ref{Thm_random_JPE} in Section~\ref{sec_random_MnAC} and proven in  Appendix~\ref{Append_prob_lemma}.
\end{proof}

We finish the proof by showing that, if $k_n =\omega(n / \log n)$, then there exists no sequence $\{E_n\}$ of order $\Omega(\log k_n)$ that satisfies \eqref{Eq_energy_bnd}. Indeed, $E_n = \Omega(\log k_n)$ and $k_n =\omega(n / \log n)$ imply that
\begin{align}
E_n =\Omega(\log n) \label{Eq_energy_bnd1}
\end{align}
because the order of $E_n$ is lower-bounded by the order of  $\log n - \log \log n$, and  $\log n - \log \log n = \Theta(\log n)$. Furthermore, \eqref{Eq_energy_bnd} and $k_n =\omega(n / \log n)$ imply that 
\begin{align}
E_n & = o(\log n). \label{Eq_energy_bnd2}
\end{align}
Since no sequence $\{E_n\}$ can simultaneously satisfy~\eqref{Eq_energy_bnd1} and~\eqref{Eq_energy_bnd2}, this contradicts the assumption that there exists a code with a positive rate per unit-energy such that the probability of error vanishes as $n$ tends to infinity. Consequently,
 if  $k_n = \omega(n/\log n)$, then no positive rate per unit-energy is achievable. This proves Part~\ref{Thm_Infeasble_convrs}) of Theorem~\ref{Thm_nonrandom}.
 
 \edit{
 	\subsubsection{Proof of Part~\ref{Thm_exact_order})}
 	\label{Sec_prop_exact}
 	
 	To show that $\dot{C}>0$, we use the same orthogonal-access scheme as in the proof of Part~\ref{Thm_Infeasble_achv}) of Theorem~\ref{Thm_nonrandom}. Thus, each user is assigned $n/k_n$ channel uses, and only one user transmits at a time. We further assume that each user uses energy $E_n=c\log n$, where $c$ is some positive constant to be determined later. By the assumption $k_n = \Theta(n/\log n)$, there exist $n_0>0$ and $0 < a_1 \leq a_2$ such that, for all $n \geq n_0$, we have
	\begin{equation}
	\label{eq:kn_Theta(n/logn)}
	a_1 \frac{n}{\log n} \leq k_n\leq a_2 \frac{n}{\log n}. 
	\end{equation}
	The probability of error in decoding the first user's message is then given by \eqref{Eq_err_uppr1}, namely,
 	\begin{align}
 	P_{e,1} & \leq \exp\left[ -E_n \rho   \left(  \frac{\ln \bigl(1+\frac{ 2E_nk_n/n}{(1+\rho)N_0}\bigr)}{2E_nk_n/n } -\frac{\CR}{ \log e} \right)\right] \nonumber\\
 	& \leq  \exp\left[ -c \log n \, \rho  \left(  \frac{\ln \bigl(1+\frac{ 2a_2 c}{(1+\rho)N_0}\bigr)}{2 a_2 c} -\frac{\CR}{ \log e} \right)\right], \quad \textnormal{for every $0<\rho \leq 1, \; n \geq n_0$}\label{eq:blabla2_t1}
 	\end{align}
 	where the last inequality follows since $k_nE_n/n \leq a_2c$ for $n \geq n_0$.
 	
 	We next set
 	\begin{equation*}
 	\dot{R} = \frac{\log e}{2}  \frac{\ln \bigl(1+\frac{ 2a_2 c}{(1+\rho)N_0}\bigr)}{2a_2 c}
 	\end{equation*}
 	which is clearly positive for fixed $a_2$, $c$, and $\rho$. The upper bound \eqref{eq:blabla2_t1} then becomes
 	\begin{equation}
 	\label{eq:blabla3_t1}
 	P_{e,1} \leq \exp\left[ -\log n \, \frac{\rho}{2}  c\frac{\ln \bigl(1+\frac{ 2a_2 c}{(1+\rho)N_0}\bigr)}{2a_2 c}\right], \quad n \geq n_0.
 	\end{equation}
 	For every fixed $a_2$ and $\rho$, the term
 	\begin{equation*}
 	\frac{\rho}{2}  c\frac{\ln \bigl(1+\frac{ 2a_2c}{(1+\rho)N_0}\bigr)}{2a_2 c} = \frac{\rho}{2}\frac{\ln \bigl(1+\frac{ 2a_2 c}{(1+\rho)N_0}\bigr)}{2a_2}
 	\end{equation*}
 	is a continuous, monotonically increasing, function of $c$ that is independent of $n$ and ranges from zero to infinity. We can therefore find a $c $ such that \eqref{eq:blabla3_t1} simplifies to
 	\begin{equation*}
 	P_{e,1} \leq \exp[-\ln n] = \frac{1}{n}, \quad n \geq n_0.
 	\end{equation*}
 	
 	The above scheme has a positive rate per unit-energy. It remains to show that this rate per unit-energy is also achievable, i.e., that the overall probability of correct decoding tends to one as $n\to\infty$. To this end, we use \eqref{Eq_ortho_prob_err} to obtain that
 	\begin{align}
 	1- P_e^{(n)} &= (1-P_{e,1})^{k_n} \notag\\
 	 &\geq \left(1-\frac{1}{n}\right)^{a_2n/\log n}, \quad n \geq n_0.  	\label{eq:blabla4_t1}
 	\end{align}
 	Since $(1-\frac{1}{n})^n\to 1/e$ and $\frac{a_2}{\log n}\to 0$ as $n\to\infty$, the RHS of \eqref{eq:blabla4_t1} tends to one as $n\to\infty$, hence so does the probability of correct decoding. 
 	
 	We next show that $\CC <\frac{\log e}{N_0}$. Lemma~\ref{Lem_convrs_err_prob} implies that,  if $k_n = \Theta(n/\log n)$, then 
 	$P_{e}^{(n)} $ vanishes only if $E_n = \Omega(\log n)$. Furthermore, if $E_n = \omega(\log n)$, then it follows from~\eqref{Eq_R_avg1} that $\CC =0$ since, in this case, $k_nE_n/n$ tends to infinity as $n \to \infty$. Without loss of generality, we can thus assume that $E_n$ must satisfy $E_n = \Theta(\log n)$. Thus, there exist $n'_0>0$ and $0<l_1 \leq l_2$ such that, for all $n \geq n'_0$, we have $l_1\log n \leq E_n \leq l_2 \log n$. Together with \eqref{eq:kn_Theta(n/logn)}, this implies that $\frac{k_n E_n}{n} \geq a_1 l_l$ for all $n \geq \max(n_0,n'_0)$. The claim that $\dot{C}<\frac{\log e}{N_0}$ follows then directly from \eqref{Eq_R_avg1}. Indeed, using that $\frac{\log (1+x)}{x}< \log e$ for every $x>0$, we obtain that
 	\begin{equation}
 	\label{eq:blabla_t1}
 	\frac{n}{2 k_nE_n} \log \left(1+\frac{ 2k_nE_n}{nN_0}\right) \leq \frac{1}{2 a_1 l_1}\log\left(1+\frac{2 a_1 l_1}{N_0}\right)< \frac{\log e}{N_0}, \quad n \geq \max(n_0,n'_0).
 	\end{equation}
 	By \eqref{Eq_R_avg1}, in the limit as $P_e^{(n)}\to 0$ and $E_n\to \infty$, the rate per unit-energy is upper-bounded by \eqref{eq:blabla_t1}. It thus follows that $\CC < \frac{\log e}{N_0}$, which concludes the proof of Part~\ref{Thm_exact_order}) of Theorem~\ref{Thm_nonrandom}.
 }

\subsection{Proof of Theorem~\ref{Thm_ortho_code}}
\label{sec_ortho}
The proof of Theorem~\ref{Thm_ortho_code} is based on the following lemma, which presents bounds on the probability of error achievable over a Gaussian point-to-point channel with an orthogonal codebook.
\begin{lemma}
	\label{Lem_ortho_code}
	The probability of error $P_{e,1} = \Pr(\hat{W} \neq W )$ achievable over a Gaussian point-to-point channel with an orthogonal codebook with $M$ codewords and energy less than or equal to $E$ satisfies the following bounds:
	\begin{enumerate}
		\item For  $0 < \CR \leq \frac{1}{4} \frac{\log e}{N_0}$,
		\begin{align}
		&  \exp\left[- \frac{\ln M}{\CR}\left(\frac{\log e}{2 N_0} - \CR + \beta_E \right) \right] \leq P_{e,1} \leq  \exp\left[- \frac{\ln M}{\CR}\left(\frac{\log e}{2 N_0} - \CR \right) \right]. \label{Eq_orth_sinlg_uppr1}
		\end{align}
		\item For $\frac{1}{4} \frac{\log e}{N_0}  \leq \CR \leq \frac{\log e}{N_0}$,
		\begin{align}
		&  \exp\left[- \frac{\ln M}{\CR}\left(\left(\sqrt{\frac{\log e}{N_0}}- \sqrt{\CR }\right)^2 +  \beta'_E \right) \right] \leq P_{e,1}  \leq  \exp\left[- \frac{\ln M}{\CR}  \left(\sqrt{\frac{\log e}{N_0}} - \sqrt{ \CR}\right)^2  \right]. \label{Eq_orth_sinlg_uppr2}
		\end{align}
	\end{enumerate}
\edit{
	In~\eqref{Eq_orth_sinlg_uppr1} and~\eqref{Eq_orth_sinlg_uppr2}, $\beta_E$ and $\beta'_E$ are some constants of order   $O(\frac{1}{\sqrt{E}})$.}
\end{lemma}
\begin{proof}
	\edit{ The upper bounds in \eqref{Eq_orth_sinlg_uppr1} and \eqref{Eq_orth_sinlg_uppr2} are obtained by upper-bounding the probability of error using Gallager's $\rho$-trick to improve upon the union bound \cite[Sec.~2.5]{ViterbiO79}, followed by a maximization over $\rho$. For $0 < \CR \leq \frac{1}{4} \frac{\log e}{N_0}$, the optimal $\rho$ is equal to $1$; for $\frac{1}{4} \frac{\log e}{N_0}  \leq \CR \leq \frac{\log e}{N_0}$, the optimal $\rho$ is a function of $\CR$. Hence,  the upper bounds in \eqref{Eq_orth_sinlg_uppr1} and~\eqref{Eq_orth_sinlg_uppr2} have different dependencies on $\CR$. The lower bounds in \eqref{Eq_orth_sinlg_uppr1} and \eqref{Eq_orth_sinlg_uppr2} follow from the sphere-packing bound by Shannon, Gallager, and Berlekamp \cite{ShannonGB67}. However, their approach to improve the sphere-packing bound at low rates by writing codewords as concatenations of subcodewords and lower-bounding the error exponent by the convex combination of the error exponents of these subcodewords does not directly apply to our setting where $\log M/E$ is held fixed and $E\to\infty$ (rather than $\log M/n$ is held fixed and $n\to\infty$). The reason is that, for some orthogonal codebooks, the energy of one of the subcodebooks is always zero, resulting in a trivial case where the Shannon-Gallager-Berlekamp approach cannot improve upon the original sphere-packing bound. To sidestep this problem, we lower-bound the probability of error by first rotating the orthogonal codebook in such a way that the energy of each subcodeword is proportional to its blocklength, after which the Shannon-Gallager-Berlekamp approach can be applied. For a full proof of Lemma~\ref{Lem_ortho_code}, see Appendix~\ref{Sec_AWGN_ortho_code}.}
\end{proof}


Next, we define
\begin{align}
a \triangleq \left\{
\begin{array}{cl}
\frac{\left(\frac{\log e}{2 N_0} - \CR  \right) }{\CR}, & \quad \mbox{if  } 0 < \CR \leq \frac{1}{4} \frac{\log e}{N_0} \\
\frac{ \left(\sqrt{\frac{\log e}{N_0}} - \sqrt{ \CR}\right)^2  }{\CR},  & \quad \mbox{if } \frac{1}{4} \frac{\log e}{N_0}  \leq \CR \leq \frac{\log e}{N_0}
\end{array}\right. \label{Eq_def_a}
\end{align}
and let $a_E \triangleq a + \max \{ \beta_E, \beta'_E\} $.
Then, the bounds in Lemma~\ref{Lem_ortho_code} can be written as 
\begin{align}
1/M^{a_E} & \leq P_{e,1} \leq 1/M^{a}. \label{Eq_prob_err_singl}
\end{align}

Now let us consider the case where the users apply an orthogonal-access scheme together with orthogonal codebooks. 
For such a scheme, the collection of codewords from all users is orthogonal, hence there are at most $n$ codewords of length $n$. Since with a symmetric code, each user transmits the same number of messages, it follows that each user transmits $M_n=n/k_n$  messages with 
codewords of energy less than or equal to $E_n$.
In this case, we obtain from \eqref{Eq_ortho_prob_err} and \eqref{Eq_prob_err_singl} that 
\begin{align*}
\left(1-\left(\frac{k_n}{n}\right)^{a}\right)^{k_n} \leq \left(1- P_{e,1}\right)^{k_n} \leq \left(1-\left(\frac{k_n}{n}\right)^{a_{E_n}}\right)^{k_n}
\end{align*}
which, denoting $a_n \triangleq a_{E_n}$, can be written as 
\begin{align}
& \left[\left(1-\left(\frac{k_n}{n}\right)^a\right)^{(\frac{n}{k_n})^a}\right]^{\frac{k_n^{1+a}}{n^a}}\leq \left(1- P_{e,1}\right)^{k_n}
\leq \left[\left(1-\left(\frac{k_n}{n}\right)^{a_n}\right)^{(\frac{n}{k_n})^{a_n}}\right]^ {\frac{k_n^{1+a_n}}{n^{a_n}}}. \label{Eq_ortho_code_upp_low}
\end{align}
Since  Theorem~\ref{Thm_ortho_code} 
only concerns a sublinear number of users, we have
\begin{align*}
\lim\limits_{n\to \infty} \left(1-\left(\frac{k_n}{n}\right)^a\right)^{(\frac{n}{k_n})^a} 
& =  \frac{1}{e}.
\end{align*}
Furthermore, if $P_e^{(n)} \to 0$ then, by Lemma~\ref{Lem_energy_infty}, $E_n \to \infty$ as $n \to \infty$. In this case, $a_n$ converges to the finite value $a$ as $n \to \infty$, and we obtain 
\begin{align*}
\lim\limits_{n \to \infty}  \left(1-\left(\frac{k_n}{n}\right)^{a_n}\right)^{(\frac{n}{k_n})^{a_n}} & =  \frac{1}{e}.
\end{align*}
So \eqref{Eq_ortho_code_upp_low} implies that 
$P_e^{(n)} \to 0$ as $n \to \infty$ if
\begin{align}
\lim\limits_{n \to \infty}	{\frac{k_n^{1+a}}{n^a}} = 0\label{Eq_ordr_lowr}
\end{align}
and only if 
\begin{align}
\lim\limits_{n \to \infty}	 {\frac{k_n^{1+a_n}}{n^{a_n}}} = 0.  \label{Eq_ordr_uppr}
\end{align}
We next use these observations to prove Parts~\ref{Thm_ortho_part1}) and \ref{Thm_ortho_part2}) of Theorem~\ref{Thm_ortho_code}. We begin with Part~\ref{Thm_ortho_part1}). Let $\CR < \frac{\log e}{N_0}$. Thus, we have $a>0$ which implies that we can find a constant $\eta < a/(1+a)$ such that $n^{\eta (1+a)}/n^a \to 0$ as $n \to \infty$. Since, by assumption, $k_n =o(n^c)$ for every $c> 0$, it follows that there exists an $n_0$ such that, for all $n\geq n_0$, we have $k_n \leq n ^{\eta(1+a)}$. This implies that  \eqref{Eq_ordr_lowr} is satisfied, from which Part~\ref{Thm_ortho_part1}) follows.

We next prove Part~\ref{Thm_ortho_part2}) of Theorem~\ref{Thm_ortho_code}.
Indeed,	if $k_n=\Theta\left({n^c}\right)$,  $0<c<1$, then there exist $0<l_1\leq l_2$ and $n_0$ such that, for all $n\geq n_0$, we have $(l_1n)^c\leq k_n \leq (l_2n)^c$. Consequently,
\begin{align}
{\frac{(l_1 n)^{c(1+a_n)}}{n^{a_n}}} \leq    {\frac{k_n^{1+a_n}}{n^{a_n}}}
\leq {\frac{(l_2 n)^{c(1+a_n)}}{n^{a_n}}}. \label{Eq_ortho_err_uppr}
\end{align}
If $P_e^{(n)}  \to 0$, then from~\eqref{Eq_ordr_uppr} we have ${\frac{k_n^{1+a_n}}{n^{a_n}}}\to 0$. Thus, \eqref{Eq_ortho_err_uppr} implies that  $c(1+a_n) - a_n$ converges to a negative value.
Since $c(1+a_n) - a_n$ tends to $c(1+a) - a$ as $n \to \infty$, it follows that  $P_e^{(n)}  \to 0$ only if  $c(1+a) - a < 0$, which is the same as $a > c/(1-c)$. 
Using similar arguments, it follows from \eqref{Eq_ordr_lowr}  that if  $a > c/(1-c)$, then $P_e^{(n)}  \to 0$. Hence, $P_e^{(n)}  \to 0$ if, and only if, $a > c/(1-c)$.
It can be observed from~\eqref{Eq_def_a} that $a$  is a monotonically decreasing function of $\CR$. So for $k_n=\Theta\left({n^c}\right), 0<c<1$, the capacity per unit-energy $\CCPP$  is given by
\begin{align}
\CCPP = \sup \{\CR\geq 0 : a(\CR) > c/(1-c)\} \notag
\end{align}
where we write $a(\CR)$ to make it clear that $a$ as defined in \eqref{Eq_def_a} is a function of $\CR$.
This supremum can be computed as 
\begin{equation*}
\CCPP = \begin{cases} \frac{\log e}{N_0}  \left(\frac{1}{1+\sqrt{\frac{c}{1-c}}}\right)^2, & \quad \mbox{if  } 0 < c \leq  1/2\\
\frac{\log e}{2 N_0} (1-c), & \quad \mbox{if } 1/2<c < 1
\end{cases}
\end{equation*}
which proves Part~\ref{Thm_ortho_part2}) of Theorem~\ref{Thm_ortho_code}.

\subsection{Proof of Theorem~\ref{Thm_capac_APE}}

\label{sec_APE}
\subsubsection{Proof of Part~\ref{Thm_APE_achv_part})}
We first argue  that $P_{e,A}^{(n)} \to 0$ only if  $E_n \to \infty$, and that in this case, $\CC^A \leq \frac{\log e}{N_0}$. Indeed, let \text{$P_{e,i} \triangleq \textnormal{Pr}\{\hat{W}_i\neq W_i\}$} denote the probability that message $W_i$ is decoded erroneously. We then have that $P_{e,A}^{(n)} \geq \min_{i} P_{e,i}$. Furthermore, $P_{e,i}$ is lower-bounded by the error probability of the Gaussian single-user channel, since a single-user channel can be obtained from the MnAC if a genie informs the receiver about the codewords transmitted by users $j \neq i$. By applying the lower bound \cite[eq.~(30)]{PolyanskiyPV11} on the error probability of the Gaussian single-user channel, we thus obtain
\begin{equation}
\label{eq:LB_P2P}
P_{e,A}^{(n)} \geq Q\left(\sqrt{\frac{2E_n}{N_0}}\right), \quad M_n \geq 2
\end{equation}
where $Q$ denotes \edit{the $Q$-function, i.e.,} the tail distribution function of the standard Gaussian distribution.
Hence, $P_{e,A}^{(n)} \to 0$ only if $E_n \to \infty$. As mentioned in Remark~\ref{remark}, when $E_n$ tends to infinity as $n\to\infty$, the capacity per unit-energy $\CC^A$ coincides with the capacity per unit-energy defined in \cite{Verdu90}, which for the Gaussian single-user channel is given by $\frac{\log e}{N_0}$ \cite[Ex.~3]{Verdu90}. Furthermore, if $P_{e,A}^{(n)} \to 0$ as $n\to\infty$, then there exists at least one user $i$ for which $P_{e,i} \to 0$ as $n\to\infty$. By the above genie argument, this user's rate per unit-energy is upper-bounded by the capacity per unit-energy of the Gaussian single-user channel. Since, for the class of symmetric codes considered in this paper, each user transmits at the same rate per unit-energy, we conclude that $\CC^A \leq \frac{\log e}{N_0}$.

We next show that any rate per unit-energy $\CR < \frac{\log e}{N_0}$ is achievable by an orthogonal-access scheme where each user uses an orthogonal codebook of blocklength $n/k_n$. 
To transmit message $w_i$, user $i$ sends in his assigned slot the codeword $\bx(w_i) = (x_1(w_i), \ldots, x_{n/k_n }(w_i))$, which is given by 
\begin{align*}
x_{j}(w_i) = \begin{cases}
\sqrt{E_n}, & \text{ if }  j=w_i\\
0, & \text{ otherwise}.
\end{cases}
\end{align*}

To show that the probability of error vanishes, we use the following bound from Lemma~\ref{Lem_ortho_code}:
\begin{align}
P_{e,i} \leq
\begin{cases}
\exp\left\{- \frac{\ln M_n}{\CR}\left(\frac{\log e}{2 N_0} - \CR \right) \right\}, & \text{ if } 0 < \CR \leq \frac{1}{4} \frac{\log e}{N_0}\\
\exp\left\{- \frac{\ln M_n}{\CR}  \left(\sqrt{\frac{\log e}{N_0}} - \sqrt{ \CR}\right)^2\right\}, & \text{ if } \frac{1}{4} \frac{\log e}{N_0}  \leq \CR \leq \frac{\log e}{N_0}.
\end{cases}
\label{Eq_ortho_prob_uppr1}
\end{align}
It follows from~\eqref{Eq_ortho_prob_uppr1} that, if $\CR < \frac{\log e}{N_0}$ and $M_n \to \infty$ as $n \to \infty$, then $P_{e,i}, i=1,\ldots, k_n$ tends to zero as $n \to \infty$.   Since $k_n =o(n)$, it follows that $M_n = n/k_n $ tends to $\infty$, as $n \to \infty$. Thus, for any $\CR < \frac{\log e}{N_0}$,  the probability of error $P_{e,i}$ vanishes. This implies that also $P_{e,A}^{(n)} $ vanishes as $n \to \infty$, thus proving Part~\ref{Thm_APE_achv_part}).

\subsubsection{Proof of Part~\ref{Thm_APE_conv_part})}
Fano's inequality yields that
\begin{equation*}
\log M_n \leq 1+ P_{e,i}\log M_n+ I(W_i; \hat{W}_i), \quad \mbox{for } i=1,\ldots, k_n.
\end{equation*}
Averaging over all $i$'s then gives
\begin{IEEEeqnarray}{lCl}
	\log M_n & \leq & 1+  \frac{1}{k_n} \sum_{i=1}^{k_n} P_{e,i}\log M_n+ \frac{1}{k_n}  I({\bf W}; {\bf \hat{W}}) \nonumber\\ 
	& \leq &  1+P_{e,A}^{(n)}\log M_n+ \frac{1}{k_n} I(\bW; \bY) \nonumber\\
	& \leq & 1 +  P_{e,A}^{(n)} \log M_n+\frac{n}{2k_n} \log \left(1+\frac{2 k_nE_n}{nN_0}\right) \IEEEeqnarraynumspace\label{eq_rate_APE_uppr}
\end{IEEEeqnarray}
where the first inequality follows because the messages $W_i, i=1, \ldots, k_n$ are independent and because conditioning reduces entropy, the second inequality follows from the definition of $P_{e,A}^{(n)}$ and the data processing inequality, and the third inequality follows by upper-bounding $I(\bW;\bY)$ by $\frac{n}{2} \log \bigl(1+\frac{2 k_nE_n}{nN_0}\bigr)$.

Dividing both sides of \eqref{eq_rate_APE_uppr} by $E_n$, and solving the inequality for $\PR$, we obtain the upper bound
\begin{equation}
\label{eq_Part2_Th1_end1}
\PR\leq \frac{ \frac{1}{E_n} + \frac{n}{2 k_nE_n}\log(1+\frac{ 2k_nE_n}{nN_0})}{1 -P_{e,A}^{(n)}}.
\end{equation}
As argued at the beginning of the proof of Part~\ref{Thm_achv_part}), we have $P_{e,A}^{(n)} \to 0$ only if $E_n \to \infty$. If $k_n = \Omega(n)$, then this implies that $k_nE_n/n \to \infty$ as $n \to \infty$. It thus follows from \eqref{eq_Part2_Th1_end1} that, if $k_n = \Omega(n)$, then $\CC^A=0$, which is Part~\ref{Thm_conv_part}) of Theorem~\ref{Thm_capac_APE}.




\section{Capacity per Unit-Energy of Random Many-Access Channels}
\label{sec_random_MnAC}
In this section, we \edit{consider} the case where the users' activation probability can be strictly smaller than $1$. In Subsection~\ref{Sec_results}, we discuss the capacity per unit-energy of random MnACs. In particular, we present  our main result in Theorem~\ref{Thm_random_JPE}, which characterizes the capacity per unit-energy in terms of $\el_n$ and $k_n$.  Then, in
Theorem~\ref{Thm_ortho_accs}, we analyze the largest rate per unit-energy achievable using an orthogonal-access scheme. Finally, in Theorem~\ref{Thm_capac_PUPE}, we briefly discuss the behaviour of the capacity per unit-energy of random MnAC for APE. The proofs of Theorems~\ref{Thm_random_JPE}--\ref{Thm_capac_PUPE} are presented in Subsections~\ref{Sec_proof_random}, \ref{Sec_ortho_access}, and \ref{sec_average}, respectively.

\subsection{Capacity per Unit-Energy of Random MnAC}
\label{Sec_results}

Before presenting our results, we first note that  the case where $k_n$ vanishes as $n \to \infty$ is uninteresting. 
Indeed, this case only happens if $\alpha_n \to 0$. 
Then, the probability that all the users are inactive, given by $\bigl( (1-\alpha_n)^{\frac{1}{\alpha_n}}\bigr)^{k_n}$,
tends to one since  $(1-\alpha_n)^{\frac{1}{\alpha_n}} \to 1/e $ and $k_n \to 0$.  Consequently, if each user employs a code with $M_n=2$ and $E_n =0$ for all $n$, and if the decoder always declares that all users are inactive, then the probability of error $P_{e}^{(n)}$ vanishes as $n \to \infty$. This implies that $\CC= \infty$. \edit{ In the following, we avoid this trivial case and assume that $\el_n$ and $\alpha_n$ are such that $k_n = \Omega(1)$. This implies that the inverse of $\alpha_n$ is upper-bounded by the order of  $\ell_n$, i.e., $\frac{1}{\alpha_n} = O(\ell_n)$.} We have the following theorem.

\begin{theorem}
	\label{Thm_random_JPE}
	Assume that $k_n =\Omega(1)$. Then the capacity per unit-energy of the Gaussian random MnAC has the following behavior:
	\begin{enumerate}
		\item 	If $k_n \log \el_n = o(n)$,   then $\CC = \frac{\log e}{N_0}$. \label{Thm_achv_part}
		\item 	If $k_n \log \el_n = \omega(n)$, then $\CC =0$. 	\label{Thm_conv_part}
		\edit{\item 	If $k_n \log \ell_n = \Theta(n)$, then $ 0< \CC < \frac{\log e}{N_0}$.
			\label{Thm_exact_order}	
		}
	\end{enumerate}
\end{theorem}

\begin{proof}
	See Subsection~\ref{Sec_proof_random}.
\end{proof}

Theorem~\ref{Thm_random_JPE} demonstrates that there  is a sharp transition between orders of growth of $k_n$ where interference-free communication is feasible and orders of growth where no positive rate per unit-energy is feasible. Recall that the same behaviour was observed for the non-random-access case ($\alpha_n=1$), where the transition threshold separating these two regimes is at the order of growth $n/ \log n$, as shown in Theorem~\ref{Thm_nonrandom}. For \edit{a} general $\alpha_n$, this transition threshold depends both on $\ell_n$ and $k_n$.
\edit{However, when $\liminf_{n\to\infty} \alpha_n >0$, then $k_n = \Theta(\ell_n)$ and the order of growth of $k_n \log \el_n$ coincides with that of both $ k_n \log k_n$ and  $\el_n \log \el_n$. It follows that, in this case, the transition thresholds for both $ k_n$ and $\el_n$ are also at $n /\log n$, since $k_n \log k_n = \Theta(n)$  is equivalent to $k_n = \Theta(n/\log n)$. 
}

When $\alpha_n \to 0$, the orders of growth of $k_n$ and $\el_n$ are different and the transition threshold for $\el_n$ is in general larger than $n / \log n$.
\edit{For example, when  $\ell_n = n$ and $\alpha_n = \frac{1}{\sqrt{n}}$, then $k_n \log \ell_n =\sqrt{n} \log n = o(n)$, so  all users can communicate without interference.} Thus, random user-activity enables interference-free communication at an order of growth above the limit $n/ \log n$. Similarly, when $\alpha_n \to 0$, the transition threshold for $k_n$ \edit{may be smaller than $n/ \log n$, even though this is only the case if $\ell_n$ is superpolynomial in $n$. For example, when $\ell_n=2^n$ and $\alpha_n=\frac{\sqrt{n}}{2^n \log n }$, then $k_n=\frac{\sqrt{n}}{\log n}=o(n/\log n)$ and $k_n\log \ell_n=\frac{n^{3/2}}{\log n}=\omega(n)$, so no positive rate per unit-energy is feasible.} This implies that treating a random MnAC with $\el_n$ users as a non-random MnAC with $k_n$ users may be overly-optimistic, since it suggests that interference-free communication is feasible at orders of growth of $k_n$ where actually no positive rate per unit-energy is feasible.

In the proof of Part~\ref{Thm_Infeasble_achv}) of Theorem~\ref{Thm_nonrandom}, we have shown that, when \mbox{$k_n =o(n/ \log n)$} and \mbox{$\alpha_n =1$}, an orthogonal-access scheme achieves the capacity per unit-energy. It turns out that this is not necessarily the case anymore when $\alpha_n \to 0$, as we show in the following theorem.
\begin{theorem}
	\label{Thm_ortho_accs}
	Assume that $k_n = \Omega(1)$. The largest rate per unit-energy $\CCP$ achievable with an orthogonal-access scheme satisfies the following:
	\begin{enumerate}[1)]
		\item  If  $ \el_n = o(n/ \log n)$, then $\CCP = \frac{\log e }{N_0}$. \label{Thm_ortho_accs_achv}
		\item If $ \el_n = \omega(n/ \log n)$, then $\CCP =0$. \label{Thm_ortho_accs_conv}
		\edit{	\item 
			If $\ell_n = \Theta( \frac{n}{\log n})$, then $0 < \CCP < \frac{\log e}{N_0}$.
			\label{Thm_exact_ortho}}
	\end{enumerate}
\end{theorem}
\begin{proof}
	See Subsection~\ref{Sec_ortho_access}.
\end{proof}

Observe that there is again a sharp transition between the orders of growth of $\el_n$ where interference-free communication is feasible and orders of growth where no positive rate per unit-energy is feasible. In contrast to the optimal transmission scheme, the transition threshold for the orthogonal-access schemes is located at $n/ \log n$, irrespective of the behavior of $\alpha_n$. Thus, by using an orthogonal-access scheme, we treat the random MnAC as if it were a non-random MnAC. This also implies that there are orders of growth of $\el_n$ and $k_n$ where non-orthogonal-access schemes are necessary to achieve the capacity per unit-energy.


Next we present our results on the behaviour of capacity per unit-energy for APE. To this end, we first note that, if $\alpha_n \to 0$ as $n \to \infty$, then $\Pr\{W_i =0\} \to 1$ for all $i=1,\ldots,\el_n$.
Consequently, if each user employs a code with $M_n =2$ and $E_n =0$ for all $n$, and if the decoder always declares that all users are inactive, then the APE vanishes as $n \to \infty$. This implies that $\CC^A = \infty$.
In the following, we avoid this trivial case and assume that $\alpha_n$ is bounded away from zero.
For $\alpha_n=1$ (non-random-access case) and APE,  we showed in Theorem~\ref{Thm_capac_APE} that if the number of users grows sublinear in $n$, then each user can achieve the single-user capacity per unit-energy, and if the order of growth is linear or superlinear, then the capacity per unit-energy is zero. Perhaps not surprisingly, the same result holds in the random-access case since, when $\alpha_n$ is bounded away from zero, $k_n$ is of the same order as $\el_n$. We have the following theorem.
\begin{theorem}
	\label{Thm_capac_PUPE}
	If \edit{$\liminf_{n\to\infty} \alpha_n > 0$}, then $\CC^A$ has the following behavior:
	\begin{enumerate}
		\item  If  $\el_n  = o(n)$, then $\CC^A = \frac{\log e}{N_0}$.  \label{Thm__avg_achv_part}
		\item If $\el_n  = \Omega(n)$, then $\CC^A =0$.	\label{Thm__avg_conv_part}
	\end{enumerate}
\end{theorem}
\begin{proof}
See Subsection~\ref{sec_average}.	
\end{proof}

\subsection{Proof of Theorem~\ref{Thm_random_JPE}}
\label{Sec_proof_random}

We first give an outline of the proof. The achievability scheme to show Part~\ref{Thm_achv_part}) is a non-orthogonal-access scheme where the codewords of all users are of length $n$ and the codebooks of  different users may be  different. In each codebook, the codewords consist of two parts. The first $n''$ symbols are  a signature part that is used to convey to the receiver that the user is active. The remaining $n-n''$ symbols  are used to send the message. The decoder follows a two-step decoding process. First, it determines which users are active, then it decodes the messages of all users that are estimated as active. For such a two-step decoding process, we analyze two types of errors: the detection error and the decoding error. We show that, if $k_n \log \el_n$ is sublinear in $n$, then the  probability of detection error and the probability of decoding error  tend to zero as $n \to \infty$.  The proof of Part~\ref{Thm_conv_part}) follows along similar lines as that of Part~\ref{Thm_Infeasble_convrs}) of Theorem~\ref{Thm_nonrandom}. We first show that the probability of error vanishes only if the energy $E_n$ scales at least \edit{logarithmically} in the total number of users, i.e., $E_n = \Omega(\log \el_n)$. We \edit{then} show that a positive rate per unit-energy is achievable only if the total power of the active users, given by $k_n E_n/n$, is bounded as $n \to \infty$. \edit{The proof of Part~\ref{Thm_conv_part}) concludes by noting that, if $k_n\log \ell_n$ is superlinear in $n$, then there is no $E_n$ that can simultaneously satisfy these two conditions. Part~\ref{Thm_exact_order}) follows by revisiting the proofs of Parts~\ref{Thm_achv_part}) and \ref{Thm_conv_part}) for the case where $k_n\log \el_n=\Theta(n)$.}

\subsubsection{Proof of Part~\ref{Thm_achv_part})}
We use an achievability scheme  with a decoding process consisting of two steps.  First, the receiver determines which users are active. \edit{It then fixes an arbitrary positive integer $\xi$, based on which it decides whether it will decode the messages of all active users, or whether it will declare an error. Specifically, if the number of estimated active users is less than or equal to $\xi k_n$, then the receiver decodes the messages of all active users.} If the number of estimated active users is greater than $\xi k_n$, then it declares an error.\footnote{\edit{The threshold $\xi k_n$ becomes inactive when $\alpha_n$ is bounded away from zero. Indeed, the number of estimated active users is a random variable taking value in $\{0,\ldots,\ell_n\}$. When $\alpha_n$ is bounded away from zero, $\ell_n = k_n/\alpha_n$ is bounded by $\xi k_n$ for some positive integer $\xi$. Hence, in this case we can find a threshold $\xi$ such that the receiver will never have to declare an error.}} \edit{By the union bound,} the total error probability of this scheme can be upper-bounded by
\begin{align*}
P(\cD) + \sum_{k'_n=1}^{\xi k_n}\Pr\{K'_n=k_n'\}\Pm+ \Pr\{K'_n>\xi k_n\}
\end{align*}
where $K'_n$ \edit{is a random variable describing the number} of active users, $P(\cD)$ is the probability of a detection error, and $\Pm$ is the probability of a decoding error when the receiver has correctly detected that there are $k'_n$ active users. In the following, we show that these probabilities vanish as $n\to\infty$ for any fixed, positive integer $\xi$. Furthermore, by Markov's inequality, we have that $\Pr\{K'_n>\xi k_n\}\leq 1/\xi$. It thus follows that the total probability of error vanishes as we let first $n\to\infty$ and then $\xi\to\infty$.

To enable user detection at the receiver, out of $n$ channel uses, each user uses the first $n''$ channel uses to send its signature and  $n'=n -n''$ channel uses for sending the message. The
signature uses energy $E_n''$ out of $E_n$, while the energy used for sending message is given by $E_n' = E_n -E_n''$.  

Let  $\bs_i$ denote the signature of  user $i$ and $\tilde{\bx}_i(w_i)$ denote the codeword of length $n'$ for sending the message $w_i$, where $w_i =1,\ldots, M_n$. Then, the codeword $\bx_i(w_i)$ is given by the concatenation of $\bs_i$ and $\tilde{\bx}_i(w_i)$, denoted as
\begin{align*}
\bx_i(w_i) = (\bs_i,  \tilde{\bx}_i(w_i)).
\end{align*}
For a given arbitrary $0 < b < 1$, we let
\begin{equation}
n'' = bn, \quad \label{Eq_channel_choice}
\end{equation}
\label{Page_energy}
\begin{equation}
\label{Eq_energy_choice}
E_n'' = bE_n,  \quad E_n = c_n \ln \el_n 
\end{equation}
with $c_n = \ln (\frac{n}{k_n\ln \el_n})$.\footnote{\edit{In our scheme, a fraction of the total energy must be assigned to the signature part in order to ensure that the detection error probability vanishes as $n\to\infty$. However, this incurs a loss in rate per unit-energy, so this fraction will be made arbitrarily small at the end of the proof. Alternatively, one could consider a sequence $\{b_n\}$ that satisfies $b_n\to 0$ and $b_n c_n \to\infty$ as $n\to\infty$.}}
Based on the first $n''$ received symbols, the receiver detects which users are active. We need the following lemma to show that the detection error probability  vanishes as $n \to \infty$.
\edit{
\begin{lemma}
	\label{Lem_usr_detect}
	Assume that $k_n \log \el_n = O(n)$, and let $E_n = c_n \ln \ell_n$, where
		\begin{equation*}
		c_n =
		\begin{cases}
		\ln\left(\frac{n}{k_n \log \ell_n}\right), &\text{ if } k_n \log \ell_n = o(n)\\
		c',  &\text{ if } k_n \log \ell_n = \Theta(n)
		\end{cases}
		\end{equation*}
		for some positive constant $c'$ that is independent of $n$. If $c'$ is sufficiently large, then there exist signatures $\bs_i, i=1, \ldots, \el_n$ with  $n''=bn$ channel uses  and energy $E_n''=bE_n$ such that $P(\cD)$ vanishes as $n \to \infty$.
\end{lemma}
}
\begin{proof}
	\edit{The proof of Lemma~\ref{Lem_usr_detect} follows along similar lines as that of~\cite[Th.~2]{ChenCG17}. However, there are some differences in the settings considered. Here, the goal is to achieve user detection with minimum energy, whereas in \cite{ChenCG17} the goal is to achieve user detection with the minimum number of channel uses. Furthermore, the energy we assign to the signature part is proportional to the total energy $E_n$, cf.~\eqref{Eq_energy_choice}, whereas in \cite{ChenCG17} the energy assigned to the signature part is proportional to the number of channel uses. These differences have the positive effect that, in our proof, the condition \cite[Eq.~(19)]{ChenCG17}, namely that
\begin{equation*}
\lim_{n\to\infty} \ell_n e^{-\delta k_n} =0
\end{equation*}	
for all $\delta>0$, is not necessary. For a full proof of Lemma~\ref{Lem_usr_detect}, see Appendix~\ref{Sec_Lem_detct_proof}.
}
\end{proof}

We next use the following lemma to show that  $\Pm$ vanishes as $n \to \infty$ uniformly in  $k'_n \in \cK_n$, where $\cK_n \triangleq \{1, \ldots, \xi k_n\}$.
\begin{lemma}
	\label{Lem_err_expnt}
	Let  $A_{k'_n} \triangleq \frac{1}{k_n'} \sum_{i=1}^{k_n'} \I{ \hat{W}_i \neq W_i}$ and \mbox{$\cA_{k'_n} \triangleq \{1/k_n', \ldots,1 \}$}, where $\I{\cdot}$ denotes the indicator function. Then, for any arbitrary $0<\rho \leq 1$, we have
	\begin{equation}
	\textnormal{Pr}\{A_{k'_n} = a\} \leq \left(\frac{1}{\mu}\right)^{2k'_n} {k_n' \choose a k_n'} M_n^{a k_n' \rho} e^{-n'E_{0,k_n'}(a, \rho,n)}, \quad a\in\cA_{k'_n}\label{Eq__random_prob_err}
	\end{equation}
	where
	\begin{align}
	E_{0,k_n'}(a, \rho,n)  \triangleq  \frac{\rho}{2} \ln \left(1+\frac{a 2k_n' E_n'}{n'(\rho +1)N_0}\right) \label{Eq_random_expnt}
	\end{align}
	and
	\begin{align}
	\mu & \triangleq \int \I{ \|\bv\|^2 \leq E_n'} \prod_{i=1}^{n'} \tilde{q}(v_i) d \bv \label{Eq_def_mu}
	\end{align}
	is a normalizing constant. In~\eqref{Eq_def_mu}, $\tilde{q}(\cdot)$ denotes the probability density function of a zero-mean Gaussian random variable with variance $E_n'/(2n')$ and $\bv = (v_1,v_2,\ldots, v_{n'})$.
\end{lemma}
\begin{proof}
	The upper bound in~\eqref{Eq__random_prob_err}  without the factor $(1/\mu)^{2k'_n}$ can be obtained using random coding with i.i.d. Gaussian inputs~\cite[Th.~2]{Gallager85}. However, 	while i.i.d. Gaussian codebooks satisfy the energy constraint on average (averaged over all codewords), there may be some codewords in the codebook that violate it. 
	We therefore need to adapt the proof of~\cite[Th.~2]{Gallager85} as follows.	Let 
	\begin{align*}
	\tilde{\bq}(\bv) & = \prod_{i=1}^{n'} \tilde{q}(v_i).
	\end{align*}
	For codewords $\tilde{\bX}_i,i=1,\ldots,k'_n$ which are distributed according to $\tilde{\bq}(\cdot)$, the probability $\Pr(A_{k'_n}=a)$ can be upper-bounded as~\cite[Th.~2]{Gallager85}
	\begin{align}
	\Pr(A_{k'_n} = a) & \leq {k_n' \choose a k_n'} M_n^{a k_n' \rho}  \int \tilde{\bq}(\tilde{\bx}_{ak'_n+1}) \cdots \tilde{\bq}(\tilde{\bx}_{k'_n}) \; G ^{1+\rho} \; d\tilde{\bx}_{ak'_n+1} \cdots d\tilde{\bx}_{k'_n} \;d\tilde{\by}  \label{eq_a1}
	\end{align}
	where 
	\begin{align*}
	G & =  \int \tilde{\bq}(\tilde{\bx}_1) \cdots \tilde{\bq}(\tilde{\bx}_{ak'_n}) \left(  p(\tilde{\by} \mid \tilde{\bx}_{1},\cdots, \tilde{\bx}_{k'_n})\right) ^{1/1+\rho} d\tilde{\bx}_{1} \cdots d\tilde{\bx}_{ak'_n} \notag.
	\end{align*}	
	Using the fact that the channel is memoryless, the RHS of~\eqref{Eq__random_prob_err} without the factor $(1/\mu)^{2k'_n}$ follows from~\eqref{eq_a1}. The  case of $k'_n =2$ was analyzed  in~\cite[Eq.~(2.33)]{Gallager85}.

	Now suppose that all codewords are generated according to the distribution
	\begin{align*}
	\bq(\bv) & = \frac{1}{\mu} \I{ \|\bv\|^2 \leq E_n'} \tilde{\bq}(\bv).
	\end{align*}
	Clearly, such codewords  satisfy the energy constraint $E'_n$ with probability one. Furthermore,
	\begin{align}
	\bq(\bv) & \leq \frac{1}{\mu} \tilde{\bq}(\bv). \label{Eq_prob_signt_uppr1}
	\end{align}
	By replacing $\tilde{\bq}(\cdot)$ in~\eqref{eq_a1} by $\bq(\cdot)$, and upper-bounding $\bq(\cdot)$ by~\eqref{Eq_prob_signt_uppr1}, we obtain that
	\begin{align}
	\textnormal{Pr}\{A_{k'_n} = a\} \leq 		\left(\frac{1}{\mu}\right)^{(1+\rho)(ak'_n)}\left(\frac{1}{\mu}\right)^{k'_n - ak'_n} {k_n' \choose a k_n'} M_n^{a k_n' \rho} e^{-n'E_{0,k_n'}(a, \rho,n)}, \quad a\in\cA_{k'_n}. \label{Eq_Prob_An_uppr}
	\end{align}
	From the definition of $\mu$, we have that $0 < \mu \leq 1$.
	Since we further have $\rho\leq 1$ and $a\leq 1$, it follows that 
	$(1/\mu)^{(1+\rho)(ak'_n)} \leq 	(1/\mu)^{a k_n'+k_n'} $. Using this bound in~\eqref{Eq_Prob_An_uppr}, we obtain~\eqref{Eq__random_prob_err}.
\end{proof}

Next, we show that   $\left(\frac{1}{\mu}\right)^{2k'_n} \to 1$ as $n \to \infty$ uniformly in $k'_n \in \cK_n$. Let $\bH$ be a Gaussian vector which is distributed according to $\tilde{\bq}(\cdot)$. Then, by the definition of $\mu$, we have
\begin{align}
\mu & = 1 - \Pr\left(\|\bH_1\|_2^2 > E_n'\right) \notag
\end{align}
so $(1/\mu)^{2 k'_n} \geq 1$.
Let us consider $\bH_0 \triangleq \frac{2 n'}{E_n'} \|\bH_1\|_2^2$,  which  has a central chi-square distribution with $n'$ degrees of freedom. Then, 
\begin{align*}
\Pr\left(\|\bH_1\|_2^2 > E_n'\right) & = \Pr(\bH_0 > 2 n').
\end{align*}
 So, from
the Chernoff bound we obtain that 
\begin{align*}
\Pr(\bH_0 > a) & \leq \frac{E(e^{t\bH_0})}{e^{ta}} \\
& = \frac{(1-2t)^{-n'/2}}{e^{ta}}
\end{align*}
for every $t > 0$. By choosing $a= 2 n'$ and $t= \frac{1}{4}$, this yields
\begin{align}
\Pr(\bH_0 > 2 n') & \leq \frac{ \left(\frac{1}{2}\right)^{-n'/2}}{ \exp(n'/2) } \notag \\
& =  \exp \left[-\frac{n'}{2} \tau \right] \notag
\end{align}
where $\tau \triangleq \left( 1 - \ln 2    \right)$ is strictly positive. Thus,
\begin{align}
1 & \leq \left(\frac{1}{\mu}\right)^{2k'_n} \notag \\
& \leq \left(\frac{1}{\mu}\right)^{2\xi k_n} \notag \\
& \leq \left(1 - \exp \left[-\frac{n'}{2} \tau \right]\right)^{-(2\xi k_n)}, \quad k'_n \in \cK_n.	 \label{Eq_mu_uppr2}
\end{align}
By assumption, we have  that $k_n=o(n)$  and $n' = \Theta(n)$. 
Since for any two non-negative sequences $\{a_n\}$ and $\{b_n\}$ satisfying $a_n\to 0$ and $a_nb_n \to 0$ as $n \to \infty$, it holds that $(1-a_n)^{-b_n} \to 1$ as $n \to \infty$, we obtain that the RHS of~\eqref{Eq_mu_uppr2} tends to one as $n \to \infty$ uniformly in $k'_n\in \cK_n$. So there exists a positive constant $n_0$  that is independent of $k'_n$ and satisfies
\begin{align}
\left(\frac{1}{\mu}\right)^{2k'_n} \leq 2, \quad k'_n \in \cK_n, n \geq n_0. \label{Eq_mu_uppr1}
\end{align}

The probability of error $\Pm$ can be written as
\begin{equation}
\Pm= \sum\limits_{a\in\cA_{k'_n}} \textnormal{Pr}\{A_{k'_n} = a\}.	\label{Eq_prob_err_def}
\end{equation}
From Lemma~\ref{Lem_err_expnt} and~\eqref{Eq_mu_uppr1}, we obtain
\begin{align}
\textnormal{Pr}\{A_{k'_n} = a\} & \leq 2	{k_n' \choose a k_n'} M_n^{a k_n' \rho}  \exp[-n'E_{0,k_n'}(a, \rho,n)] \notag\\
& \leq 2 \exp\left[ k_n'H_2(a) + a \rho k_n' \log M_n -   n'E_{0,k_n'}(a, \rho,n) \right]\notag \\
& = 2 \exp \left[-E_n'f_{k'_n}(a, \rho,n)\right], \quad n \geq n_0 \label{eq:why_is_this_unnumbered}
\end{align}
where \edit{$H_2(\cdot)$ denotes the binary entropy function, and}
\begin{align}
f_{k'_n}(a, \rho,n) \triangleq \frac{n'E_{0,k_n'}(a, \rho,n)}{E_n'} - \frac{a \rho k_n' \log M_n}{E_n'} - \frac{k_n' H_2(a)}{E_n'}. \label{Eq_fn_def}
\end{align}
We next show that, for sufficiently large $n$, we have
\begin{align}
\textnormal{Pr}\{A_{k'_n} = a\} \leq 2 \exp \left[-E_n'f_{\xi k_n}(1/(\xi k_n), \rho,n)\right], \quad a\in\cA_{k'_n}, k'_n \in \cK_n. \label{Eq_err_upp_bnd}
\end{align}
To this end, we lower-bound
\begin{align}
\frac{d f_{k'_n}(a, \rho,n)}{da} & \geq \rho k'_n \left[ \frac{1}{1+\frac{2k'_nE'_n}{n'(\rho+1)N_0}} \frac{1}{(1+\rho)N_0} - \frac{\CR}{(1-b)\log e}  \right]\nonumber \\
& \geq \rho \left[ \frac{1}{1+\frac{2 \xi k_nE'_n}{n'(\rho+1)N_0}} \frac{1}{(1+\rho)N_0} - \frac{\CR}{(1-b)\log e}  \right] \label{Eq_new16}
\end{align}
by using simple algebra. This implies that
for any fixed value of $\rho$ and our choices of $E_n'$ and \mbox{$\CR = \frac{(1-b)\log e}{(1+\rho)N_0} - \delta$} (for some arbitrary $0<\delta<\frac{(1-b)\log e}{(1+\rho)N_0}$),
\begin{equation*}
\liminf_{n\to\infty} \min_{k'_n \in \cK_n}  \min_{a \in \cA_{k'_n}} \frac{d f_{k'_n}(a, \rho,n)}{da} > 0.
\end{equation*}
Indeed, \edit{the RHS of \eqref{Eq_new16} is independent of $a$ and $k'_n$ and tends to $\rho\delta$ as $n\to\infty$}, since $\frac{k_n E_n'}{n'} \to 0$ by our choice of $E_n'$ and because $k_n = o(n / \log n)$. \edit{Thus, for sufficiently large $n$ and a given $\rho$, the function $a\mapsto f_{k'_n}(a, \rho,n)$ is monotonically increasing on $\cA_{k'_n}$ for every $k'_n \in \cK_n$.}
It follows that there exists a positive constant $n'_0$ that is independent of $k'_n$ and satisfies
\begin{equation*}
\min_{a \in \cA_{k'_n}}  f_{k'_n}(a, \rho,n) =  f_{ k'_n}(1/k'_n, \rho,n), \quad k'_n \in \cK_n, n \geq n'_0.
\end{equation*}
It further follows from the definition of $f_{k'_n}(a, \rho,n) $ in \eqref{Eq_fn_def} that, for $a = 1/k'_n$ and a given $\rho$, $f_{ k'_n}(a, \rho,n')$ is decreasing in $k'_n$, since in this case the first two terms on the RHS of~\eqref{Eq_fn_def} are independent of $k'_n$ and the third term is increasing in $k'_n$.
Hence, we can further lower-bound
\begin{equation*}
\min_{a \in \cA_{k'_n}}  f_{k'_n}(a, \rho,n) \geq  f_{\xi k_n}(1/(\xi k_n), \rho,n), \quad k'_n \in \cK_n, n \geq n'_0.
\end{equation*}

Next, we show that, for our choice of $E_n'$ and \mbox{$\CR$}, we have
\begin{equation}
\label{eq:lim_pos}
\liminf_{n \rightarrow \infty}  f_{\xi k_n}(1/(\xi k_n), \rho,n) >0.
\end{equation}
Let
\begin{IEEEeqnarray}{rCl}
	i_n(\rho) & \triangleq & \frac{n' E_{0,\xi k_n}(1/(\xi k_n),\rho,n)}{E_n'}\label{Eq_new17} \\
	j(\rho) & \triangleq & \frac{\rho \CR}{(1-b)\log e} \label{Eq_new18}\\
	h_n(1/(\xi k_n)) & \triangleq & \frac{\xi 	k_n H_2(1/(\xi k_n))}{E_n'}. \label{Eq_new19}
\end{IEEEeqnarray}
Note that $\frac{h_n(1/(\xi k_n))}{j(\rho)}$ vanishes as $n \to \infty$ for our choice of $E_n'$.
Consequently,
\begin{IEEEeqnarray*}{lCl}
	\liminf_{n \rightarrow \infty} f_{\xi k_n}(1/(\xi k_n), \rho,n) 
	& = & j(\rho)  \biggl\{\liminf_{n\to\infty} \frac{i_n(\rho)}{j(\rho)} - 1 \biggr\}.
\end{IEEEeqnarray*}
The term  $j( \rho) =\rho \CR/(1-b)\log e$ is bounded away from zero for our choice of $\CR$ and $\delta < \frac{(1-b)\log e}{(1+\rho)N_0}$.  Furthermore, since $E_n'/n' \to 0$, we get
\begin{equation}
\label{eq:we_get_this}
\lim_{n\to\infty} \frac{i_n(\rho)}{j(\rho)} = \frac{(1-b)\log e}{(1+\rho)N_0 \CR}
\end{equation}
which is strictly larger than $1$ for our choice of $\CR$. So, \eqref{eq:lim_pos} follows.

We conclude that there exist two positive constants $\gamma$ and $n''_0 \geq \max(n_0,n_0')$ that are independent of $k'_n$ and satisfy $f_{k'_n}(a, \rho,n)  \geq \gamma $ for $a\in \cA_{k'_n}$, $k'_n\in \cK_n$, and  $n \geq n''_0$.
Consequently, \edit{it follows from \eqref{eq:why_is_this_unnumbered} that}, for $n \geq n''_0$,
\begin{align}
\textnormal{Pr}\{A_{k'_n} = a\} \leq 2 e^{-E_n'\gamma}, \quad a \in \cA_{k'_n}, k'_n \in \cK_n. \label{Eq_type_uppr}
\end{align}
Since $|\cA_{k'_n}| = k_n'$, \eqref{Eq_prob_err_def} and \eqref{Eq_type_uppr} yield that
\begin{align}
\Pm\leq k_n' 2 e^{-E_n'\gamma}, \quad k'_n \in \cK_n, n \geq n''_0. \notag
\end{align}
Further upper-bounding $k'_n \leq \xi k_n$, this implies that
\begin{align}
\sum_{k'_n=1}^{\xi k_n}\Pr\{K'_n=k_n'\}\Pm& \leq \xi k_n 2 e^{-E_n'\gamma},  \quad n \geq n''_0.  \label{Eq_sum_prob_uppr}
\end{align}
Since $E_n' = (1-b)c_n \ln \el_n$ and  $k_n = O(\el_n)$, it follows that the RHS of~\eqref{Eq_sum_prob_uppr} tends to 0 as $n \to \infty$ for our choice of \mbox{$\CR = \frac{(1-b)\log e}{(1+\rho)N_0} - \delta$}.
Since $\rho,\delta,$ and $b$ are  arbitrary, any rate $\CR < \frac{\log e}{N_0}$ is thus achievable. This proves Part~\ref{Thm_achv_part}) of Theorem~\ref{Thm_random_JPE}.

\subsubsection{Proof of Part~\ref{Thm_conv_part})}
Let $\hat{W}_i$ denote the receiver's estimate of $W_i$, and
denote by $\bW$ and $\bhW$ the vectors $(W_1,\ldots, W_{\el_n})$ and $(\hat{W_1},\ldots,\hat{W}_{\el_n})$, respectively. 
The messages $W_1,\ldots, W_{\el_n}$ are independent, so it follows from~\eqref{Eq_messge_def} that
\begin{align*}
H(\bW) = \el_n H(\bW_1) = \el_n \left(H_2(\alpha_n) + \alpha_n \log M_n \right).
\end{align*} 
Since $H(\bW) = H(\bW|\bY)+I(\bW;\bY)$, we obtain
\begin{align}
\el_n \left(H_2(\alpha_n) + \alpha_n \log M_n \right) & =H(\bW|\bY)+I(\bW;\bY). \label{Eq_messge_entrpy}
\end{align}
To bound $H(\bW)$, we  use the upper bounds~\cite[Lemma~2]{ChenCG17}
\begin{align}
H(\bW|\bY) \leq & \log 4 + 4 P_{e}^{(n)}\big(k_n \log M_n + k_n  + \el_n H_2(\alpha_n) + \log M_n \big) \label{Eq_messg_cond_entrpy}
\end{align}
and~\cite[Lemma~1]{ChenCG17}  
\begin{align}
I(\bW;\bY)  \leq  \frac{n}{2} \log \left(1+\frac{ 2k_nE_n}{nN_0}\right). \label{Eq_mutl_info_uppr}
\end{align}
Using~\eqref{Eq_messg_cond_entrpy} and~\eqref{Eq_mutl_info_uppr}  in~\eqref{Eq_messge_entrpy}, rearranging terms, and dividing by $k_nE_n$, yields
\begin{IEEEeqnarray}{lCl}
\left(1-4 P_{e}^{(n)}(1+1/k_n)\right) \CR &\leq & \frac{\log 4}{k_nE_n} +   \frac{H_2(\alpha_n)}{\alpha_n E_n} \! \left(4 P_{e}^{(n)} -1\right) \nonumber \\
& & {} +  4 P_{e}^{(n)} (1/E_n + 1/k_n) +\frac{n}{2 k_nE_n} \log \left(1+\frac{ 2k_nE_n}{nN_0}\right). \label{Eq_rate_joint_uppr}
\end{IEEEeqnarray}
We next show that, if  $k_n \log \el_n = \omega(n)$, then the RHS of~\eqref{Eq_rate_joint_uppr} tends to a non-positive value. To this end, we need the following lemma.
\begin{lemma}
	\label{Lem_energy_bound}
	If $\CR > 0$ \edit{and $\ell_n \geq 5$}, then $P_{e}^{(n)}$ vanishes as $n\to\infty$ only if  \mbox{$E_n = \Omega(\log \el_n)$}.
\end{lemma}
	\begin{proof}
	\edit{ See Appendix~\ref{Append_prob_lemma}.}
	\end{proof}

Part~\ref{Thm_conv_part}) of Theorem~\ref{Thm_random_JPE} follows now by contradiction. Indeed, let us assume that $k_n \log \el_n = \omega(n)$, $P_{e}^{(n)} \to 0$,  and $\CR >0$. \edit{The assumption $k_n\log\ell_n=\omega(n)$ implies that $\ell_n\to\infty$ as $n\to\infty$.} Then, Lemma~\ref{Lem_energy_bound} together with the assumption that $k_n = \Omega(1)$ implies that \edit{$E_n\to\infty$ and $k_n E_n=\omega(n)$.} It follows that the last term on the RHS of~\eqref{Eq_rate_joint_uppr} tends to zero as $n \to \infty$.  \edit{Furthermore,} together with the assumption that $k_n = \Omega(1)$,
and since $P_{e}^{(n)}$ tends to zero as $n \to \infty$, this implies that the first and third term on the RHS of~\eqref{Eq_rate_joint_uppr} vanish as $n \to \infty$. Finally, 
$\frac{H_2(\alpha_n)}{\alpha_n E_n}$ is a sequence of non-negative numbers and $(4 P_{e}^{(n)} -1) \to -1$ as $n \to \infty$, so the second term converges to a non-positive  value. \edit{Noting that, by the assumption $k_n=\Omega(1)$, the term $(1-4 P_{e}^{(n)}(1+1/k_n))$ tends to one as $P_{e}^{(n)} \to 0$}, we thus obtain from \eqref{Eq_rate_joint_uppr} that $\CR$ tends to a non-positive value as $n \to \infty$. This contradicts the assumption $\CR > 0$, so Part~\ref{Thm_conv_part}) of Theorem~\ref{Thm_random_JPE} follows.

\edit{\subsubsection{Proof of Part~\ref{Thm_exact_order})}
	\label{Sec_exact_random}
	To show that $\dot{C}>0$ when $k_n\log\ell_n=\Theta(n)$, we use the same achievability scheme and follow the same analysis as in the proof of Part~\ref{Thm_achv_part}) of Theorem~\ref{Thm_random_JPE}. That is, each user uses $n''=b n$ channel uses for sending a signature and $n'=n - n''$ channel uses for sending the message. Furthermore, the decoding process consists of two steps. First, the receiver determines which users are active. If the number of estimated active users is less than or equal to $\xi k_n$, for some arbitrary positive integer $\xi$, then the receiver decodes in a second step the messages of all active users. If the number of estimated active users is greater than $\xi k_n$, then the receiver declares an error. We set $E_n = c' \ln \ell_n$ for some $c'>0$ chosen sufficiently large so that, by Lemma~\ref{Lem_usr_detect}, the probability of a detection error vanishes as $n \to \infty$. We next show that there exists an $\CR>0$ such that the probability of a decoding error also vanishes as $n \to \infty$. To this end, we first argue that, if $k_n \log \ell_n = \Theta(n)$, then we can find an $\CR>0$ such that $f_{k'_n}(a, \rho,n)$ defined in~\eqref{Eq_fn_def} satisfies
	\begin{align}
	\liminf_{n\to\infty} \min_{k'_n \in \cK_n}  \min_{a \in \cA_{k'_n}}  f_{k'_n}(a, \rho,n) & > 0. \label{Eq_new13}
	\end{align}
	In Part~\ref{Thm_achv_part}) of Theorem~\ref{Thm_random_JPE}, we proved \eqref{Eq_new13} by first showing that there exists a positive constant $n'_0$ such that
	\begin{align}
	\min_{k'_n \in \cK_n} 	\min_{a \in \cA_{k'_n}}  f_{k'_n}(a, \rho,n) \geq  f_{\xi k_n}(1/(\xi k_n), \rho,n), \quad  n \geq n'_0 \label{Eq_new14}
	\end{align}
	and then 
	\begin{align}
	\liminf_{n \rightarrow \infty}  f_{\xi k_n}(1/(\xi k_n), \rho,n) >0.\label{Eq_new15}
	\end{align}
	We follow the same steps here, too. Since, by assumption, $k_n E'_n = k_n (1-b) c'\ln\ell_n=\Theta(n)$ for every fixed $c'>0$ and $n'=\Theta(n)$, there exist $r_1>0$ and $\tilde{n}_0>0$ such that
	\begin{align*}
	\frac{k_nE'_n}{n'} & \leq r_1, \quad  n\geq \tilde{n}_0.
	\end{align*}	
	It then follows from \eqref{Eq_new16} that, for every
	\begin{align}
	0< \CR < \frac{ (1-b)\log e}{(\rho+1)N_0+2\xi r_1} \label{Eq_new21}
	\end{align}
	we have that
	\begin{align*}
	\liminf_{n\to\infty} \min_{k'_n \in \cK_n}  \min_{a \in \cA_{k'_n}}  \frac{d  f_{k'_n}(a, \rho,n) }{d a} >0.
	\end{align*}
	Thus, for sufficiently large $n$ and a given $\rho$, the function $a\mapsto f_{k'_n}(a, \rho,n)$ is monotonically increasing on $\cA_{k'_n}$ for every $k'_n \in \cK_n$. This gives \eqref{Eq_new14}.
	
	To prove \eqref{Eq_new15}, we write
	\begin{equation}
	\label{eq:f(bla,bla)}
	 f_{\xi k_n}(1/(\xi k_n), \rho,n) = j(\rho)\left( \frac{i_n(\rho)}{j(\rho)} - 1 - \frac{h_n(1/(\xi k_n))}{j(\rho)} \right)
	\end{equation}
	where $i_n(\rho)$, $j(\rho)$, and $h_n(1/(\xi k_n))$ are defined in \eqref{Eq_new17}, \eqref{Eq_new18}, and \eqref{Eq_new19}, respectively. We consider two cases:
	
	\subsubsection*{Case 1---$k_n$ is unbounded} In this case, the assumption $k_n\log\ell_n=\Theta(n)$ implies that $\Theta(\log\ell_n)=o(n)$. Since $E_n'=\Theta(\log \ell_n)$ for every fixed $c'>0$, it follows that $E'_n/n'\to 0$. We thus get \eqref{eq:we_get_this}, namely
	\begin{equation*}
	\lim_{n \to \infty} \frac{i_n(\rho)}{j(\rho)}  = \frac{ (1-b)\log e}{\CR(\rho+1)N_0}.
	\end{equation*}
	If $\CR$ satisfies \eqref{Eq_new21}, then this is strictly larger than $1$. Furthermore, if $k_n$ is unbounded, then we can make $\frac{h_n(1/(\xi k_n)) }{j(\rho)}$ arbitrarily small by choosing $c'$ sufficiently large. Since $j(\rho)$ is bounded away from zero for every positive $\rho$ and $\CR$, we then obtain \eqref{Eq_new15} from \eqref{eq:f(bla,bla)}.
	
	\subsubsection*{Case 2---$k_n$ is bounded} In this case, $\Theta(\log \ell_n)=\Theta(k_n\log\ell_n)=\Theta(n)$, so for every fixed $c'>0$ we have $E'_n = \Theta(n)$. It follows that we can find $r_2 >0$ and $\tilde{n}'_0$ such that 
	\begin{equation*}
	\frac{2 E'_n}{n'(1+\rho)N_0} \leq r_2, \quad n\geq \tilde{n}'_0.
	\end{equation*}
	If we choose
	\begin{equation}
	\CR  < \frac{\ln (1+r_2)}{r_2}\frac{(1-b)\log e}{(\rho+1) N_0}  \label{Eq_new22}
	\end{equation}
	then 
	\begin{align}
	\lim_{n \to \infty} \frac{i_n(\rho)}{j(\rho)} >1.
	\end{align}
	Furthermore, if $k_n$ is bounded, then $h_n(1/(\xi k_n))$ vanishes as $n \to \infty$, since $\xi k_n$ is bounded and $E'_n \to \infty$. Recalling that $j(\rho)$ is bounded away from zero for every positive $\rho$ and $\CR$, we then again obtain \eqref{Eq_new15} from \eqref{eq:f(bla,bla)}.
	
	From \eqref{Eq_new14} and \eqref{Eq_new15}, it follows that, for every positive $\CR$ satisfying both \eqref{Eq_new21} and \eqref{Eq_new22}, there exist two positive constants $\gamma$ and $n_0''\geq \max(n_0,n_0',\tilde{n}_0')$ (where $n_0$ is as in \eqref{eq:why_is_this_unnumbered}) that are independent of $k'_n$ and satisfy $f_{k'_n}(a, \rho,n)  \geq \gamma $ for $a\in \cA_{k'_n}$, $k'_n\in \cK_n$, and  $n \geq n''_0$. It follows from \eqref{eq:why_is_this_unnumbered} that
	\begin{align}
	\sum_{k'_n=1}^{\xi k_n}\Pr\{K'_n=k_n'\}\Pm & \leq \xi k_n 2 e^{-E_n'\gamma} \notag \\
	& = 2 \xi \exp\left[-E_n'\left(\gamma - \frac{\ln  k_n}{E_n'}\right)\right],  \quad n \geq n_0''. \label{Eq_new24}
	\end{align}
	The term
	\begin{align*}
	\frac{\ln k_n}{E'_n} & = \frac{\ln k_n}{(1-b)c'\ln \ell_n}
	\end{align*}
	can be made arbitrarily small by choosing $c'$ sufficiently large since $\ln k_n \leq \ln \ell_n$. We thus have that $\frac{\ln k_n}{E'_n} < \gamma$ for sufficiently large $c'$, in which case the RHS of~\eqref{Eq_new24} vanishes as $n \to \infty$. Since, by Markov's inequality, we further have that $\Pr\{K_n'>\xi k_n\}\leq 1/\xi$, we conclude that the probability of a decoding error vanishes as we let first $n\to\infty$ and then $\xi\to\infty$. Consequently, if $k_n \log \ell_n = \Theta(n)$, then $\CC>0$.
	
	To prove that $\CC < \frac{\log e}{N_0}$,  we first note that the assumption $k_n \log \ell_n = \Theta(n)$ implies that $\ell_n \to \infty$ as $n \to \infty$.  Then, Lemma~\ref{Lem_energy_bound}
	shows that $P_e^{(n)}$ vanishes as $n \to \infty$ only if $E_n = \Omega(\log \ell_n)$. This further implies that $P_e^{(n)}\to 0$ only if $k_n E_n = \Omega(n)$. If $k_n E_n = \omega(n)$, then it follows from the proof of Part~\ref{Thm_conv_part}) of Theorem~\ref{Thm_random_JPE} that $\CC=0$.  We can thus assume without loss of optimality that $k_n E_n = \Theta(n)$. In this case, by following the arguments given in the proof of
	Part~\ref{Thm_conv_part}) of Theorem~\ref{Thm_random_JPE}, we obtain that the first and the third term on the RHS of \eqref{Eq_rate_joint_uppr} vanish as $n \to \infty$.  Furthermore, the second term tends to a non-positive value, and the factor $(1-4 P_{e}^{(n)}(1+1/k_n))$ tends to one. It then follows from \eqref{Eq_rate_joint_uppr} that
	\begin{equation}
	\label{eq:Th7_juhuu}
\CR \leq \limsup_{n\to\infty} \frac{n}{2 k_nE_n} \log \left(1+\frac{ 2k_nE_n}{nN_0}\right).
\end{equation}
Since $k_n E_n = \Theta(n)$, there exist $n_0>0$ and $l_1>0$ such that, for $n \geq n_0$, we have $\frac{k_nE_n}{n} \geq l_1$. By noting that $\frac{\log (1+x)}{x} < \log e$ for every $x >0$, we thus obtain that the RHS of \eqref{eq:Th7_juhuu} is strictly less than $\frac{\log e}{N_0}$ for $n \geq n_0$. Hence $\CC< \frac{\log e}{N_0}$, which concludes the proof of Part~\ref{Thm_exact_order}) of Theorem~\ref{Thm_random_JPE}.
}

\subsection{Proof of Theorem~\ref{Thm_ortho_accs}}
\label{Sec_ortho_access}

\subsubsection{Proof of Part~\ref{Thm_ortho_accs_achv})}
To prove Part~\ref{Thm_ortho_accs_achv}) of Theorem~\ref{Thm_ortho_accs}, we present a scheme that is similar to the one used in the proof of Part~\ref{Thm_Infeasble_achv}) of Theorem~\ref{Thm_nonrandom}.  Specifically, each user is assigned  $n/\el_n$ channel uses, out of which the first one is used for sending a pilot signal  and the rest are used for sending the message. Out of the  available energy $E_n$,   $t E_n$ (for some arbitrary $0 < t < 1$ \edit{to be determined later}) is used for the pilot signal and $(1-t)E_n$ is used for sending the message. Let $\tilde{\bx}(w) $ denote the codeword of length $\frac{n}{\el_n}-1$ for sending message $w$. Then,
user $i$ sends in his assigned slot the codeword 
\begin{align*}
\bx(w_i) = \left(\sqrt{t E_n}, \tilde{\bx}(w_i)\right).
\end{align*}
The receiver first detects from the pilot signal whether user $i$ is active or not. If the user is estimated as active, then the receiver decodes the user's message.
Let  $P_{e,i} = \textnormal{Pr}\{\hat{W_i} \neq W_i\}$ denote the probability that user $i$'s message is decoded erroneously.
Since all users follow the same coding scheme, the probability of correct decoding is given by
\begin{align}
P_c^{(n)}  = \left(1-P_{e,1}\right)^{\el_n}. \label{Eq_ortho_corrct}
\end{align}

By employing the transmission scheme that was used to prove  Theorem~\ref{Thm_nonrandom}, we get an upper bound on the probability of error $P_{e,1}$ as follows. Let $\bY_1$ denote the received vector of length $n/\el_n$ corresponding to user 1 in the orthogonal-access scheme.
From the pilot signal, which is the first symbol $Y_{11} $ of $\bY_1$, the receiver guesses whether user 1 is active or not. Specifically, the user is estimated as active if $Y_{11} > \frac{\sqrt{tE_n}}{2}$ and as inactive otherwise.
If the user is declared as active, then the receiver decodes the message from the rest of $\bY_1$.
Let $\Pr( \hat{W}_1 \neq w |W_1 = w)$ denote the decoding error probability when  message $w,w=0, \ldots, M_n$ was transmitted.
Then, $P_{e,1}$ is given by
\begin{align}
P_{e,1} & = (1-\alpha_n)\Pr( \hat{W}_1 \neq 0|W_1=0) + \frac{\alpha_n}{M_n} \sum_{w=1}^{M_n} \Pr( \hat{W}_1 \neq w |W_1 = w) \notag \\
& \leq \Pr( \hat{W}_1 \neq 0|W_1=0) + \frac{1}{M_n} \sum_{w=1}^{M_n} \Pr( \hat{W}_1 \neq w  | W_1 = w). \label{Eq_err_prob_uppr}
\end{align}
If $W_1=0$, then an error occurs  if $Y_{11} > \frac{\sqrt{tE_n}}{2}$. So, we have
\begin{align}
\Pr( \hat{W}_1 \neq 0|W_1=0) & = Q\left( \frac{\sqrt{tE_n}}{2} \right). \label{Eq_err_prob_uppr2}
\end{align}
If $w=1,\ldots,M_n$, then an error happens either \edit{by declaring the user as inactive} or by erroneously decoding the message. An active user is declared as inactive if $Y_{11} < \frac{\sqrt{tE_n}}{2}$. So, \edit{by the union bound}
\begin{align}
\frac{1}{M_n} \sum_{w=1}^{M_n} \Pr( \hat{W}_1 \neq w  | W_1 = w)& \leq  Q\left( \frac{\sqrt{tE_n}}{2} \right) + \edit{P_m}
\end{align}
\edit{where $P_m$ is the probability that the decoder correctly declares user 1 as active but erroneously decodes its message.} It then follows from~\eqref{Eq_err_prob_uppr} and \eqref{Eq_err_prob_uppr2} that
\begin{align}
P_{e,1} & \leq 2 Q\left( \frac{\sqrt{tE_n}}{2} \right)+ \edit{P_m}. \label{Eq_singl_usr_uppr}
\end{align}
By choosing $E_n=c_n \ln n$ with $c_n=\ln\left(\frac{n}{\ln n}\right)$, we can upper-bound \edit{$P_m$} by following the steps that led to~\eqref{Eq_prob_corrct}. Thus, \edit{for every $\dot{R}<\frac{\log e}{N_0}$, there exists a sufficiently large $n_0$ and a $0<t<1$ such that}
\begin{align}
P_m & \leq \frac{1}{n^2}, \quad n\geq n_0. \label{Eq_prob_decd}
\end{align}
Furthermore, for the above choice of $E_n$, there exists a sufficiently large $n_0'$ such that  
\begin{align}
2 Q\left( \frac{\sqrt{tE_n}}{2} \right) & \leq  \frac{1}{n^2}, \quad n\geq n'_0. \label{Eq_detect_prob_uppr}
\end{align}
Using~\eqref{Eq_prob_decd} and~\eqref{Eq_detect_prob_uppr}
in~\eqref{Eq_singl_usr_uppr}, we then obtain from~\eqref{Eq_ortho_corrct} that
\begin{align}
P_c^{(n)} &  \geq \left(1-\frac{2}{n^{2}}\right)^{\el_n} \notag \\
& \geq	\left(1-\frac{2}{n^{2}}\right)^{\frac{n}{\log n}}, \quad \edit{n \geq \max(n_0,n'_0)}\notag
\end{align}
which tends to one as $n \to \infty$. This proves Part~\ref{Thm_ortho_accs_achv}) of Theorem~\ref{Thm_ortho_accs}.

\subsubsection{Proof of Part~\ref{Thm_ortho_accs_conv})}
\edit{Recall that} we consider symmetric codes, i.e., the pair $(M_n,E_n)$ is the same for all users. However, each user may be assigned different numbers of channel uses. Let $n_i$ denote the number of channel uses assigned to  user $i$. For an orthogonal-access scheme, if $\el_n = \omega(n/ \log n)$, then there exists at least one user, say $i=1$, such that $n_i = o(\log n)$.  
Using that $H(W_1 |  W_1 \neq 0 ) = \log M_n$, it follows from Fano's inequality that
\begin{align}
\log M_n	& \leq 1+P_{e,1} \log M_n + \frac{n_1}{2 }\log\left(1+\frac{ 2 E_n}{n_1N_0}\right).  \nonumber 
\end{align}
This implies that the rate per unit-energy $\CR=(\log M_n)/E_n$ for user 1 is upper-bounded by 
\begin{align}
\CR \leq \frac{ \frac{1}{E_n} + \frac{n_1}{2 E_n}\log\left(1+\frac{ 2E_n}{n_1N_0}\right)}{1 -P_{e,1}}.\label{Eq_R_avg}
\end{align}
Since $\el_n = \omega(n/ \log n)$, it follows from Lemma~\ref{Lem_energy_bound} that $P_{e}^{(n)}$ goes to zero only if 
\begin{align}
E_n = \Omega(\log n). \label{Eq_ortho_enrg_lowr}
\end{align}
Furthermore, \eqref{Eq_R_avg}  implies  that $\CR>0$ only if 
$E_n = O(n_1)$. Since $n_1 = o(\log n)$, this further implies that 
\begin{align}
E_n = o(\log n). \label{Eq_ortho_enrg_uppr}
\end{align}
No sequence $\{E_n\}$ can satisfy both~\eqref{Eq_ortho_enrg_uppr} and~\eqref{Eq_ortho_enrg_lowr} simultaneously. We thus obtain that if $\el_n =\omega(n/ \log n)$, then the capacity per unit-energy is zero. This is Part~\ref{Thm_ortho_accs_conv}) of Theorem~\ref{Thm_ortho_accs}.

\edit{
	\subsubsection{Proof of Part~\ref{Thm_exact_ortho})}
	\label{Sec_exact_ortho}

	To show that $\CCP>0$, we use the same achievability scheme given in the proof of Part~\ref{Thm_ortho_accs_achv}) of Theorem~\ref{Thm_ortho_accs}. That is, each user is assigned $n/\ell_n$ channel uses, out of which one is used for sending a pilot signal and the rest are used for sending the message. Out of the available energy $E_n$, $tE_n$ (for some arbitrary $0<t<1$ to be determined later) is used for the pilot signal and $(1-t)E_n$ is used for sending the message. We choose $E_n = c'\log n$, where $c' = \frac{c}{1-t}$ and $c$ is chosen as in the proof of Part~\ref{Thm_exact_order}) of Theorem~\ref{Thm_nonrandom}. The probability of error in decoding user 1's message is then upper-bounded by \eqref{Eq_singl_usr_uppr}, namely,
	\begin{align}
	P_{e,1} & \leq 2 Q\left(\frac{\sqrt{t E_n}}{2}\right) +P_m \label{Eq_ortho_exct1}
	\end{align}
	where $P_m$ denotes the probability that the decoder correctly declares user 1 as active but makes an error in decoding its message.
	
	By the assumption $\ell_n = \Theta(n/\log n)$, there exist $n_0>0$ and $0<a_1\leq a_2$ such that, for $n \geq n_0$, we have $a_1 \frac{n}{\log n} \leq \ell_n \leq a_2 \frac{n}{\log n}$.  Since $\alpha_n \leq 1$, it follows that $k_n \leq a_2 \frac{n}{\log n}$ for $n \geq n_0$. By following the proof of
	Part~\ref{Thm_exact_order}) of Theorem~\ref{Thm_nonrandom}, we then obtain that one can set
	\begin{equation*}
 	\dot{R} = (1-t)\frac{\log e}{2}  \frac{\ln \bigl(1+\frac{ 2a_2 c}{(1+\rho)N_0}\bigr)}{2a_2 c}
 	\end{equation*}
	(for an arbitrary $0<\rho \leq 1$) and find a $c$ independent of $n$ and $t$ such that
	\begin{align}
	P_m & \leq \frac{1}{n}, \quad n \geq n_0. \label{Eq_ortho_exct2}
	\end{align}
	Furthermore, for $E_n=c'\log n$, the upper bound $Q(x) \leq \frac{1}{2} e^{-x^2/2}$, $x\geq 0$ yields that
	\begin{equation*}
	2 Q\left(\frac{\sqrt{t E_n}}{2}\right) \leq \exp\left[-\log n \frac{t}{1-t}\frac{c}{8}\right].
	\end{equation*}
	For every fixed $c$, the term $\frac{t}{1-t}\frac{c}{8}$ is a continuous, monotonically increasing, function of $t$ that is independent of $n$ and ranges from zero to infinity. 
	We can therefore find a $0<t<1$ such that
	\begin{align*}
	2 Q\left(\frac{\sqrt{t E_n}}{2}\right) & \leq \frac{1}{n}.
	\end{align*}
	Together with~\eqref{Eq_ortho_exct1} and~\eqref{Eq_ortho_exct2}, this implies that
	\begin{align}
	P_{e,1} & \leq \frac{2}{n}, \quad  n\geq n_0. \label{Eq_ortho_exct3}
	\end{align}
	
	The above scheme has a positive rate per unit-energy. It remains to show that this rate per unit-energy is also achievable. To this end, we note that,
	for an orthogonal-access scheme, the probability of  correct decoding is given by $P_c^{(n)} = (1-P_{e,1})^{\ell_n}$. It therefore follows from~\eqref{Eq_ortho_exct3} that
	\begin{align}
	P_c^{(n)} & \geq \left(1- \frac{2}{n} \right)^{a_2\frac{n}{\log n}},\quad  n\geq n_0. \label{Eq_ortho_exct4}
	\end{align}
	Since $ \left(1- \frac{2}{n} \right)^{n/2} \to 1/e$ and $\frac{2a_2}{\log n}\to 0$ as $n \to \infty$, the RHS of~\eqref{Eq_ortho_exct4} tends to one as $n \to \infty$. This implies that the probability of correct decoding tends to  one as $n\to\infty$, hence the rate per unit-energy is indeed achievable. Thus, if $\ell_n = \Theta(n/\log n)$, then $\CCP >0$.
	
		 We next show that $\CCP < \frac{\log e}{N_0}$. To this end, we first note that, if $\ell_n = \Theta( \frac{n}{\log n})$, and if we employ an orthogonal-access scheme, then there exists at least one user, say $i=1$, such that $n_1=O(\log n)$. That is, there exist $n_0>0$ and $a>0$ such that, for all $n\geq n_0$, we have $n_1 \leq a \log n$. Furthermore, Lemma~\ref{Lem_energy_bound} implies that, if $\ell_n = \Theta(n/\log n)$, then $P_e^{(n)}$ vanishes only if $E_n = \Omega(\log n)$. If $E_n = \omega(\log n)$, then it follows from~\eqref{Eq_R_avg} that a positive $\CR$ is achievable only if $n_1 = \omega(\log n)$, which contradicts the fact that $n_1=O(\log n)$. We can thus assume without loss of optimality that $E_n = \Theta(\log n)$, i.e., there exist $n'_0>0$ and $0<l_1\leq l_2$ such that, for all $n \geq n'_0$, we have $l_1 \log n \leq E_n \leq l_2 \log n$. Consequently, $\frac{E_n}{n_1} \geq \frac{l_1}{a}$ for $n \geq \max(n_0,n_0')$. The claim that $\CCP<\frac{\log e}{N_0}$ follows then directly from \eqref{Eq_R_avg}. Indeed, using that $\frac{\log (1+x)}{x}< \log e$ for every $x>0$, we obtain that
	\begin{align}
	\frac{n_1}{2 E_n}\log\left(1+\frac{ 2E_n}{n_1N_0}\right) & \leq \frac{a}{2l_1} \log \left(1+\frac{2l_1}{a N_0}\right) < \frac{\log e}{N_0}, \quad n \geq \max(n_0,n_0'). \label{Eq_R_ortho}
	\end{align}
	By \eqref{Eq_R_avg}, in the limit as $P_{e,1} \to 0$ and $E_n \to \infty$, the rate per unit-energy is upper-bounded by~\eqref{Eq_R_ortho}. It thus follows that $\CCP < \frac{\log e}{N_0}$, which concludes the proof of Part~\ref{Thm_exact_ortho}) of Theorem~\ref{Thm_ortho_accs}.
}

\subsection{Proof of Theorem~\ref{Thm_capac_PUPE}}

\label{sec_average}

	The proofs of Part~\ref{Thm__avg_achv_part}) and Part~\ref{Thm__avg_conv_part}) follow along the similar lines as those of Part~\ref{Thm_APE_achv_part}) and Part~\ref{Thm_APE_conv_part}) of Theorem~\ref{Thm_capac_APE}, respectively.
	
	\subsubsection{Proof of Part~\ref{Thm__avg_achv_part})} We first argue that $P_{e,A}^{(n)} \to 0$ only if  $E_n \to \infty$, and that in this case $\CC^A \leq \frac{\log e}{N_0}$. Indeed, we have
	\begin{align*}
	P_{e,A}^{(n)} & \geq \min_{i} \text{Pr}\{\hat{W}_i\neq W_i\} \\
	& \geq \alpha_n \text{Pr}(\hat{W_i} \neq W_i | W_i \neq 0) \; \text{ for some } i.
	\end{align*}
	\edit{Since $\liminf_{n \to  \infty} \alpha_n >0$}, this implies that 
	$P_{e,A}^{(n)} $ vanishes only if $ \text{Pr}(\hat{W_i} \neq W_i | W_i \neq 0)$ vanishes.  We next note that \mbox{$\text{Pr}(\hat{W_i} \neq W_i | W_i \neq 0)$} is lower-bounded by the error probability of the Gaussian single-user channel. By following the arguments presented at the beginning of the proof of Theorem~\ref{Thm_capac_APE}, we obtain that  $P_{e,A}^{(n)} \to 0$ only if $E_n \to \infty$, which also implies that $\CC^A \leq \frac{\log e}{N_0}$.
	
	For the achievability in Part~\ref{Thm__avg_achv_part}), we use an orthogonal-access scheme where each user uses an orthogonal codebook of blocklength $n/\el_n$. 
	Out of these $n/\el_n$ channel uses, the first one is used for sending a pilot signal to convey that the user is active, and the remaining channel uses are used to send the message. Specifically, the codeword $\bx_i(j)$ sent by user $i$ to convey message $j$ is  given by
	\begin{align*}
	x_{ik}(j) = \begin{cases}
	\sqrt{t E_n}, & \text{ if } k=1  \\
	\sqrt{(1-t) E_n}, & \text{ if }  k=j+1\\
	0, & \text{ otherwise}
	\end{cases}
	\end{align*}
	for some arbitrary $0 < t <1$. From the pilot signal, the receiver first detects whether the user is active or not. For this detection method, as noted before, the probability of detection error  is given by $2Q\left(\frac{\sqrt{tE_n}}{2} \right)$. Since
	\begin{equation*}
	E_n = \frac{\log M_n}{\CR} = \frac{\log ( \frac{n}{\el_n} -1)}{\CR}
	\end{equation*}
	and since $\el_n$ is sublinear in $n$, $E_n$ tends to infinity as $n \to \infty$. \edit{This} implies that the detection error vanishes as $n \to \infty$.	If $\CR < \frac{\log e}{N_0}$, then the probability of erroneously decoding the message also vanishes for this code, which follows from the proof of Theorem~\ref{Thm_capac_APE}. This proves Part~\ref{Thm__avg_achv_part}) of Theorem~\ref{Thm_capac_PUPE}.

	\subsubsection{Proof of Part~\ref{Thm__avg_conv_part})} Fano's inequality yields that $H(\hat{W}_i|W_i) \leq 1+ P_{e,i}\log M_n$. Since $H(W_i) = H_2(\alpha_n) + \alpha_n \log M_n$, we have
\begin{equation*}
H_2(\alpha_n) + \alpha_n \log M_n \leq 1+ P_{e,i}\log M_n+ I(W_i; \hat{W}_i)
\end{equation*}
for $i=1,\ldots, \el_n$. Averaging over all $i$'s then gives
\begin{align}
H_2(\alpha_n) + \alpha_n \log M_n &  \leq  1+  \frac{1}{\el_n} \sum_{i=1}^{\el_n} P_{e,i}\log M_n+ \frac{1}{\el_n}  I({\bf W}; {\bf \hat{W}}) \nonumber\\ 
&  \leq   1+P_{e,A}^{(n)}\log M_n+ \frac{1}{\el_n} I(\bW; \bY) \nonumber\\
&  \leq 1 +  P_{e,A}^{(n)} \log M_n+\frac{n}{2\el_n} \log \left(1+\frac{2 k_nE_n}{nN_0}\right). \label{Eq_avg_prob_uppr}
\end{align}
Here, the first inequality follows because the messages $W_i, i=1, \ldots, \el_n$ are independent and because conditioning reduces entropy, the second inequality follows from the definition of $P_{e,A}^{(n)}$ and the data processing inequality, and the third inequality follows from~\eqref{Eq_mutl_info_uppr}.

	Dividing both sides of \eqref{Eq_avg_prob_uppr} by $E_n$, and rearranging terms, yields the following upper-bound on the rate per unit-energy $\CR^A$:
\begin{equation}
\label{eq:Part2_Th1_end}
\PR\leq \frac{ \frac{1 - H_2(\alpha_n)}{E_n} + \frac{n}{2 \el_nE_n}\log(1+\frac{ 2k_nE_n}{nN_0})}{\alpha_n -P_{e,A}^{(n)}}.
\end{equation}
As noted before, $P_{e,A}^{(n)} \to 0$ only if $E_n \to \infty$.
It follows that  $\frac{1 - H_2(\alpha_n)}{E_n}$ vanishes as $n \to \infty$.
Furthermore, together with the assumptions $\el_n=\Omega(n)$ and
\edit{$\liminf_{n \to  \infty} \alpha_n >0$}, $E_n\to\infty$ \edit{implies} that $k_nE_n/n=\alpha_n \ell_nE_n/n$ tends to infinity as $n\to\infty$. This in turn implies that
\begin{equation*}
\frac{n}{2\ell_n E_n} \log\left(1+\frac{2 k_n E_n}{n N_0}\right)=\frac{n \alpha_n}{2 k_n E_n}\log\left(1+\frac{2k_n E_n}{n N_0}\right)
\end{equation*}
vanishes as $n\to \infty $. It thus follows from~\eqref{eq:Part2_Th1_end} that $\CR^A$ vanishes as $n\to\infty$, thereby proving  Part~\ref{Thm__avg_conv_part}) of Theorem~\ref{Thm_capac_PUPE}.


\section{\edit{Comparison With the Polyanskiy Setting of Many-Access Channels}}
\label{Sec_discuss}

\edit{In this paper, we basically follow the setting of the MnAC introduced Chen \emph{et al.} \cite{ChenCG17}. That is, we assume that each user has a different codebook and require the probability of error to vanish as $n\to\infty$. By Lemma~\ref{Lem_energy_infty}, the latter requirement can only be satisfied if $E_n\to\infty$ as $n\to\infty$, which for a fixed rate per unit-energy implies that $M_n\to\infty$. In other words, the payload of the user tends to infinity as $n\to\infty$.

In an attempt to introduce a notion of a random-access code that is appealing to the different communities interested in the multiple-access problem, Polyanskiy \cite{Polyanskiy17} proposed a different setting, where
\begin{enumerate}
\item all encoders use the same codebook;
\item the decoding is up to permutations of messages;
\item the probability of error is not required to vanish as $n\to\infty$.
\end{enumerate}
He further introduced the per-user probability of error
\begin{equation}
\label{eq:PUPE}
\frac{1}{k_n} \sum_{i=1}^{k_n} \text{Pr}\bigl(\{\hat{W}_i \neq W_i\} \cup \{\text{$W_j=W_i$ for some $j\neq i$}\}\bigr).
\end{equation}
As argued in \cite{Polyanskiy17}, the probability that two messages are equal is typically small, in which case the event \[\{\text{$W_j=W_i$ for some $j\neq i$}\}\] can be ignored and \eqref{eq:PUPE} is essentially equivalent to the APE defined in \eqref{eq_Pe_A}.

The setting where all encoders use the same codebook and decoding is up to permutations of messages is sometimes also referred to as \emph{unsourced multiple-access}. Unsourced multiple-access has two benefits: it may be more practical in scenarios with a large number of users, and many popular schemes, such as slotted ALOHA and coded slotted ALOHA, become achievability bounds and can be compared against each other and against information-theoretic benchmarks.
By not requiring the probability of error to vanish as $n\to\infty$, it is not necessary to let $M_n\to\infty$ with the blocklength. So, in the above setting, the payload can be of fixed size, which may be appealing from a practical perspective.

In \cite{Polyanskiy17}, Polyanskiy presented a random-coding achievability bound and used it as a benchmark for the performance of practical schemes, including coded slotted ALOHA, treating intereference as noise, and time-division multiple-access. He further studied the minimum energy-per-bit that can be achieved by an $(n,M_n,E_n,\epsilon)$ code for APE when each user has a different codebook, the payload $M_n$ and the probability of error $\epsilon$ are fixed, and the number of users grows linearly with the blocklength, i.e., $k_n=\mu n$ for some $0<\mu\ll 1$. The bounds obtained in \cite{Polyanskiy17} and in the follow-up work \cite{ZadikPT19} suggest that, whenever $\mu$ is below some critical value, the minimum energy-per-bit is independent of $\mu$. In other words, there exists a critical density of users below which interference-free communication is feasible. This is consistent with the conclusions we drew from Theorems~\ref{Thm_nonrandom} and \ref{Thm_random_JPE} for JPE, and from Theorems~\ref{Thm_capac_APE} and \ref{Thm_capac_PUPE} for APE. However, these theorems also demonstrate that there is an important difference: According to Theorems~\ref{Thm_nonrandom} and \ref{Thm_capac_APE}, a linear growth of the number of users in $n$ implies that the capacity per unit-energy $\CC$ is zero, irrespective of the value of $\mu$, and irrespective of whether JPE or APE is considered. Since rate per unit-energy is the reciprocal of energy-per-bit, this implies that the minimum energy-per-bit is infinite. In contrast, the bounds presented in \cite{Polyanskiy17} and \cite{ZadikPT19} show that the minimum energy-per-bit for a fixed probability of error $\epsilon$ is finite or, equivalently, that the $\epsilon$-capacity per unit-energy $\CCe$ is strictly positive. Thus, the capacity per unit-energy is strictly smaller than the $\epsilon$-capacity per unit-energy, which implies that, for APE, the strong converse does not hold.
}

\edit{In order to explore this point further, we discuss in the rest of this section how the largest achievable rate per unit-energy changes if we allow for a non-vanishing error probability. For the sake of simplicity, we shall assume throughout the section that users are active with probability one, i.e., $\alpha_n=1$.}

\edit{We first argue that, when the number of users is bounded in $n$, then a simple orthogonal-access scheme achieves an $\epsilon$-capacity per unit-energy that can even be larger than the single-user capacity per unit-energy $\frac{\log e}{N_0}$, irrespective of whether JPE or APE is assumed.} We shall do so by means of the following example.

\begin{example}
\label{Ex_finite_users}
Consider a $k$-user Gaussian MAC with normalized noise variance $N_0/2=1$ and where the number of users is independent of $n$. Suppose that each user has two messages  to transmit using energy $E_n=1$. Consider an orthogonal-access scheme where each user gets one channel use and remains silent in the remaining channel uses. In this channel use, each user transmits either $+1$ or $-1$ to convey its message. Since the access scheme is orthogonal, the receiver can perform independent decoding for each user, which yields $\Pr(\hat{W}_i \neq W_i)=Q(1)$. Consequently, we can achieve the rate per unit-energy $\frac{\log M_n}{E_n}=1$ at APE $P_{e,A}^{(n)}=Q(1)$ and at JPE $P_{e}^{(n)}=1 - (1 -Q(1))^k$. \edit{Since $\frac{\log e}{N_0}=\frac{\log e}{2}\approx 0.7213$, we conclude that, if $\epsilon\geq Q(1)$ (for APE) or $\epsilon\geq 1 - (1 -Q(1))^k$ (for JPE), then the $\epsilon$-capacity per unit-energy exceeds the single-user capacity per unit-energy.}
\end{example}

\begin{remark}
\label{remark2}
A crucial ingredient in the above scheme is that the energy $E_n$ is bounded in $n$. Indeed, it follows from~\cite[Th.~3]{PolyanskiyPV11} that,  if $E_n \to \infty$ as $n \to \infty$, as required, e.g., in~\cite[Def.~2]{Verdu90} (See Remark~\ref{remark}), then the $\epsilon$-capacity per unit-energy of the Gaussian single-user channel is equal to $\frac{\log e}{N_0}$, irrespective of \mbox{$0<\epsilon<1$}. The genie argument provided at the beginning of \edit{the} proof of Theorem~\ref{Thm_capac_APE} then yields that the same is true for the Gaussian MnAC.
\end{remark}

\edit{In the following two subsections,} we discuss the $\epsilon$-capacity per unit-energy when the number of users $k_n$ tends to infinity as $n$ tends to infinity. Specifically, in Subsection~\ref{Sec_non_vanish_JPE} we demonstrate that, irrespective of the order of growth of $k_n$, the $\epsilon$-capacity per unit-energy for JPE \edit{is the same as $\CC$, i.e., the strong converse holds in this case.}
 In Subsection~\ref{Sec_non_vanish_APE}, \edit{we consider the case where $k_n=\mu n$ and show by means of a simple example that, for some fixed payload $M_n$ and sufficiently small $\mu$, the $\epsilon$-capacity per unit-energy for APE is indeed independent of $\mu$, as suggested by the bounds in \cite{Polyanskiy17} and \cite{ZadikPT19}.}
 
\subsection{Non-Vanishing JPE}
\label{Sec_non_vanish_JPE}

The following theorem characterizes the behavior of the  $\epsilon$-capacity per unit-energy for JPE and  an unbounded number of users.
\begin{theorem}
	\label{Thm_non_vanish_JPE}
	The \mbox{$\epsilon$-capacity} per unit-energy $\CCe$ \edit{of the non-random MnAC with JPE} has the following behavior:
	\begin{enumerate}
		\item  	If $k_n=\omega(1)$ and $k_n = o(n/\log n)$, then $\CCe=\frac{\log e}{N_0}$ for every $0 < \epsilon<1$. \label{Thm_eps_part1}
		\item  If $k_n = \omega(n/ \log n)$, then $\CCe=0$ for every $0 < \epsilon<1$.\label{Thm_eps_part2}
	\end{enumerate}
\end{theorem}
\begin{proof}
	We first prove Part~\ref{Thm_eps_part1}). It follows from~\eqref{Eq_prob_typ_lowr} in the proof of Lemma~\ref{Lem_energy_bound} that, for \edit{$M_n\geq 2$ and $k_n\geq 5$},
	\begin{equation}
	P_{e}^{(n)} \geq  1  -   \frac{256 E_n/N_0+\log 2}{\log  k_n}.\label{Eq_joint_lowr}
	\end{equation}
	This implies that $P_{e}^{(n)}$ tends to one unless $E_n = \Omega(\log k_n)$. Since, by the theorem's assumption, \edit{we have} $k_n=\omega(1)$, it follows that $E_n\to\infty$ is necessary to achieve a JPE strictly smaller than one. As argued in Remark~\ref{remark2} (see also Remark~\ref{remark}), if $E_n\to\infty$ as $n\to\infty$, then the $\epsilon$-capacity per unit-energy of the Gaussian MnAC cannot exceed the single-user capacity per unit-energy $\frac{\log e}{N_0}$. Furthermore, by Theorem~\ref{Thm_capac_APE}, if $k_n=o(n/\log n)$, then any rate per unit-energy satisfying $\CR<\frac{\log e}{N_0}$ is achievable, hence it is also $\epsilon$-achievable. We thus conclude that, if $k_n=\omega(1)$ and $k_n=o(n/\log n)$, then $\CCe=\frac{\log e}{N_0}$ for every $0 < \epsilon<1$.
	 
To prove Part~\ref{Thm_eps_part2}), we use the upper bound~\eqref{Eq_R_avg1}, namely
	\begin{equation}
	\CR\leq \frac{ \frac{1}{k_nE_n} + \frac{n}{2 k_nE_n}\log(1+\frac{ 2k_nE_n}{nN_0})}{1 -P_{e}^{(n)}}.\label{Eq_R_avg_JPE}
	\end{equation}
By~\eqref{Eq_joint_lowr}, $P_{e}^{(n)}$ tends to one unless $E_n = \Omega(\log k_n)$. For $k_n=\omega(n/\log n)$, this implies that \mbox{$k_nE_n/n \to \infty$} as  $n \to \infty$, so the RHS of \eqref{Eq_R_avg_JPE} vanishes as $n$ tends to infinity. We thus conclude that, if $k_n=\omega(n/\log n)$, then $\CCe = 0$ for every $0<\epsilon<1$.
\end{proof}

Theorems~\ref{Thm_nonrandom} and~\ref{Thm_non_vanish_JPE} demonstrate that 
$\CCe=\CC$ for every $0 < \epsilon < 1$, provided that the number of users is unbounded in $n$. Consequently, the strong converse holds for JPE. As argued \edit{in the proof of Theorem~\ref{Thm_non_vanish_JPE}}, this result hinges on the fact that the probability of error can be strictly smaller than one only if the energy tends to infinity as $n \to \infty$. As explained in Remarks~\ref{remark} and \ref{remark2}, in this case the capacity per unit-energy cannot exceed $\frac{\log e}{N_0}$. As we shall see in the next subsection, an APE strictly smaller than one can also be achieved at a positive rate per unit-energy if the energy is bounded in $n$. \edit{This allows for a positive $\epsilon$-capacity per unit-energy for APE when $k_n$ grows linearly in $n$.}

\subsection{Non-Vanishing APE}
\label{Sec_non_vanish_APE}

In this subsection, \edit{we focus on the case where $k_n=\mu n$ and show that, when the payload of each user is $1$ bit and $\mu\leq 1$, the $\epsilon$-capacity per unit-energy for APE is indeed independent of $\mu$. This supports the conjecture in \cite{ZadikPT19} that there exists a critical density of users below which interference-free communication is feasible.}

\edit{Let $\cE^*(M,\mu,\epsilon)$ denote the minimum energy-per-bit required  to send $M$ messages at an APE not exceeding $\epsilon$ when the number of users is given by $k_n=\mu n$. While it is difficult to obtain the exact closed form expression of $\cE^*(M,\mu,\epsilon)$ for general $M,\mu$ and $\epsilon$, tight upper and lower bounds of $\cE^*(M,\mu,\epsilon)$ were derived in~\cite{Polyanskiy17, ZadikPT19}. Furthermore, as we shall argue next, if the payload of each user is 1 bit and $\mu\leq 1$, then $\cE^*(M,\mu,\epsilon)$ can be evaluated in closed form.}

For simplicity, assume that $N_0/2=1$. Then,
\begin{equation}
\label{eq:1bit_nomu}
\cE^*(2,\mu, \epsilon) = \left( \max\{0, Q^{-1}(\epsilon)\}\right)^2, \quad 0 < \mu \leq 1
\end{equation}
where $Q^{-1}$ denotes the inverse of $Q$ function.
Indeed, that $\cE^*(2,\mu, \epsilon)\geq(\max\{0, Q^{-1}(\epsilon)\})^2$ follows from \eqref{eq:LB_P2P}. Furthermore, if $\mu\leq 1$, then we can assign each user one channel use.  Following the orthogonal-access scheme presented in Example~\ref{Ex_finite_users}, but where each user transmits either $+\sqrt{E}$ or $-\sqrt{E}$ (instead of $+1$ or $-1$) with energy $E=(\max\{0,Q^{-1}(\epsilon)\})^2$, we can achieve \edit{$P_{e,A}^{(n)}\leq\epsilon$}. \edit{Thus, with energy \eqref{eq:1bit_nomu} we can send $2$ messages at an APE not exceeding $\epsilon$.}

Observe that the RHS of~\eqref{eq:1bit_nomu} does not depend on $\mu$ and agrees with the minimum energy-per-bit required to send one bit over the Gaussian single-user channel with error probability $\epsilon$. Thus, when $\mu\leq 1$, we can send one bit free of interference.
Further observe that~\eqref{eq:1bit_nomu} is finite for every positive $\epsilon$. Consequently, the $\epsilon$-capacity per unit-energy, which is given by the reciprocal of~\eqref{eq:1bit_nomu}, is strictly positive. \edit{This is in contrast to the capacity per unit-energy which}, by Part~\ref{Thm_APE_conv_part}) of Theorem~\ref{Thm_capac_APE}, is zero. \edit{Thus,} the strong converse does not hold \edit{for APE} when the number of users grows linearly in $n$.

As mentioned \edit{in the previous subsection}, to achieve a positive rate per unit-energy, it is crucial that the energy $E_n$ and payload $\log M_n$ are bounded in $n$. Indeed, for $k_n=\mu n$, the RHS of~\eqref{eq_Part2_Th1_end1} vanishes as $E_n$ tends to infinity, \edit{in which case} no positive rate per unit-energy is $\epsilon$-achievable. Moreover, for $k_n=\mu n$ and a bounded $E_n$, \eqref{eq_rate_APE_uppr}
 implies that the payload $\log M_n$ is bounded, too. We conclude that the arguably most common assumptions in the literature on MnACs---linear growth of the number of users, a non-vanishing APE, and a fixed payload---are the only set of assumptions under which a positive rate per unit-energy is achievable, unless we consider \edit{sublinear} growths of $k_n$.

\section{Conclusion}
\label{Sec_conclusion}

In this work, we analyzed  scaling laws 
of a Gaussian random \edit{MnAC} where the total number of users as well as the average number of active users may grow with the blocklength. In particular, we characterized the behaviour of the capacity per unit-energy as a function of the order of growth of the number of users for two notions of probability of error: the classical JPE and the APE proposed by Polyanskiy in~\cite{Polyanskiy17}. For both cases, we demonstrated that  there is a sharp transition between orders of growth where  all users can achieve the single-user capacity per unit-energy and orders of growth where no positive rate per unit-energy is feasible. 
When all users are active with probability one, we showed that the  transition threshold separating the two regimes  is at the order of growth $n/\log n$ for JPE, and at the order of growth $n$ for APE. While the qualitative behaviour of the capacity per unit-energy remains the same in both cases, there are some interesting differences between JPE and APE in some other aspects. For example, we  showed that an orthogonal-access scheme together with orthogonal codebooks is optimal for APE, but it is suboptimal for JPE. Furthermore, when the number of users is unbounded in $n$, the strong converse holds for JPE, but it does not hold for APE. For MnACs where the number of users grows linearly in $n$ and APE---the most common assumptions in the literature---our results imply that a positive rate per unit-energy is infeasible if we require the APE to vanish asymptotically. In contrast, due to the absence of a strong converse, a positive $\epsilon$-rate per unit-energy is feasible. To this end, however, it is necessary that the energy $E_n$ and the payload $\log M_n$  are bounded in $n$.

 For the case of  random user activity and JPE, we characterized the behaviour of the capacity per unit-energy in terms of the total number of users $\ell_n$ and the average number of active users $k_n$. We showed that, if $k_n\log \ell_n$ is sublinear in $n$, then all users can achieve the single-user capacity per unit-energy, and if $k_n \log \ell_n$ is superlinear in $n$, then the capacity per unit-energy is zero.  Consequently, there is again a sharp transition between orders of growth where interference-free communication is feasible and orders of growth where no positive rate per unit-energy is feasible, and the transition threshold separating these two regimes depends in this case on the orders of growth of both $\ell_n$ and $k_n$.
 \emph{Inter alia}, this result recovers our characterization of the non-random-access case ($\alpha_n=1$), since  $k_n \log k_n = \Theta(n)$ is equivalent to $k_n = \Theta(n/\log n)$. 
 Our result further implies that
 the orders of growth of $\ell_n$ for which interference-free communication is feasible are in general larger than $n/\log n$, and the orders of growth of $k_n$  for which interference-free communication is feasible \edit{may be} smaller than $n/\log n$. This suggests that treating a random MnAC with total number of users $\ell_n$  and average number of users $k_n$ as a non-random MnAC with $k_n$ users \edit{may be} overly-optimistic.

We finally showed that, under JPE, orthogonal-access schemes achieve the single-user capacity per unit-energy when the order of growth of $\ell_n$ is strictly below $n/ \log n$, and they cannot achieve a positive rate per unit-energy when the order of growth of $\ell_n$ is strictly above $n/\log n$, irrespective of the behaviour of $k_n$. Intuitively, by using an orthogonal-access scheme, we treat the random MnAC as if it were non-random. We conclude that orthogonal-access schemes are optimal when all users are active with probability one. However, for general $\alpha_n$, non-orthogonal-access schemes are necessary to achieve the capacity per unit-energy.

\begin{appendices}

\section{Proof of Lemma~\ref{Lem_energy_infty}}
\label{Sec_energy_infty}
The probability of error of the Gaussian MnAC cannot be smaller than that of the Gaussian point-to-point channel. Indeed, suppose a genie informs the receiver about all transmitted codewords except that of user $i$. Then the receiver can subtract the known codewords from the received vector, resulting in a point-to-point Gaussian channel. Since access to additional information does not increase the probability of error, the claim follows.

We next note that, for a Gaussian point-to-point channel, any $(n,M_n, E_n, \epsilon)$-code satisfies~\cite[Th.~2]{PolyanskiyPV11}
\begin{align}
\frac{1}{M_n} \geq Q\left(\sqrt{\frac{2E_n}{N_0}}+Q^{-1}(1-\epsilon)\right). \label{Eq_finite_energy}
\end{align}
Solving~\eqref{Eq_finite_energy}  for $\epsilon$ yields
\begin{align}
\epsilon  &\geq 1 - 	Q\left(Q^{-1}\left(\frac{1}{M_n}\right)  - \sqrt{\frac{2E_n}{N_0}}\right). \nonumber
\end{align}
It follows that the probability of error tends to zero as $n \to \infty$ only if $Q^{-1}\left(1/M_n\right)  - \sqrt{\frac{2E_n}{N_0}} \rightarrow -\infty$.
Since $Q^{-1}\left(1/M_n\right) \geq 0$ for $M_n\geq 2 $, this in turn is only the case if
$E_n \rightarrow \infty$. This proves Lemma~\ref{Lem_energy_infty}.

\section{Proof of Lemma~\ref{Lem_ortho_code}}
\label{Sec_AWGN_ortho_code}
The upper bounds on the probability of error presented in~\eqref{Eq_orth_sinlg_uppr1} and \eqref{Eq_orth_sinlg_uppr2} are proved in Appendix~\ref{Sec_uppr}. The lower bounds are proved in Appendix~\ref{Sec_lowr}.
\subsection{Upper bounds}
\label{Sec_uppr}
An upper bound on the probability of error for $M$ orthogonal codewords of maximum energy $E$ can be found in~\cite[Sec.~2.5]{ViterbiO79}:
\begin{align}
P_{e,1} & \leq (M-1)^{\rho} \exp\left[-\frac{E}{N_0} \left (\frac{\rho}{\rho+1}\right)\right]  \notag\\
& \leq  \exp\left[-\frac{E}{N_0} \left(\frac{\rho}{\rho+1}\right)+ \rho \ln M\right], \quad  \mbox{for } 0 \leq \rho\leq 1.\label{Eq_enrgy_AWGN}
\end{align}
For the rate per unit-energy $\CR = \frac{ \log M}{E}$, it follows from~\eqref{Eq_enrgy_AWGN}
that
\begin{align}
P_{e,1}  & \leq  \exp[-E E_0(\rho, \CR)], \quad \mbox{for } 0 \leq \rho \leq 1 \label{Eq_achvbl_err_exp}
\end{align}
where
\begin{align}
E_0(\rho, \CR) & \triangleq \left(\frac{1}{N_0} \frac{\rho}{\rho+1} -\frac{\rho \CR}{\log e}\right). \label{Eq_err_exp}
\end{align}
When $ 0 < \CR \leq \frac{1}{4} \frac{\log e}{N_0} $, the maximum of $E_0(\rho, \CR)$ over all $0 \leq \rho \leq 1$ is achieved for $\rho=1$.
When $\frac{1}{4} \frac{\log e}{N_0} \leq \CR \leq  \frac{\log e}{N_0}$, the maximum of $ E_0(\rho, \CR) $  is achieved for $\rho = \sqrt{ \frac{\log e}{N_0} \frac{1}{\CR} } -1 \in[0,1]$.   
So we have
\begin{align}
\max_{0 \leq \rho \leq 1 } E_0(\rho, \CR) =
\begin{cases}
\frac{1}{2 N_0} - \frac{\CR}{\log e}, & 0 < \CR \leq \frac{1}{4} \frac{\log e}{N_0}  \\
\left(\sqrt{\frac{1}{N_0}} - \sqrt{ \frac{\CR}{\log e}}\right)^2,   & \frac{1}{4} \frac{\log e}{N_0}  \leq \CR \leq \frac{\log e}{N_0} .
\end{cases}
\label{Eq_err_exp_achv}
\end{align}

Since $E = \frac{\log M}{\CR}$, we obtain from \eqref{Eq_achvbl_err_exp} and \eqref{Eq_err_exp_achv} that
\begin{align*}
P_{e,1} \leq
& \exp\left[- \frac{\ln M}{\CR}\left(\frac{\log e}{2 N_0} - \CR \right) \right], \quad \text{if } 0 < \CR \leq \frac{1}{4} \frac{\log e}{N_0}
\end{align*}
and
\begin{align*}
P_{e,1} \leq   & \exp\left[- \frac{\ln M}{\CR}  \left(\sqrt{\frac{\log e}{N_0}} - \sqrt{ \CR}\right)^2\right], \quad \text{if } \frac{1}{4} \frac{\log e}{N_0}  \leq \CR \leq \frac{\log e}{N_0}.
\end{align*}  
This proves the upper bounds on the probability of error in~\eqref{Eq_orth_sinlg_uppr1} and \eqref{Eq_orth_sinlg_uppr2}.

\subsection{Lower bounds}
\label{Sec_lowr}

To prove the lower bounds on the probability of error presented in~\eqref{Eq_orth_sinlg_uppr1} and~\eqref{Eq_orth_sinlg_uppr2}, we first argue that, for an orthogonal codebook, the optimal probability of error is achieved by codewords of equal energy. Then, for any given $\CR$ and an orthogonal codebook where all codewords have equal energy, we derive the lower bound in~\eqref{Eq_orth_sinlg_uppr2}, which is optimal at high rates. We further obtain an improved lower bound on the probability of error for low rates. 
Finally, the lower bound in~\eqref{Eq_orth_sinlg_uppr1} follows by showing that a combination of the two lower bounds yields a lower bound, too.

\subsubsection{Equal-energy codewords are optimal}

We shall argue that, for an orthogonal code with energy upper-bounded by $E_n$, there is no loss in optimality in assuming that all codewords have energy  $E_n$. To this end,
we first note that, without loss of generality, we can restrict ourselves to codewords of the form
\begin{align}
\bx_m = (0,\ldots, \sqrt{E_{\bx_m}},\ldots,0), \quad m=1, \ldots, M \label{Eq_ortho_code}
\end{align}
where  $E_{\bx_m}\leq E_n$ denotes the energy of codeword $\bx_m$. Indeed, any orthogonal codebook can be multiplied by an orthogonal matrix to obtain this form. Since the additive Gaussian noise $\bZ$ is zero mean and has a diagonal covariance matrix, this does not change the probability of error.

To argue that equal energy codewords are optimal, let us consider a code $\cC$ for which some codewords have energy strictly less than  $E_n$. From $\cC$, we can construct a new code $\cC'$ by multiplying each codeword $\bx_m$ by $\sqrt{E_n/E_{\bx_m}}$. Clearly, each codeword in $\cC'$ has energy $E_n$. Let $\bY$ and $\bY'$ denote the channel outputs when we transmit codewords from $\cC$ and $\cC'$, respectively, and let $P_e(\cC)$ and $P_e(\cC')$ denote the corresponding minimum probabilities of error. By multiplying each dimension of the channel output $\bY'$ by $\sqrt{E_{\bx_m}/E_n}$ and adding Gaussian noise of zero mean and variance $E_n/E_{\bx_m}$, we can construct a new channel output $\tilde{\bY}$ that has the same distribution as $\bY$. Consequently, $\cC'$ can achieve the same probability of error as $\cC$ by applying the decoding rule of $\cC$ to $\tilde{\bY}$. It follows that $P_e(\cC')\leq P_e(\cC)$. We conclude that, in order to find lower bounds on the probability of error, it suffices to consider codes whose codewords all have energy $E_n$.

\subsubsection{High-rate lower bound}
\edit{We next derive lower bound \eqref{Eq_orth_sinlg_uppr2},  which applies to high rates per unit-energy. To obtain this bound,}  we follow the analysis given in~\cite{ShannonGB67} (see also~\cite[Sec.~3.6.1]{ViterbiO79}).  \edit{To this end}, we shall first derive a lower bound on the maximum probability of error
\begin{align*}
P_{e_{\max}}  & \triangleq  \max_{m} P_{e_m}
\end{align*}
where $P_{e_m}$ denotes the probability of error in decoding message $m$.
In a second step, we derive from this bound a lower bound on the average probability of error $P_{e,1}$ by means of expurgation.
For $P_{e_{\max}}$, it was shown that at least one of the 
following two inequalities is always satisfied~\cite[Sec.~3.6.1]{ViterbiO79}:
\begin{align}
1/M & \geq \frac{1}{4} \exp\left[\mu(s) - s\mu'(s) -s \sqrt{2\mu''(s)}\right]  \label{Eq_lowr_first} \\
P_{e_{\max}} & \geq \frac{1}{4} \exp\left[\mu(s) + (1-s) \mu'(s) -(1-s) \sqrt{2\mu''(s)}\right] \label{Eq_lowr_second}
\end{align}
for all $0 \leq s \leq 1$, where
\begin{align}
\mu(s) & =-\frac{E}{N_0} s(1-s), \label{Eq_const1} \\
\mu'(s) & = -\frac{E}{N_0}(1-2s),\label{Eq_const2}\\
\mu''(s) & = \frac{2E}{N_0}. \label{Eq_const3}
\end{align}
By substituting these values in \eqref{Eq_lowr_first}, we obtain
\begin{align*}
\ln M \leq \frac{E}{N_0}\left[s^2 + \frac{2s}{\sqrt{E/N_0}}+ \frac{ \ln 4}{E/N_0}\right].
\end{align*}
Using that $0 \leq s \leq 1$ and that  $E = \frac{\log M}{\CR}$, this yields
\begin{align}
\CR & \leq \frac{\log e}{N_0} \left[s^2 + \frac{2}{\sqrt{E/N_0}}+ \frac{ \ln 4}{E/N_0}\right].  \label{Eq_uppr_rate}
\end{align}
Similarly, substituting \eqref{Eq_const1}-\eqref{Eq_const3} in \eqref{Eq_lowr_second} yields 
\begin{align}
P_{e_{\max}}  & \geq \exp\left[- \frac{E}{N_0}(1-s)^2 - 2(1-s)\sqrt{\frac{E}{N_0}} - \ln 4\right] \notag\\
& \geq \exp\left[- \frac{E}{N_0}\left((1-s)^2 + \frac{2}{\sqrt {E/N_0}} +\frac{ \ln 4}{E/N_0}\right)\right]. \label{Eq_low_bnd2_err_prob}
\end{align}
For a given $E$, let $\delta_E$ be defined as  $\delta_E \triangleq 2\left(\frac{2}{\sqrt {E/N_0}} +\frac{ \ln 4}{E/N_0}\right)$,
and let $s_{E} \triangleq \sqrt{\CR \frac{N_0}{\log e}- \delta_E}$.
For $s=s_{E}$, \edit{the bound \eqref{Eq_uppr_rate}}, and hence also~\eqref{Eq_lowr_first}, is violated  which implies that \eqref{Eq_low_bnd2_err_prob} must be satisfied for $s=s_{E}$. By substituting $s=s_E$ in~\eqref{Eq_low_bnd2_err_prob}, we obtain 
\begin{align}
&P_{e_{\max}} 
\geq \exp\left[- E\left(\left(\sqrt{ \frac{1}{N_0}}- \sqrt{ \frac{\CR}{\log e}  -  \frac{\delta_E}{N_0} }\right)^2 + \frac{\delta_E}{2N_0} \right) \right]. \label{Eq_err_exp_upp}
\end{align}

We next use~\eqref{Eq_err_exp_upp} \edit{and expurgation} to derive a lower bound on $P_{e,1}$. Indeed, we divide the codebook $\cC$ with $M$ messages  into two codebooks $\cC_1$ and $\cC_2$ of $M/2$ messages each, such that  $\cC_1$ contains the codewords with the smallest probability of error $P_{e_m}$  and $\cC_2$ contains the codewords with the largest $P_{e_m}$. It then holds that each codeword  in $\cC_1$ has a probability of error satisfying  $P_{e_m} \leq 2 P_{e,1}$. Consequently, the largest error probability of code $\cC_1$, denoted as $P_{e_{\max}}(\cC_1)$,  and the average error probability 
of code $\cC$, denoted as $P_e(\cC)$, satisfy
\begin{align}
P_e(\cC) \geq \frac{1}{2} P_{e_{\max}}(\cC_1). \label{Eq_max_half}
\end{align}
Applying \eqref{Eq_err_exp_upp} for $\cC_1$, and using that the rate per unit-energy of $\cC_1$ satisfies $\CR' = \frac{\log M/2}{E} = \CR - \frac{1}{E}$, we obtain 
\begin{align*}
& P_{e_{\max}} (\cC_1)
\geq\exp\left[- E\left(\left(\sqrt{ \frac{1}{N_0}}- \sqrt{ \frac{\CR}{\log e} - \frac{1}{E}  -   \frac{\delta_E}{N_0} }\right)^2 + \frac{\delta_E}{2N_0} \right) \right].
\end{align*}
Together with \eqref{Eq_max_half}, this yields
\begin{align}
P_{e,1} & \geq \exp\left[- E\left(\left(\sqrt{ \frac{1}{N_0}}- \sqrt{ \frac{\CR}{\log e} - \frac{1}{E}  -   \frac{\delta_E}{N_0} }\right)^2 + \frac{\delta_E}{2N_0} - \frac{\ln 2}{E}  \right) \right].\label{Eq_prob_low_bnd}
\end{align}
Let $\delta'_E \triangleq \frac{1}{E}  +   \frac{\delta_E}{N_0} $. Then 
\begin{align*}
\sqrt{ \frac{\CR}{\log e} - \frac{1}{E}  -   \frac{\delta_E}{N_0} } &= \sqrt{ \frac{\CR}{\log e} - \delta'_E}\\
& =  \sqrt{ \frac{\CR}{\log e} } + O(\delta'_{E})\\
&=  \sqrt{ \frac{\CR}{\log e} } + O\left(\frac{1}{\sqrt{E}}\right)
\end{align*}
where the last step follows by noting that $O(\delta'_E)=O(\delta_E)=O(1/\sqrt{E})$. Further defining $ \delta''_E \triangleq \frac{\delta_E}{2N_0} - \frac{\ln 2}{E}$, we may write \eqref{Eq_prob_low_bnd} as
\begin{align}
P_{e,1} &\geq \exp\left[- E\left(\left(\sqrt{ \frac{1}{N_0}}- \sqrt{ \frac{\CR}{\log e}} + O\left(\frac{1}{\sqrt{E}}\right)\right)^2 + \delta''_E  \right) \right] \notag \\
&  = \exp\left[- E\left(\left(\sqrt{ \frac{1}{N_0}}- \sqrt{ \frac{\CR}{\log e}} \right)^2 +O\left(\frac{1}{\sqrt{E}} \right)\right) \right] \label{Eq_prob_low_bnd1}
\end{align}
since $O(\delta''_E)=O(\delta_E)=O(1/\sqrt{E})$. By substituting $E = \frac{\log M}{\CR}$, \eqref{Eq_prob_low_bnd1} yields
\begin{align}																						
P_{e,1}  & \geq \exp\left[- \frac{\ln M}{\CR}\left(\left(\sqrt{\frac{\log e}{N_0}}- \sqrt{\CR }\right)^2 + O\left(\frac{1}{\sqrt{E}}\right) \right) \right].\label{Eq_err_exp_lwr}
\end{align}
We can thus find a function $E\mapsto \beta'_E$
\edit{of order $O(1/\sqrt{E})$} for which the lower bound in~\eqref{Eq_orth_sinlg_uppr2} holds.

\subsubsection{ Low-rate lower bound}
\edit{We next derive a lower bound on the probability of error that applies to low rates per unit-energy. This bound will then be used later to derive the lower bound~\eqref{Eq_orth_sinlg_uppr1}. To obtain this bound,} we first derive a lower bound on $P_{e,1}$ that, for low rates \edit{per unit-energy}, is tighter than \eqref{Eq_err_exp_lwr}. This bound is based on the fact that for $M$ codewords of energy $E$, the minimum Euclidean distance $d_{\min}$ between  codewords is upper-bounded by $\sqrt{2EM/(M-1)}$~\cite[Sec.~3.7.1]{ViterbiO79}. Since, for the Gaussian channel, the maximum error probability is lower-bounded by the largest  pairwise error probability, it follows that 
\begin{align}
P_{e_{\max}} & \geq Q\left( \frac{d_{\min}}{\sqrt{2 N_0}}\right) \notag \\
& \geq Q\left(  \sqrt{ \frac{ EM}{N_0(M-1)}}\right) \notag\\
& \geq \left(1-\frac{1}{EM/(N_0(M-1))}\right) \frac{e^{-\frac{EM}{2N_0(M-1)}}}{\sqrt{2 \pi} \sqrt{ EM/(N_0(M-1))}} \label{Eq_Low_rate_low}
\end{align}	
where the last inequality follows because~\cite[Prop.~19.4.2]{Lapidoth17}
\begin{align*}
Q(\beta) \geq \left(1-\frac{1}{\beta^2}\right) \frac{e^{-\beta^2/2}}{\sqrt{2 \pi} \beta}, \quad \beta>0.
\end{align*}
Let $\beta_E \triangleq \sqrt{EM/N_0(M-1)}$. It follows that
\begin{align}
\sqrt{E/N_0} \leq \beta_E \leq \sqrt{2E/N_0}, \quad M\geq 2. \label{Eq_beta_bounds}
\end{align}
Applying \eqref{Eq_beta_bounds} to \eqref{Eq_Low_rate_low} yields 
\begin{align}	
P_{e_{\max}} & \geq \frac{1}{ \sqrt{2 \pi}} \exp\left[-E\left(\frac{1}{2N_0}\left(1+ \frac{1}{M-1}\right)\right)\right]
\exp\left[\ln \left(\frac{1}{\beta_E} - \frac{1}{\beta_E^3}\right)\right] \ \notag  \\
& \geq \frac{1}{ \sqrt{2 \pi}} \exp\left[-E\left(\frac{1}{2N_0}\left(1+ \frac{1}{M-1}\right)\right)\right]
\exp\left[-E \frac{\frac{3}{2}\ln (2E/N_0)- \ln(E/N_0 -1)}{E}\right] \notag\\
& =  \exp\left[-E\left(\frac{1}{2N_0}\left(1+ \frac{1}{M-1}\right) + O\left(\frac{\ln E }{E}\right)\right)\right].\label{Eq_lowr_prob}
\end{align}
Following similar steps of expurgation as before, we obtain from \eqref{Eq_lowr_prob} the lower bound
\begin{align}
P_{e,1} &\geq \exp\left[-E\left(\frac{1}{2N_0}\left(1+ \frac{1}{\frac{M}{2}-1}\right) + O\left(\frac{\ln E }{E}\right)\right)\right]. \notag 
\end{align}
By using that $M = 2^{\CR E}$, it follows that	
\begin{align}
P_{e,1} &\geq \exp\left[-E\left(\frac{1}{2N_0}\left(1+ \frac{1}{\frac{2^{\CR E}}{2}-1}\right) + O\left(\frac{\ln E }{E}\right)\right)\right]\label{Eq_low_any_rate} 
\end{align}
from which we obtain that, for any rate per unit-energy $\CR > 0$,
\begin{align}
P_{e,1} &\geq \exp\left[-E\left(\frac{1}{2N_0} + O\left(\frac{\ln E}{E}\right)\right)\right]. \label{Eq_low_any_rate1} 
\end{align}

\subsubsection{Combining the high-rate and the low-rate bounds}
We finally show that a combination of the lower bounds \eqref{Eq_prob_low_bnd1} and \eqref{Eq_low_any_rate} yield the lower bound in \eqref{Eq_orth_sinlg_uppr1}. 

Let $P_e^{\bot}(E,M)$ denote the smallest probability of error that can be achieved by an orthogonal codebook with $M$ codewords of energy $E$. 
We first note that $P_e^{\bot}(E,M)$ is monotonically increasing in $M$. Indeed, without loss of optimality, we can restrict ourselves to codewords of the form \eqref{Eq_ortho_code}, all having energy $E$. In this case, the probability of correctly decoding message $m$ is given by~\cite[Sec.~8.2]{Gallager68}
\begin{align}
P_{c,m}^{\bot} & = \textnormal{Pr}\biggl(\bigcap_{i\neq m} \{Y_m > Y_i\}\biggm|\bX=\bx_m\biggr) \notag \\
& =\frac{1}{\sqrt{\pi N_0}}\int_{-\infty}^{\infty} \exp\left[ \frac {(y_m - \sqrt{E})^2}{N_0}  \right] \textnormal{Pr}\biggl(\bigcap_{i\neq m}\{Y_i < y_m\}\biggm|\bX=\bx_m\biggr) d y_m \notag\\
& =\frac{1}{\sqrt{\pi N_0}}\int_{-\infty}^{\infty} \exp\left[ \frac {(y_m - \sqrt{E})^2}{N_0}  \right] \left(1-Q(y_m)\right)^{M-1} d y_m
\label{Eq_prob_ortho_corrct}
\end{align}
where $Y_i$ denotes the $i$-th component of the received vector $\bY$. In the last step of \eqref{Eq_prob_ortho_corrct}, we have used that, conditioned on $\bX=\bx_m$, the events $ \{Y_i<y_m\}$, $i\neq m$ are independent and $\textnormal{Pr}(Y_i < y_m|\bX=\bx_m)$ can be computed as $1-Q(y_m)$.
Since $P_{c,m}^{\bot}$ is the same for all $m$, we have $P_e^{\bot}(E,M) = 1 -P_{c,m}^{\bot}$. The claim then follows by observing that \eqref{Eq_prob_ortho_corrct} is monotonically decreasing in $M$.

Let $ \tilde{M}$ be the largest power of 2 less than or equal to $M$. It follows by the monotonicity of $P_e^{\bot}(E,M)$ \edit{in M} that
\begin{align}
P_e^{\bot}(E,M) \geq P_e^{\bot}(E, \tilde{M}).\label{Eq_ortho_two_code1}
\end{align} 
We next show that, for every $E_1$ and $E_2$ satisfying $E=E_1+E_2$, we have
\begin{align}
P_e^{\bot}(E, \tilde{M}) \geq P_e(E_{1},  \tilde{M}, L)P_e(E_{2}, L+1) \label{Eq_prob_prod1}
\end{align}
where $P_e(E_1,  \tilde{M}, L)$ denotes the smallest probability of error that can be achieved by a codebook with $\tilde{M}$ codewords of energy $E_1$ and a list decoder of list size $L$, and $P_e(E_{2}, L+1)$ denotes the smallest probability of error that can be achieved by a codebook with $L+1$ codewords of energy $E_2$.

To prove \eqref{Eq_prob_prod1}, we follow along the lines of \cite{ShannonGB67}, which showed the corresponding result for codebooks of a given blocklength rather than a given energy. Specifically, it was shown in \cite[Th.~1]{ShannonGB67} that, for every codebook $\cC$ with $M$ codewords of blocklength $n$, and for any $n_1$ and $n_2$ satisfying $n=n_1+n_2$, we can lower-bound the probability of error by
\begin{equation}
P_e(\cC) \geq P_e(n_1,  M,L )P_e(n_2, L+1)	 \label{Eq_prob_prod2}
\end{equation}
where $P_e(n_1, M, L )$ denotes the smallest probability of error that can be achieved by a codebook with $M$ codewords of blocklength $n_1$ and a list decoder of list size $L$, and $P_e(n_2, L+1)$ denotes the smallest probability of error that can be achieved by a codebook with $L+1$ codewords of blocklength $n_2$.
This result follows by writing the codewords $\bx_m$ of blocklength $n$ as concatenations of the vectors
\begin{align*}
\bx'_m =(x_{m,1}, x_{m,2}, \ldots, x_{m,n_1})
\end{align*}
and 
\begin{align*}
\bx''_m =(x_{m,n_1+1}, x_{m,n_1+2}, \ldots,x_{m,n_1+n_2})
\end{align*}
and, likewise, by writing the received vector $\by$ as the concatenation of the  vectors $\by'$ and $\by''$ of length $n_1$ and $n_2$, respectively. Defining $\Delta_m$ as the decoding region for message $m$ and $\Delta''_m(\by')$ as the decoding region for message $m$ when $\by'$ was received, we can then write $P_e(\cC)$ as
\begin{align}
\label{eq:SGB67}
P_e(\cC) & = \frac{1}{M} \sum_{m=1}^M \sum_{\by'} p(\by'|\bx'_m) \sum_{\by''\in\bar{\Delta}''_m}p(\by''|\bx''_m)
\end{align}
where $\bar{\Delta}''_m$ denotes the complement of $\Delta''_m$. Lower-bounding first the inner-most sum in \eqref{eq:SGB67} and then the remaining terms, one can prove \eqref{Eq_prob_prod2}.

A codebook with $\tilde{M}$ codewords of the form \eqref{Eq_ortho_code} can be transmitted in $\tilde{M}$ time instants, since in the remaining time instants all codewords are zero. We can thus assume without loss of optimality that the codebook's blocklength is $\tilde{M}$. Unfortunately, when the codewords are of the form~\eqref{Eq_ortho_code}, the above approach yields \eqref{Eq_prob_prod1} only in the trivial cases where either $E_1=0$ or $E_2=0$. Indeed, $E_1$ and $E_2$ correspond to the energies of the vectors $\bx'_m$ and $\bx''_m$, respectively, and for \eqref{Eq_ortho_code} we have $\bx'_m=\mathbf{0}$ if $m>n_1$ and $\bx''_m=\mathbf{0}$ if $m\leq n_1$, where $\mathbf{0}$ denotes the all-zero vector. We sidestep this problem by multiplying the codewords by a normalized Hadamard matrix. The Hadamard matrix, denoted by $H_{j}$, is a square matrix of size $j \times j$ with entries $\pm 1$ and has the property that all rows are orthogonal. Sylvester's construction shows that there exists a Hadamard matrix of order $j$ if $j$ is a power of 2. Recalling that $\tilde{M}$ is a power of $2$, we can thus find a normalized Hadamard matrix \[\tilde{H}\triangleq\frac{1}{\sqrt{\tilde{M}}} H_{\tilde{M}}.\] Since the rows of $\tilde{H}$ are orthonormal, it follows that the matrix $\tilde{H}$ is orthogonal. Further noting that the additive Gaussian noise $\bZ$ is zero mean and has a diagonal covariance matrix, we conclude that the set of codewords $\{\tilde{H}\bx_m,\,m=1,\ldots,\tilde{M}\}$ achieve the same probability of error as the set of codewords $\{\bx_m,\,m=1,\ldots,\tilde{M}\}$. Thus, without loss of generality, we can restrict ourselves to codewords of the form $\tilde{\bx}_m=\tilde{H}\bx_m$, where $\bx_m$ is as in \eqref{Eq_ortho_code}. Such codewords have constant modulus, i.e., $|\tilde{x}_{m,k}| = \sqrt{\frac{E}{\tilde{M}}}, k=1, \ldots,\tilde{M}$. This has the advantage that the energies of the vectors
\begin{align*}
\tilde{\bx}'_m =(\tilde{x}_{m,1}, \tilde{x}_{m,2}, \ldots, \tilde{x}_{m,n_1})
\end{align*}
and 
\begin{align*}
\tilde{\bx}''_m =(\tilde{x}_{m,n_1+1}, \tilde{x}_{m,n_1+2}, \ldots,\tilde{x}_{m,n_1+n_2})
\end{align*}
are proportional to $n_1$ and $n_2$, respectively. Thus, by emulating the proof of \eqref{Eq_prob_prod2}, we can show that for every $n_1$ and $n_2$ satisfying $\tilde{M}=n_1+n_2$ and $E_i=E n_i/\tilde{M}$, $i=1,2$, we have
\begin{equation}
\label{Eq_prob_prod3}
P_e^{\bot}(E, \tilde{M}) \geq P_e(E_{1},n_1,\tilde{M}, L)P_e(E_{2},n_2,L+1)
\end{equation}
where $P_e(E_{1},n_1,\tilde{M}, L)$ denotes the smallest probability of error that can be achieved by a codebook with $\tilde{M}$ codewords of energy $E_1$ and blocklength $n_1$ and a list decoder of list size $L$, and $P_e(E_{2}, n_2, L+1)$ denotes the smallest probability of error that can be achieved by a codebook with $L+1$ codewords of energy $E_2$ and blocklength $n_2$. We then obtain \eqref{Eq_prob_prod1} from \eqref{Eq_prob_prod3} because
\begin{equation*}
P_e(E_{1},n_1,\tilde{M}, L) \geq P_e(E_{1},\tilde{M}, L) \quad \textnormal{and} \quad P_e(E_{2},n_2,L+1) \geq P_e(E_{2},L+1).
\end{equation*}

We next give a lower bound on $ P_e(E_{1}, \tilde{M}, L)$. Indeed, for list decoding of list size $L$, the inequalities \eqref{Eq_lowr_first} and \eqref{Eq_lowr_second} can be replaced by~\cite[Lemma~3.8.1]{ViterbiO79}
\begin{align}
L/M  & \geq \frac{1}{4} \exp\left[\mu(s) - s\mu'(s) -s \sqrt{2\mu''(s)}\right] \label{Eq_lowr_first_list} \\
P_{e_{\max}} & \geq \frac{1}{4} \exp\left[\mu(s) + (1-s) \mu'(s) -(1-s) \sqrt{2\mu''(s)}\right]. \label{Eq_lowr_second_list}
\end{align}
Let $\CR_1 \triangleq \frac{ \log (M/L)}{E_{1}}$ and $\tilde{\CR}_1 \triangleq \frac{ \log (\tilde{M} /L)}{E_{1}}$. From the definition of $\tilde{M} $, we have $\tilde{M} \leq M \leq 2\tilde{M}$. Consequently,
\begin{align}
\CR_1 - \frac{1}{E_1} \leq  \tilde{\CR}_1 \leq \CR_1. \label{Eq_rate_low_high}
\end{align}
By following the steps that led to \eqref{Eq_prob_low_bnd1}, we thus obtain
\begin{align}
P_e(E_{1}, \tilde{M}, L) &\geq \exp\left[- E_{1}\left(\left(\sqrt{ \frac{1}{N_0}}- \sqrt{ \frac{\tilde{\CR}_1}{\log e}}\right)^2 + O\left(\frac{1}{\sqrt{E_1}}\right) \right) \right] \notag \\
& = \exp\left[- E_{1}\left(\left(\sqrt{ \frac{1}{N_0}}- \sqrt{ \frac{\CR_1}{\log e}}\right)^2 + O\left(\frac{1}{\sqrt{E_1}}\right) \right) \right]. \label{Eq_prob_low_bnd2_list}
\end{align}
To lower-bound $P_e(E_{2}, L+1)$, we apply \eqref{Eq_low_any_rate} with $\CR_2 \triangleq \frac{\log  (L+1)}{E_2}$. \edit{This yields}
\begin{align}
P_e(E_{2}, L+1) & \geq \exp\left[-E_2\left(\frac{1}{2N_0}\left(1+ \frac{1}{\frac{2^{\CR_2 E_2}}{2}-1}\right) + O\left(\frac{\ln E_2 }{E_2}\right)\right)\right]. \label{Eq_low_rate_lowr_bnd}
\end{align}

Let \[\Xi_1(\CR_1)\triangleq\left(\sqrt{ \frac{1}{N_0}}- \sqrt{ \frac{\CR_1}{\log e}}\right)^2\] and \[\Xi_2(\CR_2) \triangleq \frac{1}{2N_0}\left(1+ \frac{1}{\frac{2^{\CR_2 E_2}}{2}-1}\right).\] Then, by substituting \eqref{Eq_prob_low_bnd2_list} and \eqref{Eq_low_rate_lowr_bnd} in \eqref{Eq_prob_prod1}, and by using \eqref{Eq_ortho_two_code1}, we get
\begin{align}
P_e^{\bot}(E, M) & \geq  \exp\left[- E_{1}\left(\Xi_1(\CR_1)+ O\left(\frac{1}{\sqrt{E_1}}\right) \right) \right] \exp\left[- E_{2} \left(\Xi_2(\CR_2)+O\left(\frac{\ln E_2 }{E_2}\right)\right) \right]. \label{Eq_ortho_lowr}
\end{align}

Applying \eqref{Eq_ortho_lowr} with a clever choice of $E_1$ and $E_2$, we can show that the error exponent of $P_e^{\bot}(E, M)$ is upper-bounded by a convex combination of $\Xi_1(\CR_1)$ and $\Xi_2(\CR_2)$. Indeed, let $\lambda \triangleq \frac{E_{1}}{E}$. Then, \edit{\eqref{Eq_ortho_lowr} can be written as}
\begin{align}
P_e^{\bot}(E, M) \geq \exp\left[ -E \left(\lambda \Xi_1(\CR_1)+ (1 - \lambda) \Xi_2(\CR_2)+ O\left(\frac{1}{\sqrt{E}}\right) \right) \right] \label{Eq_prob_prod_low}
\end{align}
and
\begin{align*}
\frac{\log M}{E} & = \frac{\log (M/L) + \log L}{E}\\ 
& = \lambda  \frac{\log (M/L)  }{E_{1}} + (1-\lambda)\frac{\log L}{E_{2}} \\
& = \lambda \CR_1 + (1-\lambda) \CR_2.
\end{align*}

Let $\CR\triangleq\frac{\log M}{E} \leq \frac{\log e}{4N_0}$ and $\gamma_E\triangleq\min\left\{\frac{1}{\sqrt{E}},\frac{\CR}{2}\right\}$. We conclude the proof of the lower bound in \eqref{Eq_orth_sinlg_uppr1} by \edit{choosing in~\eqref{Eq_prob_prod_low}}
\edit{
\begin{align*}
\lambda_E = \lambda_E & \triangleq \frac{\CR - \gamma_E}{\frac{\log e}{4N_0} - \gamma_E}
\end{align*}
}
and the rates per unit-energy $\CR_1 = \frac{1}{4} \frac{\log e}{N_0}$ and $\CR_2 =\gamma_E$. It follows that
\begin{align}
P_e^{\bot}(E, M) & \geq \exp\left[-E \left(\frac{\lambda_E}{4N_0}+ \frac{1-\lambda_E}{2N_0}+ \frac{1-\lambda_E}{2N_0(\frac{2^{\gamma_E (1-\lambda_E)E}}{2}-1)}+O\left(\frac{1}{\sqrt{E}}\right)\right)\right]\notag\\ 
& = \exp\left[-E \left(\frac{1}{2N_0}- \frac{\lambda_E}{4N_0}+ \frac{1-\lambda_E}{2N_0(\frac{2^{\gamma_E (1-\lambda_E)E}}{2}-1)}+O\left(\frac{1}{\sqrt{E}}\right)\right)\right]. \label{Eq_low_mid_rate}
\end{align}
Noting that $\lambda_E =  \frac{\CR}{(\log e)/4N_0} + O\left(\frac{1}{\sqrt{E}}\right) $, \edit{and that $\frac{1}{2^{\gamma_E (1-\lambda_E)E}} = O(1/\sqrt{E})$,}
\eqref{Eq_low_mid_rate} can be written as
\begin{align}
P_e^{\bot}(E, M) \geq \exp \left[-E \left( \frac{1}{2 N_0} - \frac{\CR}{\log e}+ O\left(\frac{1}{\sqrt{E}}\right)\right)\right], \quad 0 < \CR \leq \frac{1}{4} \frac{\log e}{N_0}. \label{Eq_lowrtae_lowr}
\end{align}
We can thus find a function $E \mapsto \beta_E$ \edit{of order $O(1/\sqrt{E})$} for which the lower bound in~\eqref{Eq_orth_sinlg_uppr1} holds.

\section{Proof of Lemma~\ref{Lem_usr_detect}}
\label{Sec_Lem_detct_proof}

\edit{
To prove Lemma~\ref{Lem_usr_detect}, we treat the cases where $\ell_n=O(1)$ and where $\ell_n=\omega(1)$ separately. In the former case, each user is assigned an exclusive channel use to convey whether it is active or not. The probability of a detection error $P(\cD)$ can then be analyzed by similar steps as in the proof of Theorem~\ref{Thm_ortho_accs}. In the latter case, we proceed similarly as in the proof of \cite[Th.~2]{ChenCG17}. That is, we draw signatures i.i.d.\ at random according to a zero-mean Gaussian distribution, followed by a truncation step to ensure that the energy of each signature is upper-bounded by $E_n''$. The decoder then produces a vector of length $\ell_n$ with zeros and ones, where a one in the $i$-th position indicates that user $i$ is active. To this end, it chooses the vector that, among all zero-one vectors with not more than a predefined number of ones, approximates the received symbols best in terms of Euclidean distance. The probability of a detection error probability $P(\cD)$ can then be analyzed by following similar steps as in the proof of \cite[Th.~2]{ChenCG17}.}

\edit{\subsection{Bounded $\ell_n$}}

\edit{We first prove the lemma when $\el_n$ is bounded in $n$.}
In this case, one can employ a scheme where each user gets an exclusive channel use to convey whether it is active or not. For such a scheme, it is easy to show that (see the proof of Theorem~\ref{Thm_ortho_accs}) the probability of a detection error $P(\cD)$ is upper-bounded by
\begin{align*}
P(\cD) & \leq \el_n e^{-E_n'' t}
\end{align*}
for some $t > 0$. \edit{Clearly, when $\ell_n$ is bounded, we have $k_n\log\ell_n =o(n)$. Thus, the energy $E_n''$ used for detection is given by $b c_n \ln \el_n$ and tends to infinity  since $c_n \to \infty$ as $n \to\infty$.} It follows that $P(\cD)$ tends to zero as $n \to \infty$.

\edit{\subsection{Unbounded $\ell_n$}}

Next we prove Lemma~\ref{Lem_usr_detect} for the case where $\el_n \to  \infty$ as $n \to \infty$. To this end, we closely follow the proof of~\cite[Th.~2]{ChenCG17}, but with the power constraint replaced by an energy constraint.
Specifically, we analyze  $P(\cD)$  for  the user-detection scheme given in~\cite{ChenCG17}, where signatures are drawn i.i.d. according to a zero-mean Gaussian distribution.
 Note that the proof in~\cite{ChenCG17} assumes that
\begin{align}
\lim\limits_{n \to \infty} \el_n e^{-\delta k_n} =0 \label{Eq_Guo_cond}
\end{align}
for all $\delta > 0$. However, in our case this assumption is not necessary.

To show that all signatures satisfy the energy constraint, we follow the technique used in the proof of Lemma~\ref{Lem_err_expnt}. Similar to Lemma~\ref{Lem_err_expnt}, we denote by $\tilde{q}(\cdot)$  the probability density function of a zero-mean Gaussian random variable with variance $E_n''/(2n'')$. We further let
\begin{align*}
\tilde{\bq}(\bu) & = \prod_{i=1}^{n''} \tilde{q}(u_i), \quad \bu =(u_1,\ldots,u_{n''})
\end{align*}
and
\begin{align*}
\bq(\bu) & = \frac{1}{\mu} \I{ \|\bu\|^2 \leq E_n''} \tilde{\bq}(\bu) 
\end{align*}
where
\begin{align*}
\mu & = \int \I{ \|\bu\|^2 \leq E_n''} \; \tilde{\bq}(\bu)  d \bu
\end{align*}
is a normalizing constant.  
Clearly, any vector $\bS_i$ distributed according to $\bq(\cdot)$ satisfies the energy constraint $E''_n$ with probability one. 
For any index set $\cI \subseteq \{1,\ldots, \el_n\}$, let the matrices $\underline{\bS}_{\cI}$ and $ \tilde{\underline{\bS}}_{\cI}$ denote the set of signatures for the users in $\cI$ that are distributed respectively as
\begin{align*}
\underline{\bS}_{\cI} & \sim \prod_{i\in I} \bq(\bS_i)
\end{align*}
and
\begin{align*}
\tilde{\underline{\bS}}_{\cI} & \sim \prod_{i\in I} \tilde{\bq}(\bS_i).
\end{align*}
As noted in the proof of Lemma~\ref{Lem_err_expnt}, we have
\begin{align}
\bq(\bs_i) & \leq \frac{1}{\mu} \tilde{\bq}( \bs_i). \label{Eq_prob_signt_uppr}
\end{align}

To analyze the detection error probability, we first define the $\el_n$-length vector $\bD^a$ as 
\begin{align*}
\bD^a \triangleq ( \I{W_1\neq0}, \ldots, \I{W_{\el_n} \neq 0}).
\end{align*}
\edit{
For some $c''>0$, let
\begin{align*}
v_n \triangleq k_n(1 + c'').
\end{align*}
}
Further let
\begin{align*}
\cB^n(v_n) \triangleq \{ \bd \in \{0,1 \}^{\el_n} : 1 \leq  |\bd| \leq v_n \}
\end{align*}
where $|\bd|$ denotes the number of $1$'s in $\bd$. We denote by $\bS^a$  the matrix of signatures of all users, which are generated independently according to $\bq(\cdot)$, and 
we denote by $\mathbf{Y}^a$ the first $n''$ received symbols, based on which the receiver performs user detection. The receiver outputs the $\hat{\bd}$ given by
\begin{align}
\hat{\bd} = \mathrm{ arg\,min}_{ \bd \in \cB^n(v_n) } \| \bY^a - \bS^a \bd \| \label{Eq_decod_rule}
\end{align}
as a length-$\el_n$ vector \edit{guessing} the set of active users. \edit{By the union bound}, the probability of a detection error $P(\cD)$ is upper-bounded by
\begin{align}
P(\cD) 
& \leq \Pr(|\bD^a| > v_n ) + \sum_{\bd \in \cB^n(v_n)} \Pr(\cE_d|\bD^a = \bd) \Pr(\bD^a = \bd) + \Pr(\cE_d| |\bD^a| = 0) \Pr(|\bD^a| = 0) \label{Eq_detect_err_uoor}
\end{align}
where \edit{$|\bD^a|$ denotes the number of $1$'s in $\bD^a$ and
$\cE_d$  denotes the event that there is a detection error}.  Next, we show that each term on the RHS of~\eqref{Eq_detect_err_uoor} vanishes as $n \to \infty$. 

Using the Chernoff bound for the binomial distribution, we have
\edit{
\begin{align}
\Pr(|\bD^a| > v_n)  & \leq \exp(-k_n c''/3)
\end{align}
which vanishes as $n \to \infty$ if $k_n$ is unbounded. For bounded $k_n$,  this probability of error vanishes by first letting $n \to \infty$ and then letting $c'' \to \infty$.
}

We continue with the term $\Pr(\cE_d|\bD^a = \bd)$. For a given $\bD^a = \bd$, let $\kappa_1$ and $\kappa_2$ denote the number of miss detections and false alarms, respectively, i.e.,
\begin{align*}
\kappa_1 &= |\{ j: d_j \neq 0, \hat{d}_j = 0 \} |\\
\kappa_2 &= |\{ j: d_j = 0, \hat{d}_j \neq 0 \} |
\end{align*}
where $d_j$ and $\hat{d}_j$ denote the $j$-th components of the \edit{vectors $\bd$ and $\hat{\bd}$, respectively}. An error happens only if either $\kappa_1$, or $\kappa_2$, or both are strictly positive. The number of users that are either active or are declared as active by the receiver satisfies $|\bd|+\kappa_2 = |\hat{\bd}|+\kappa_1$, so
\begin{align*}
|\bd|+ \kappa_2 & \leq v_n + \kappa_1
\end{align*}
since $|\hat{\bd}|$ is upper-bounded by $v_n$ by the decoding rule~\eqref{Eq_decod_rule}.
So, the pair $(\kappa_1, \kappa_2)$ belongs to the following set:
\begin{align}
\cW^{\el_n}_{\bd} = & \left\{(\kappa_1,\kappa_2): \kappa_1 \in \{0,1,\ldots, |\bd| \}, \kappa_2 \in \{0,1,\ldots,v_n\},   \kappa_1+\kappa_2 \geq 1, |\bd|+\kappa_2 \leq v_n + \kappa_1 \right\}. \label{Eq_decision_sets}
\end{align}
Let $\Pr(\cE_{\kappa_1,\kappa_2}|\bD^a = \bd)$ be the probability of having  exactly $\kappa_1$  miss detections  and $\kappa_2$ false alarms when $\bD^a = \bd$. 
For  given $ \bd$ and $\hat{\bd}$, let $\cA^* \triangleq \{j : d_j \neq 0 \}$ and $\cA \triangleq \{j : \hat{d}_j \neq 0 \}$. We further define $\cA_1 \triangleq \cA^* \setminus \cA$, $\cA_2 \triangleq \cA \setminus \cA^*$, and 
\begin{align*}
T_{\cA} & \triangleq \|\bY^a - \sum_{j \in \cA} \bS_j \|^2 -  \|\bY^a - \sum_{j \in \cA^*} \bS_j \|^2.
\end{align*}
Using the analysis that led to~\cite[eq.~(67)]{ChenCG17}, we obtain
\begin{align}
\Pr(\cE_{\kappa_1,\kappa_2}|\bD^a = \bd) & \leq \binom{|\cA^*|}{\kappa_1} \binom{\el_n}{\kappa_2} \mathrm{E}_{\underline{\bS}_{\cA^*}, \bY} \{ [\mathrm{E}_{\underline{\bS}_{\cA_2}}\{ \I{T_{\cA} \leq 0} |\underline{\bS}_{\cA^*}, \bY \}]^{\rho}|\} \notag \\
& \leq \binom{|\cA^*|}{\kappa_1} \binom{\el_n}{\kappa_2}  \left(\frac{1}{\mu} \right)^{\rho \kappa_2} \mathrm{E}_{\underline{\bS}_{\cA^*}, \bY} \{ [\mathrm{E}_{\underline{\tilde{\bS}}_{\cA_2}}\{ \I{T_{\cA} \leq 0} |\underline{\bS}_{\cA^*}, \bY \}]^{\rho}\} \notag\\
& \leq \binom{|\cA^*|}{\kappa_1} \binom{\el_n}{\kappa_2} \left(\frac{1}{\mu}\right)^{|\cA^*|} \left(\frac{1}{\mu} \right)^{\rho \kappa_2} \mathrm{E}_{\underline{\tilde{\bS}}_{\cA^*}, \bY} \{ [\mathrm{E}_{\underline{\tilde{\bS}}_{\cA_2}}\{ \I{T_{\cA} \leq 0} |\underline{\tilde{\bS}}_{\cA^*}, \bY \}]^{\rho} \} \label{Eq_detect_err_new_distr2}
\end{align}
where in the second inequality we used that 
\begin{align}
\bq( \underline{\bs}_{\cA_2}) \leq \left( \frac{1}{\mu}\right)^{\kappa_2}\prod_{i \in \cA_2} \tilde{\bq}(\bs_i ) \label{Eq_Q_uppr1}
\end{align}
and in the third inequality we used that  
\begin{align}
\bq(\underline{\bs}_{\cA^*}) \leq \left( \frac{1}{\mu}\right)^{|\cA^*|}\prod_{i \in \cA^*} \tilde{\bq}(\bs_i ). \label{Eq_Q_uppr2}
\end{align}
Here, \eqref{Eq_Q_uppr1} and~\eqref{Eq_Q_uppr2} follow from~\eqref{Eq_prob_signt_uppr}.

 For every $\rho \in [0,1]$ and $\lambda \geq 0$, we obtain 	 from~\cite[eq.~(78)]{ChenCG17} that
 \edit{
\begin{align}
\binom{|\cA^*|}{\kappa_1} \binom{\el_n}{\kappa_2} \mathrm{E}_{\underline{\tilde{\bS}}_{\cA^*}, \bY} \{ [\mathrm{E}_{\underline{\tilde{\bS}}_{\cA_2}}\{ \I{T_{\cA} \leq 0} |\underline{\tilde{\bS}}_{\cA^*}, \bY \}]^{\rho} \}& \leq \exp[ -\tilde{E}_n g^n_{\lambda, \rho}(\kappa_1,\kappa_2, \bd) ] \label{Eq_detect_prob_old_distrb}
\end{align}
}
where 
\begin{align}
\tilde{E}_n  \triangleq & E_n''/2, \notag \\
g^n_{\lambda, \rho}(\kappa_1,\kappa_2, \bd)  \triangleq & -\frac{(1-\rho)n''}{2 \tilde{E}_n} \log (1+\lambda \kappa_2 \tilde{E}_n/n'') + \frac{n''}{2\tilde{E}_n} \log \left(1+ \lambda(1-\lambda \rho)\kappa_2 \tilde{E}_n/n'' + \lambda \rho (1-\lambda \rho) \kappa_1 \tilde{E}_n/n''\right) \notag \\
&  - \frac{|\bd|}{\tilde{E}_n} H_2\left(\frac{\kappa_1}{|\bd|}\right) - \frac{\rho \el_n}{\tilde{E}_n} H_2\left(\frac{\kappa_2}{\el_n}\right). \label{Eq_err_exp}
\end{align}
It thus follows from~\eqref{Eq_detect_err_new_distr2} and \eqref{Eq_detect_prob_old_distrb} that
\begin{align}
\Pr(\cE_{\kappa_1,\kappa_2}|\bD^a = \bd) & \leq \left(\frac{1}{\mu}\right)^{|\cA^*| + \rho\kappa_2}  \exp[ -\tilde{E}_n g^n_{\lambda, \rho}(\kappa_1,\kappa_2, \bd) ].\label{Eq_err_mu_uppr}
\end{align}

\edit{We next} show that the RHS of~\eqref{Eq_err_mu_uppr} vanishes as $n \to \infty$. To this end,  we first show that $\left(\frac{1}{\mu}\right)^{|\cA^*| + \rho\kappa_2} \to 1$ as $n \to \infty$ uniformly in $(\kappa_1, \kappa_2) \in \cW^{\el_n}_{\bd}$ and $\bd \in \cB^n(v_n) $. 
From the definition of $\mu$, we have
\begin{align*}
\mu & = 1 - \Pr\left(\|\tilde{\bS}_1\|_2^2 > E_n''\right).
\end{align*}
Furthermore, by defining $\tilde{\bS}_0 \triangleq \frac{2 n''}{E_n''} \|\tilde{\bS}_1\|_2^2$ and  following the steps that led to~\eqref{Eq_mu_uppr2}, we obtain that, \edit{for  $(\kappa_1, \kappa_2) \in \cW^{\el_n}_{\bd}$ and $\bd \in \cB^n(v_n) $},
\begin{align}
1 & \leq \left(\frac{1}{\mu}\right)^{|\cA^*| + \rho\kappa_2} \notag \\
& \leq \edit{\left(\frac{1}{\mu}\right)^{2 v_n}} \notag\\
& \leq  \edit{\left(1 - \exp \left[-\frac{n''}{2} \tau \right]\right)^{-2 v_n}} \label{Eq_mu_uppr}
\end{align}
where $\tau = (1 - \ln 2)$. Here, in the second inequality we used that \edit{$|\cA^*| =|\bd| \leq v_n$ and $\rho \kappa_2 \leq v_n$.
Since $k_n \log \el_n = O(n)$, we have $k_n = o(n)$. This implies that, for every fixed $c''>0$, we have $v_n = o(n)$ because $v_n = \Theta(k_n)$.}
 Furthermore, $n'' = \Theta(n)$. 
As noted before, for any two non-negative sequences $\{a_n\}$ and $\{b_n\}$ satisfying $a_n\to 0$ and $a_nb_n \to 0$ as $n \to \infty$, it holds that $(1-a_n)^{-b_n} \to 1$ as $n \to \infty$. It follows that the RHS of~\eqref{Eq_mu_uppr} tends to one as $n \to \infty$ uniformly in $(\kappa_1, \kappa_2)\in \cW_{\bd}^{\el_n}$ and $\bd \in \cB^n(v_n)$. Consequently, there exists a positive constant $n_0$ that is independent of $\kappa_1$, $\kappa_2$, and $\bd$ and satisfies 
\begin{align}
\left( \frac{1}{\mu} \right)^{|\cA^*| + \rho\kappa_2} & \leq 2, \quad (\kappa_1, \kappa_2)\in \cW_{\bd}^{\el_n}, \bd \in \cB^n(v_n), n \geq n_0. \label{Eq_mu_lowr}
\end{align}

\edit{To bound the exponential term on the RHS of \eqref{Eq_err_mu_uppr}, we need the following lemma.
\begin{lemma}
	\label{Lem_err_exp}
	If $k_n \log \ell_n = O(n)$, and if $c'$ and $c''$ are sufficiently large, then there exist two positive constants $\gamma >0$ and $n'_0$ such that $g^n_{\frac{2}{3}, \frac{3}{4}}(\kappa_1,\kappa_2, \bd) $, i.e., \eqref{Eq_err_exp} evaluated at $\lambda=2/3$ and $\rho=3/4$, satisfies
	\begin{align}
	\min_{\bd \in \cB^n(v_n)} \min_{(\kappa_1,\kappa_2) \in \cW^{\el_n}_{\bd}} g^n_{\frac{2}{3}, \frac{3}{4}}(\kappa_1,\kappa_2, \bd)  \geq \gamma, \quad n \geq n'_0. \label{Eq_detect_err_exp}
	\end{align}
\end{lemma}
\begin{proof}
	See Appendix~\ref{Sec_err_exp}.
\end{proof}
}
\edit{Lemma~\ref{Lem_err_exp}}  implies that $\Pr(\cE_{\kappa_1,\kappa_2}| \bD^a=\ \bd)$ vanishes as $n \to\infty$ uniformly in $\bd \in \cB^n(v_n)$. Indeed, if $\bd \in \cB^n(v_n)$, then $|\bd| \leq v_n$, which implies that $\kappa_1 \leq v_n$. Furthermore, since the decoder outputs a vector in $\cB^n(v_n)$, we also have $\kappa_2 \leq v_n$.  It thus follows from \eqref{Eq_err_mu_uppr}, \eqref{Eq_mu_lowr}, and \eqref{Eq_detect_err_exp} that 
\begin{align}
\Pr(\cE_d| \bD^a=\ \bd) &= \sum_{(\kappa_1,\kappa_2) \in \mathcal{W}_{d}^{\ell_n}} \Pr(\cE_{\kappa_1,\kappa_2}| \bD^a=\ \bd) \notag \\
& \leq 2  v_n^2 \exp[-\tilde{E}_n \gamma ] \notag \\
& = 2 \exp\left[ -\tilde{E}_n \left(\gamma - \frac{ 2 \ln v_n}{\tilde{E}_n}\right)\right], \quad \bd \in \cB^n(v_n), n \geq \max(n_0,n'_0). \label{Eq_detect_err_uppr}
\end{align}
By the definition of $v_n$ and $\tilde{E}_n$,
\edit{
\begin{equation}
\frac{2\ln v_n}{\tilde{E}_n} =  \frac{4  \ln (1+c'')}{b c_n \ln \el_n} +  \frac{ 4  \ln k_n}{b c_n \ln \el_n}. \label{Eq_sn_En2}
\end{equation}
The first term on the RHS of~\eqref{Eq_sn_En2} vanishes as $n \to \infty$ since $\ell_n$ is unbounded. The second term on the RHS of~\eqref{Eq_sn_En2} is upper-bounded by $4/(b c_n)$ since $\ln k_n \leq \ln \ell_n$. This vanishes as $c_n\to\infty$ (which is the case when $k_n\log\ell_n=o(n)$), or it can be made arbitrarily small by choosing $c_n=c'$ sufficiently large (when $k_n\log\ell_n=\Theta(n)$). Consequently, we obtain that $\frac{2\ln v_n}{\tilde{E}_n} < \gamma$ for sufficiently large $n$ and $c_n$, which implies that the RHS of~\eqref{Eq_detect_err_uppr} tends to zero as $n \to \infty$.
}

We finish the proof of Lemma~\ref{Lem_usr_detect} by analyzing the third term on the RHS of~\eqref{Eq_detect_err_uoor}, namely,  $\Pr(\cE_d| |\bD^a| = 0) \Pr(|\bD^a| = 0)$. \edit{Since $ \Pr(|\bD^a| = 0)= \left((1-\alpha_n)^{\frac{1}{\alpha_n}}\right)^{k_n}$, this term is given by}
\begin{align*}
 \Pr(\cE_d| |\bD^a| = 0) \left((1-\alpha_n)^{\frac{1}{\alpha_n}}\right)^{k_n}
\end{align*}
 and vanishes if $k_n$ is unbounded. Next we show that this term also vanishes when  $k_n$ is bounded. When $|\bD^a| = 0$, an error occurs only if there are false alarms. For $\kappa_2$ false alarms, let  $\bar{\bS} \triangleq \sum_{j=1}^{\kappa_2} \bS_j$, and let $S'_i$ denote the $i$-th component of $\bar{\bS}$. From~\cite[eq.~(303)]{ChenCG17},  we obtain  the following upper bound on the probability that there are $\kappa_2$ false alarms when $|\bD^a| = 0$:
\begin{align*}
P(\cE_{\kappa_2}|  |\bd| =0) & \leq  \binom{\el_n}{\kappa_2} \mathrm{E}_{\underline{\bS}_{\cA_2}} \left[ \Pr\left\{ \sum_{i=1}^{n''} Z_i S'_i\geq \frac{1}{2} \|\bar{\bS} \|^2 \right\} \bigg| \bar{\bS} \right] \notag  \\
& \leq \left(\frac{1}{\mu}\right)^{\kappa_2}  \binom{\el_n}{\kappa_2} \mathrm{E}_{\tilde{\underline{\bS}}_{\cA_2}} \left[ \Pr\left\{ \sum_{i=1}^{n''} Z_i S'_i \geq \frac{1}{2} \|\bar{\bS} \|^2 \right\} \bigg|\bar{\bS} \right]
\end{align*}
where in the last inequality we used~\eqref{Eq_prob_signt_uppr}. By following the analysis that led to~\cite[eq.~(309)]{ChenCG17}, we obtain
\begin{align}
P(\cE_{\kappa_2}|  |\bd| =0)  & \leq  \left(\frac{1}{\mu}\right)^{\kappa_2}   \exp \left[ -\tilde{E}_n (q'_n(\kappa_2) - u_n'(\kappa_2) ) \right] \notag
\end{align}
where 
\begin{align}
q'_n(\kappa_2) &\triangleq  \frac{n''}{2\tilde{E}_n}\log \left( 1+ \frac{\kappa_2 \tilde{E}_n}{4 n''} \right) \notag
\end{align}
and 
\begin{align}
u'_n(\kappa_2) & \triangleq  \frac{ \el_n}{\tilde{E}_n} H_2\left(\frac{\kappa_2}{\el_n}\right). \notag
\end{align}
\edit{As in~\eqref{Eq_mu_lowr}}, we upper-bound $ \left(\frac{1}{\mu}\right)^{\kappa_2} \leq 2$ uniformly in $\kappa_2$  for $n \geq n_0$. Furthermore, we observe that the behaviours of $q'_n(\kappa_2)$ and $u'_n(\kappa_2)$ are similar to $q_n(\kappa_2)$ and $v_n(\kappa_2)$ given \edit{later} in~\eqref{Eq_qn} and in~\eqref{Eq_sn}, respectively. \edit{So by following those steps}, we can show that
\begin{align}
\liminf_{n \to \infty} \min_{1 \leq \kappa_2 \leq v_n} q'_n(\kappa_2) > 0  \notag
\end{align}
and 
\begin{align}
\lim_{n \to \infty}   \min_{1 \leq \kappa_2 \leq v_n} \frac{u'_n(\kappa_2)}{ q'_n(\kappa_2)} <1. \notag
\end{align}
It follows that there exist positive constants $\tau'$ and $ \tilde{n}_0$ such that
\begin{align}
P(\cE_d |  |\bd| =0)  & = \sum_{\kappa_2=1}^{v_n}P(\cE_{\kappa_2} |  |\bd| =0)\\
& \leq 2 v_n \exp \left[ -\tilde{E}_n \tau' \right], \quad n\geq \max(n_0, \tilde{n}_0). \notag
\end{align}
We have already shown that $ v_n^2 \exp [ -\tilde{E}_n \tau' ]$ vanishes as $n \to \infty$ (cf. \eqref{Eq_detect_err_uppr}--\eqref{Eq_sn_En2}), which implies that $ 2 v_n \exp [ -\tilde{E}_n \tau' ]$  vanishes, too as $n \to \infty$. It thus follows that $P(\cE_d |  |\bd| =0)$ tends to zero as $n \to \infty$. This was the last step required to prove Lemma~\ref{Lem_usr_detect}.

\subsection{Proof of Lemma~\ref{Lem_err_exp}}
\label{Sec_err_exp}

We first note that 
\begin{align}
\min_{\bd \in \cB^n(v_n)} \min_{(\kappa_1,\kappa_2) \in \cW^{\el_n}_{\bd}} g^n_{\lambda, \rho}(\kappa_1,\kappa_2, \bd)  & = \min \{ \min_{\bd \in \cB^n(v_n)}  \min_{1 \leq \kappa_1 \leq v_n} g^n_{\lambda, \rho}(\kappa_1,0,\bd), \min_{\bd \in \cB^n(v_n)}  \min_{ 1 \leq \kappa_2 \leq v_n } g^n_{\lambda, \rho}(0,\kappa_2, \bd), \notag \\
& \qquad \qquad \min_{\bd \in \cB^n(v_n)}  \min_{ \substack{ 1 \leq \kappa_1 \leq v_n \\ 1 \leq \kappa_2 \leq v_n}} g^n_{\lambda, \rho}(\kappa_1,\kappa_2, \bd)  \}. \label{Eq_inf_gn}
\end{align}
Then, we show that, for $\lambda = 2/3$ and $\rho = 3/4$,
\begin{align}
\liminf_{n \rightarrow \infty} \min_{\bd \in \cB^n(v_n)} \min_{1 \leq \kappa_1 \leq v_n} g^n_{\lambda, \rho}(\kappa_1, 0,\bd) & > 0 \label{Eq_w1_lowr} \\
\liminf_{n \rightarrow \infty} \min_{\bd \in \cB^n(v_n)} \min_{ 1 \leq \kappa_2 \leq v_n } g^n_{\lambda, \rho}(0,\kappa_2,\bd) & > 0 \label{Eq_w2_lowr} \\
\liminf_{n \rightarrow \infty} \min_{\bd \in \cB^n(v_n)} \min_{ \substack{ 1 \leq \kappa_1 \leq v_n \\ 1 \leq \kappa_2 \leq v_n}} g^n_{\lambda, \rho}(\kappa_1,\kappa_2, \bd)  & > 0  \label{Eq_w1w2_lowr}
\end{align}
from which Lemma~\ref{Lem_err_exp} follows.

In order to prove \eqref{Eq_w1_lowr}--\eqref{Eq_w1w2_lowr}, we first lower-bound $g^n_{\lambda, \rho}(\kappa_1,\kappa_2, \bd)$ by using that, for $0 \leq \lambda \rho \leq 1$, we have
\begin{align}
& 2 \log \left(1+ \lambda(1-\lambda \rho)\kappa_2 \tilde{E}_n/n'' + \lambda \rho (1-\lambda \rho) \kappa_1 \tilde{E}_n/n''\right) \notag \\
& \qquad \qquad \geq  \log \left(1+ \lambda(1-\lambda \rho)\kappa_2 \tilde{E}_n/n'' \right) +  \log \left(1+\lambda \rho (1-\lambda \rho) \kappa_1 \tilde{E}_n/n''\right). \label{Eq_log_lowr}
\end{align}
Using~\eqref{Eq_log_lowr} in the second term on the RHS of~\eqref{Eq_err_exp}, we obtain that
\begin{align}
g^n_{\lambda, \rho}(\kappa_1, \kappa_2,\bd) & \geq  a^n_{\lambda, \rho}(\kappa_1,\bd)  + b^{n}_{\lambda, \rho}(\kappa_2)
\label{Eq_gn_lowr}
\end{align}
where 
\begin{align*}
a^n_{\lambda, \rho}(\kappa_1,\bd)  \triangleq \frac{n''}{4\tilde{E}_n} \log \left(1 + \lambda \rho (1-\lambda \rho) \kappa_1 \tilde{E}_n/n''\right) - \frac{|\bd|}{\tilde{E}_n} H_2\left(\frac{\kappa_1}{|\bd|}\right)
\end{align*}
and 
\begin{align*}
b^{n}_{\lambda, \rho}(\kappa_2) \triangleq \frac{n''}{4\tilde{E}_n} \log \left(1+ \lambda(1-\lambda \rho)\kappa_2 \tilde{E}_n/n'' \right) -\frac{(1-\rho)}{2 \tilde{E}_n} \log (1+\lambda \kappa_2 \tilde{E}_n/n'') - \frac{\rho \el_n}{\tilde{E}_n} H_2\left(\frac{\kappa_2}{\el_n}\right).
\end{align*}

\subsubsection{Proof of~\eqref{Eq_w1_lowr}}
We have
\begin{align}
g^n_{\lambda, \rho}(\kappa_1, 0,\bd) & \geq a^n_{\lambda, \rho}(\kappa_1,\bd) + b^{n}_{\lambda, \rho}(0) \notag \\
& \geq  a^n_{\lambda, \rho}(\kappa_1,\bd) \label{Eq_gn_w1_lowr}
\end{align}
by~\eqref{Eq_gn_lowr} and \edit{because} $b^{n}_{\lambda, \rho}(0)  =0$.
Consequently,
\begin{align*}
\min_{\bd \in \cB^n(v_n)}  \min_{1 \leq \kappa_1 \leq v_n} g^n_{\lambda, \rho}(\kappa_1, 0,\bd) & \geq  \min_{\bd \in \cB^n(v_n)}  \min_{1 \leq \kappa_1 \leq v_n} a^n_{\lambda, \rho}(\kappa_1,\bd)
\end{align*}
\edit{and}~\eqref{Eq_w1_lowr} follows by showing that
\begin{align}
\liminf_{n \rightarrow \infty}  \min_{\bd \in \cB^n(v_n)}  \min_{1 \leq \kappa_1 \leq v_n} a^n_{\lambda, \rho}(\kappa_1,\bd) > 0. \label{Eq_first_lowr}
\end{align}
To this end, let
\begin{align*}
i_n(\kappa_1) & \triangleq \frac{n''}{4\tilde{E}_n} \log \left(1 + \lambda \rho (1-\lambda \rho) \kappa_1 \tilde{E}_n/n''\right) \\
j_n(\kappa_1,\bd) & \triangleq \frac{|\bd|}{\tilde{E}_n} H_2\left(\frac{\kappa_1}{|\bd|}\right)
\end{align*}
so that
\begin{align}
a^n_{\lambda, \rho}(\kappa_1,\bd) =  i_n(\kappa_1)   \left(1- \frac{j_n(\kappa_1,\bd)}{i_n(\kappa_1)}\right). \label{Eq_an}
\end{align}
Note that
\begin{equation}
i_n(\kappa_1) \geq  \frac{n''}{4\tilde{E}_n} \log \left(1 + \lambda \rho (1-\lambda \rho) \tilde{E}_n/n''\right), \quad 1 \leq \kappa_1 \leq v_n \label{Eq_in_lowr}
\end{equation}
and
\begin{align}
\frac{j_n(\kappa_1,\bd)}{i_n(\kappa_1)} & = \frac{4 |\bd| H_2\left(\frac{\kappa_1}{|\bd|}\right)}{n''  \log \left(1 + \lambda \rho (1-\lambda \rho) \kappa_1 \tilde{E}_n/n''\right) } \notag \\
& = \frac{4 \kappa_1 \log (|\bd|/\kappa_1) + 4 |\bd|(\kappa_1/|\bd| - 1)   \log (1 -\kappa_1/|\bd|)   }{n''  \log \left(1 + \lambda \rho (1-\lambda \rho) \kappa_1 \tilde{E}_n/n''\right) }  \label{Eq_jn_in_ratio}.
\end{align}
Next, we upper-bound $(\kappa_1/|\bd| - 1)   \log (1 -\kappa_1/|\bd|) $. To this end, we note that the function $f(p) = p - (p-1) \ln (1-p)$, $0 \leq p \leq 1$ satisfies $f(0)=0$ and is monotonically increasing in $p$. \edit{It follows that} $(p-1) \ln (1-p) \leq p$, $0\leq p \leq 1$, which for $ p=\kappa_1/|\bd|$ gives 
\begin{align}
(\kappa_1/|\bd| - 1)   \log (1 -\kappa_1/|\bd|) \leq (\log e) \kappa_1/|\bd|. \label{Eq_fn_uppr}
\end{align}
Using~\eqref{Eq_fn_uppr} in~\eqref{Eq_jn_in_ratio}, we obtain that
\begin{align}
\frac{j_n(\kappa_1,\bd)}{i_n(\kappa_1)}  
& \leq \frac{ 4 \log (|\bd|/\kappa_1) +  4 \log e}{n''   \log \left(1 + \lambda \rho (1-\lambda \rho)  \kappa_1 \tilde{E}_n/n''\right)/\kappa_1 } \notag \\
& \leq \frac{4 \log (|\bd|/\kappa_1) +4 \log e }{n''   \log \left(1 + \lambda \rho (1-\lambda \rho) v_n \tilde{E}_n/n''\right)/v_n} \notag  \\
&  \edit{\leq  \frac{v_n (4 \log (v_n) +4 \log e) }{n''   \log \left(1 + \lambda \rho (1-\lambda \rho) v_n \tilde{E}_n/n''\right)}} \label{Eq_new1}  \\
& = \frac{4 \log v_n +  4 \log e }{\tilde{E}_n \frac{ \log \left(1 + \lambda \rho (1-\lambda \rho) v_n \tilde{E}_n/n''\right)}{v_n\tilde{E}_n/n''}} \label{Eq_ratio_uppr4}  
\end{align}
where the second inequality follows because  $\frac{\log (1+x)}{x}$ is monotonically decreasing in $x > 0$, and the subsequent inequality follows because $|\bd| \leq v_n$ and $1 \leq \kappa_1 \leq v_n$. Combining~\eqref{Eq_an}, \eqref{Eq_in_lowr}, and~\eqref{Eq_ratio_uppr4}, $ a^n_{\lambda, \rho}(\kappa_1,\bd)$ can thus be lower-bounded by
\begin{align}
a^n_{\lambda, \rho}(\kappa_1,\bd) &\geq   \frac{n''}{4\tilde{E}_n} \log \left(1 + \lambda \rho (1-\lambda \rho) \tilde{E}_n/n''\right)\left(1- \frac{4 \log v_n + 4 \log e }{\tilde{E}_n \frac{ \log \left(1 + \lambda \rho (1-\lambda \rho) v_n \tilde{E}_n/n''\right)}{v_n\tilde{E}_n/n''}}  \right). \label{Eq_an_lowr}
\end{align}
Note that the RHS of~\eqref{Eq_an_lowr} is independent of $\kappa_1$ and $\bd$.
\edit{To show~\eqref{Eq_first_lowr}, we consider the following two cases:

\subsubsection*{Case 1---$k_n \log \ell_n = o(n)$} Recall that if $k_n \log \ell_n = o(n)$, then $ \tilde{E}_n = bc_n \ln \ell_n/2$. 
}
 It follows that the term 
 \begin{align}
\frac{ v_n \tilde{E}_n}{n''} & = \frac{ (1+c'')}{2} c_n \frac{k_n \ln \el_n}{n}  \label{Eq_En_sn_zero2} 
\end{align}
tends to zero as $n \to \infty$ since $c_n = \ln \left(\frac{n}{k_n \ln \ell_n}\right)$ and $\frac{k_n \ln \el_n}{n} = o(1)$ by assumption. \edit{This also implies that $\tilde{E}_n/n'' \to 0$ as $n \to \infty$ since $v_n =\Omega(1)$.} We further have that $\tilde{E}_n\to\infty$ and \edit{$\frac{\log v_n}{\tilde{E}_n} \to 0$ as $n \to \infty$ since $v_n = \Theta(k_n)$ and $\tilde{E}_n = \omega(\log \ell_n)$.} It follows that
\begin{align}
& \liminf_{n \to \infty}   \min_{\bd \in \cB^n(v_n)}   \min_{1 \leq \kappa_1 \leq v_n} a^n_{\lambda, \rho}(\kappa_1,\bd) \notag \\
& \qquad \qquad \geq \lim_{n \to \infty} \frac{n''}{4\tilde{E}_n} \log \left(1 + \lambda \rho (1-\lambda \rho) \tilde{E}_n/n''\right)   \lim_{n \to \infty} \left(1- \frac{ 4\log v_n +  4\log e }{\tilde{E}_n \frac{ \log \left(1 + \lambda \rho (1-\lambda \rho) v_n \tilde{E}_n/n''\right)}{v_n\tilde{E}_n/n''}}  \right) \notag \\
& \qquad \qquad= \frac{(\log e) \;  \lambda \rho (1-\lambda \rho)}{4} \label{Eq_new7} 
\end{align}
which for $\lambda=2/3$ and $\rho = 3/4$ is equal to $(\log e)/16$.

\edit{\subsubsection*{Case 2---$k_n\log \ell_n = \Theta(n)$} To analyze this case, we first note that, by the definition of $v_n$, and because $\log k_n\leq \log\ell_n$, the numerator in \eqref{Eq_new1} satisfies $v_n (4 \log (v_n) +4 \log e) = O(k_n\log \ell_n)$. Since $k_n\log \ell_n = \Theta(n)$, and since $n'' =\Theta(n)$, this further implies that there exist $a_2>0$ and $\tilde{n}'_0>0$ such that 
\begin{align}
\frac{v_n (4 \log (v_n) +4 \log e) }{n''} \leq a_2, \quad n \geq \tilde{n}'_0. \label{Eq_new8}
\end{align}
We have
\begin{align}
a^n_{\lambda, \rho}(\kappa_1,\bd) &\geq   \frac{n''}{4\tilde{E}_n} \log \left(1 + \lambda \rho (1-\lambda \rho) \tilde{E}_n/n''\right)\left(1- \frac{v_n (4 \log (v_n) +4 \log e) }{n''   \log \left(1 + \lambda \rho (1-\lambda \rho) v_n \tilde{E}_n/n''\right)} \right). \label{Eq_new9}
\end{align}
The RHS of \eqref{Eq_new9} is independent of $\kappa_1$ and $\bd$. We next show that it is bounded away from zero.

Recall that, if $k_n\log\ell_n = \Theta(n)$, then $\tilde{E}_n = b c' \ln\ell_n /2$. We then choose $c'$ sufficiently large such that, for some $\tilde{n}''_0\geq \tilde{n}_0'$,
\begin{equation}
\label{eq:Lemma14_make_a_choice}
\log \left(1 + \lambda \rho (1-\lambda \rho) v_n \tilde{E}_n/n''\right) >a_2, \quad n \geq \tilde{n}''_0.
\end{equation}
Such a choice is possible since we have
\begin{equation*}
\frac{v_n \tilde{E}_n}{n''} = \frac{1+c''}{2} c' \frac{k_n\ln \ell_n}{n}
\end{equation*}
and, by assumption, $\frac{k_n\ln\ell_n}{n}=\Theta(1)$. Consequently, for sufficiently large $n$, $v_n\tilde{E}_n/n''$ is monotonically increasing in $c'$ and ranges from zero to infinity.
Combining \eqref{eq:Lemma14_make_a_choice} with \eqref{Eq_new8} implies that the expression inside the large parentheses on the RHS of \eqref{Eq_new9} is bounded away from zero for $n \geq \tilde{n}''_0$.

We next consider the remaining term on the RHS of \eqref{Eq_new9}. To this end, we note that, if $k_n$ is unbounded, then the assumption $k_n\log\ell_n=\Theta(n)$ implies that $\ln\ell_n=o(n)$. It follows that $\frac{\tilde{E}_n}{n''}=c'\frac{\ln\ell_n}{2n}\to 0$ as $n\to\infty$.  If $k_n$ is bounded, then $\ln\ell_n=\Theta(n)$, so $\frac{\tilde{E}_n}{n''}$ is bounded. In both cases, $\frac{n''}{4\tilde{E}_n} \log(1 + \lambda \rho (1-\lambda \rho) \tilde{E}_n/n'')$ tends to a positive value as $n \to \infty$.

Applying the above lines of argument to \eqref{Eq_new9}, we obtain that
\begin{align}
 \liminf_{n \to \infty}   \min_{\bd \in \cB^n(v_n)}   \min_{1 \leq \kappa_1 \leq v_n} a^n_{\lambda, \rho}(\kappa_1,\bd) >0 \label{Eq_new3}
\end{align}
which concludes the analysis of the second case.

The claim \eqref{Eq_w1_lowr} follows now by combining the above two cases, i.e., \eqref{Eq_new7} and~\eqref{Eq_new3}.
}

\subsubsection{Proof of~\eqref{Eq_w2_lowr}}
Since $a^{n}_{\lambda, \rho}(0,\bd)  =0 $, we have that
\begin{align*}
\min_{\bd \in \cB^n(v_n)}  \min_{1 \leq \kappa_2 \leq v_n} g^n_{\lambda, \rho}(0, \kappa_2,\bd) & \geq \min_{1 \leq \kappa_2 \leq v_n} b^n_{\lambda, \rho}(\kappa_2).
\end{align*}
Thus, \eqref{Eq_w2_lowr} follows by showing that
\begin{align}
\liminf_{n \rightarrow \infty} \min_{1 \leq \kappa_2 \leq v_n} b^n_{\lambda, \rho}(\kappa_2) > 0. \label{Eq_w2_lowr1}
\end{align}
To prove \eqref{Eq_w2_lowr1}, we define
\begin{align}
q_n(\kappa_2) & \triangleq \frac{n''}{4\tilde{E}_n} \log \left(1+ \lambda(1-\lambda \rho)\kappa_2 \tilde{E}_n/n'' \right) \label{Eq_qn} \\
r_n(\kappa_2) & \triangleq \frac{(1-\rho)}{2 \tilde{E}_n} \log (1+\lambda \kappa_2 \tilde{E}_n/n'')  \\
u_n(\kappa_2) & \triangleq \frac{\rho \el_n}{\tilde{E}_n} H_2\left(\frac{\kappa_2}{\el_n}\right). \label{Eq_sn}
\end{align}
Then, 
\begin{align}
b^n_{\lambda, \rho}(\kappa_1) & =  q_n(\kappa_2)  \left(1- \frac{r_n(\kappa_2)}{q_n(\kappa_2)} - \frac{u_n(\kappa_2)}{q_n(\kappa_2)}\right). \label{Eq_bn_lowr}
\end{align}
Note that
\begin{equation}
q_n(\kappa_2) \geq  \frac{n''}{4\tilde{E}_n} \log \left(1+ \lambda(1-\lambda \rho) \tilde{E}_n/n'' \right), \quad 1 \leq \kappa_2 \leq v_n. \label{Eq_qn_lowr}
\end{equation}
Furthermore,
\begin{align}
\frac{r_n(\kappa_2)}{ q_n(\kappa_2)} &= \frac{\frac{(1-\rho)}{2 \tilde{E}_n} \log (1+\lambda \kappa_2 \tilde{E}_n/n'')}{\frac{n''}{4\tilde{E}_n} \log \left(1+ \lambda(1-\lambda \rho)\kappa_2 \tilde{E}_n/n'' \right)} \notag \\
& \leq  \frac{\frac{(1-\rho)}{2 \tilde{E}_n} \log (1+\lambda v_n \tilde{E}_n/n'')}{\frac{n''}{4\tilde{E}_n} \log \left(1+ \lambda(1-\lambda \rho) \tilde{E}_n/n'' \right)} \notag \\
& =  \frac{\frac{(1-\rho)v_n}{2 n'' } \frac{ \log (1+\lambda v_n \tilde{E}_n/n'')}{\tilde{E}_nv_n/n''}}{ \frac{ \log \left(1+ \lambda(1-\lambda \rho) \tilde{E}_n/n'' \right)}{4\tilde{E}_n/n''}}, \quad 1 \leq \kappa_2 \leq v_n. \label{Eq_rn_qn_ratio}
\end{align}
Finally,
\begin{align}
\frac{u_n(\kappa_2)}{ q_n(\kappa_2)}  & = \frac{4 \rho \el_n H_2\left(\frac{\kappa_2}{\el_n}\right)}{n'' \log \left(1+ \lambda(1-\lambda \rho)\kappa_2 \tilde{E}_n/n'' \right)} \notag \\
& = \frac{4 \rho \left[ \kappa_2 \log (\el_n/\kappa_2) + \el_n (\kappa_2/\el_n -1)   \log (1 -\kappa_2/\el_n) \right]  }{n'' \log \left(1+ \lambda(1-\lambda \rho)\kappa_2 \tilde{E}_n/n'' \right)}\notag\\
& \leq \frac{4 \rho \left[ \kappa_2 \log (\el_n/\kappa_2) + \kappa_2 \log e \right]  }{n'' \log \left(1+ \lambda(1-\lambda \rho)\kappa_2 \tilde{E}_n/n'' \right)} \notag\\
& \leq \edit{ \frac{4 \rho v_n \left[  \log \ell_n +  \log e \right]  }{n'' \log \left(1+ \lambda(1-\lambda \rho)v_n \tilde{E}_n/n'' \right)} }\label{Eq_new2} \\
& =  \frac{4 \rho \left[  \log \el_n +  \log e \right]  }{\tilde{E}_n \frac{ \log \left(1+ \lambda(1-\lambda \rho)v_n \tilde{E}_n/n'' \right)}{v_n\tilde{E}_n/n''}}, \quad 1 \leq \kappa_2 \leq v_n \label{Eq_sn_qn_uppr}
\end{align}
where the first inequality follows from~\eqref{Eq_fn_uppr}. 
Combining~\eqref{Eq_qn_lowr}--\eqref{Eq_sn_qn_uppr} with~\eqref{Eq_bn_lowr} yields the lower bound
\begin{align}
b^n_{\lambda, \rho}(\kappa_2) & \geq \frac{n''}{4\tilde{E}_n} \log \left(1+ \lambda(1-\lambda \rho) \tilde{E}_n/n'' \right)\left(1 - \frac{\frac{(1-\rho)v_n}{2 n'' } \frac{ \log (1+\lambda v_n \tilde{E}_n/n'')}{\tilde{E}_nv_n/n''}}{ \frac{ \log \left(1+ \lambda(1-\lambda \rho) \tilde{E}_n/n'' \right)}{4\tilde{E}_n/n''}}  - \frac{4 \rho \left[  \log \el_n +  \log e \right]  }{\tilde{E}_n \frac{ \log \left(1+ \lambda(1-\lambda \rho)v_n \tilde{E}_n/n'' \right)}{v_n\tilde{E}_n/n''}} \right) \label{Eq_bn_lowr_bnd}
\end{align}
which is independent of $\kappa_2$ and $\bd$.

\edit{To prove \eqref{Eq_w2_lowr1}, we first note that, since $k_n = \Omega(1)$, we have $v_n \geq 1 - \lambda\rho$ for $c''$ sufficiently large. Thus, by the monotonicity of $x\mapsto\frac{\log (1+x)}{x}$,
	\begin{equation}
	\frac{ \frac{ \log (1+\lambda v_n \tilde{E}_n/n'')}{\tilde{E}_nv_n/n''}}{ \frac{ \log \left(1+ \lambda(1-\lambda \rho) \tilde{E}_n/n'' \right)}{\tilde{E}_n/n''}} \leq \frac{1}{1-\lambda \rho}, \quad v_n \geq 1-\lambda\rho. \label{Eq_new12}
	\end{equation}
 Furthermore, the term 
\begin{align}
\frac{2(1-\rho)v_n}{n''} & = \frac{2(1-\rho) k_n(1+c'')}{bn} \label{Eq_new11}
\end{align}
vanishes as $n\to\infty$ since $k_n = o(n)$ by the lemma's assumption that $k_n \log \el_n = O(n)$. Consequently, the RHS of \eqref{Eq_rn_qn_ratio} tends to zero as $n\to\infty$. It follows that
\begin{equation}
\liminf_{n \rightarrow \infty}   \min_{1 \leq \kappa_2 \leq v_n} b^n_{\lambda, \rho}(\kappa_2) \geq \lim_{n \to \infty}   \frac{n''}{4\tilde{E}_n} \log \left(1+ \lambda(1-\lambda \rho) \tilde{E}_n/n'' \right) \left(1-\limsup_{n\to\infty} \frac{4 \rho \left[  \log \el_n +  \log e \right]  }{\tilde{E}_n \frac{ \log \left(1+ \lambda(1-\lambda \rho)v_n \tilde{E}_n/n'' \right)}{v_n\tilde{E}_n/n''}}\right). \label{Eq_new10} 
\end{equation}
To show that the RHS of \eqref{Eq_new10} is positive, we consider the following two cases:}

	\subsubsection*{Case 1---$k_n \log \ell_n = o(n)$} Recall that, in this case, $\tilde{E}_n=b c_n \ln \ell_n /2$ and $c_n\to\infty$ as $n\to\infty$. This implies that $\tilde{E}_n \to \infty$, $\frac{\log \el_n}{ \tilde{E}_n} \to 0$, and $v_n\tilde{E}_n/n'' \to 0$ as $n \to \infty$. It follows that the RHS of \eqref{Eq_sn_qn_uppr} vanishes as $n \to \infty$, so \eqref{Eq_new10} becomes
\begin{align}
\liminf_{n \rightarrow \infty}   \min_{1 \leq \kappa_2 \leq v_n} b^n_{\lambda, \rho}(\kappa_2) &  \geq \lim_{n \to \infty}   \frac{n''}{4\tilde{E}_n} \log \left(1+ \lambda(1-\lambda \rho) \tilde{E}_n/n'' \right) \notag \\
&  =  \frac{(\log e)  \; \lambda(1-\lambda \rho)}{4} \label{Eq_new10_b}
\end{align}
which for $\lambda= 2/3$ and $\rho =3/4$ is equal to $(\log e)/12$. Here, the last step follows by noting that $v_n \tilde{E}_n/n''\to 0$ implies that $\tilde{E}_n/n''\to 0$ since $v_n=\Omega(1)$.

\edit{
\subsubsection*{Case 2---$k_n \log \ell_n = \Theta(n)$} We first note that, in this case, $4 \rho v_n \left[  \log \ell_n +  \log e \right]  = \Theta(n)$. Thus, there exist two positive constants $a_3$ and $\tilde{n}'''_0$ such that 
	\begin{align*}
	\frac{4 \rho v_n \left[  \log \ell_n +  \log e \right] }{n''} & \leq a_3, \quad n \geq \tilde{n}'''_0
	\end{align*}
	By the same arguments that demonstrate \eqref{eq:Lemma14_make_a_choice}, we can show that $c'$ can be chosen sufficiently large so that, for some $\bar{n}_0\geq \tilde{n}'''_0$,
	\begin{equation*}
	\log \left(1 + \lambda \rho (1-\lambda \rho) v_n \tilde{E}_n/n''\right) >a_3, \quad n \geq \bar{n}_0.
	\end{equation*}
	For such a $c'$, the RHS of \eqref{Eq_new2} is strictly less than one. Consequently, the expression inside the large parentheses on the RHS of \eqref{Eq_new10} is bounded away from zero for $n\geq\bar{n}_0$. Furthermore, as noted in the proof of \eqref{Eq_w1_lowr}, when $k_n\log\ell_n=\Theta(n)$, the expression $\frac{n''}{4\tilde{E}_n} \log(1 + \lambda \rho (1-\lambda \rho) \tilde{E}_n/n'')$ tends to a positive value as $n \to \infty$. It follows that
	\begin{align}
	\liminf_{n \rightarrow \infty}   \min_{1 \leq \kappa_2 \leq v_n} b^n_{\lambda, \rho}(\kappa_2) & > 0 \label{Eq_new4}
	\end{align}
	which concludes the analysis of the second case.
	
	The claim \eqref{Eq_w2_lowr} follows now by combining the above two cases, i.e., \eqref{Eq_new10_b} and \eqref{Eq_new4}.
}

\subsubsection{Proof of~\eqref{Eq_w1w2_lowr}} We use~\eqref{Eq_gn_lowr}, \eqref{Eq_an_lowr}, and \eqref{Eq_bn_lowr_bnd} to lower-bound
\begin{align}
g^n_{\lambda, \rho}(\kappa_1,\kappa_2, \bd)  & \geq  \frac{n''}{4\tilde{E}_n} \log \left(1 + \lambda \rho (1-\lambda \rho) \tilde{E}_n/n''\right)\left(1- \frac{ 4\log v_n +  4\log e }{\tilde{E}_n \frac{ \log \left(1 + \lambda \rho (1-\lambda \rho) v_n \tilde{E}_n/n''\right)}{v_n\tilde{E}_n/n''}}  \right) \notag \\
& + \frac{n''}{4\tilde{E}_n} \log \left(1+ \lambda(1-\lambda \rho) \tilde{E}_n/n'' \right)\left(1 - \frac{\frac{(1-\rho)v_n}{2 n'' } \frac{ \log (1+\lambda v_n \tilde{E}_n/n'')}{\tilde{E}_nv_n/n''}}{ \frac{ \log \left(1+ \lambda(1-\lambda \rho) \tilde{E}_n/n'' \right)}{4\tilde{E}_n/n''}}  - \frac{4 \rho \left[  \log \el_n +  \log e \right]  }{\tilde{E}_n \frac{ \log \left(1+ \lambda(1-\lambda \rho)v_n \tilde{E}_n/n'' \right)}{v_n\tilde{E}_n/n''}} \right) \notag
\end{align}
which is independent of $\kappa_1, \kappa_2$, and $\bd$. \edit{It follows that
\begin{equation*}
\liminf_{n \to \infty}  \min_{\bd \in \cB^n(v_n)}  \min_{ \substack{ 1 \leq \kappa_1 \leq v_n \\ 1 \leq \kappa_2 \leq v_n}} g^n_{\lambda, \rho}(\kappa_1,\kappa_2, \bd) \geq \underline{a} + \underline{b}
\end{equation*}
where
\begin{equation*}
\underline{a} \triangleq \liminf_{n \to \infty}\left\{ \frac{n''}{4\tilde{E}_n} \log \left(1 + \lambda \rho (1-\lambda \rho) \tilde{E}_n/n''\right)\left(1- \frac{ 4\log v_n +  4\log e }{\tilde{E}_n \frac{ \log \left(1 + \lambda \rho (1-\lambda \rho) v_n \tilde{E}_n/n''\right)}{v_n\tilde{E}_n/n''}}  \right)\right\}
\end{equation*}
and
\begin{equation*}
\underline{b} \triangleq \liminf_{n \to \infty} \left\{\frac{n''}{4\tilde{E}_n} \log \left(1+ \lambda(1-\lambda \rho) \tilde{E}_n/n'' \right)\left(1 - \frac{\frac{(1-\rho)v_n}{2 n'' } \frac{ \log (1+\lambda v_n \tilde{E}_n/n'')}{\tilde{E}_nv_n/n''}}{ \frac{ \log \left(1+ \lambda(1-\lambda \rho) \tilde{E}_n/n'' \right)}{4\tilde{E}_n/n''}}  - \frac{4 \rho \left[  \log \el_n +  \log e \right]  }{\tilde{E}_n \frac{ \log \left(1+ \lambda(1-\lambda \rho)v_n \tilde{E}_n/n'' \right)}{v_n\tilde{E}_n/n''}} \right)\right\}.
\end{equation*}
We then obtain \eqref{Eq_w1w2_lowr} by noting that it was shown in the proof of \eqref{Eq_w1_lowr} that $\underline{a}>0$, and in the proof of \eqref{Eq_w2_lowr} that $\underline{b}>0$.

Since \eqref{Eq_w1_lowr}--\eqref{Eq_w1w2_lowr} prove Lemma~\ref{Lem_err_exp}, this concludes the proof.
}

\section{Proof of Lemma~\ref{Lem_energy_bound}}
\label{Append_prob_lemma}

\comment{
	To prove Lemma~\ref{Lem_energy_bound}, we represent the messages $(W_1, \ldots, W_{\el_n})$  using an $\ell_n$-length vector such that the $i$-th position of the vector is set to $j,j=0,1,\ldots, M_n$ if user $i$ has message $j$. We further partition the set of all message vectors into $\ell_n +1$ groups depending on the number of zeros in the vector. We then compute the average probability of error for each group except for the group corresponds to all zero vector whose probability vanishes in all non-trivial cases. We then find a lower bound to the probability of error for each group. This lower bound is obtained by further grouping the vectors in each group into subgroups such that any two vector in a subgroup are at a maximum Hamming distance two. Then the maximum difference in the energy of two vectors in a subgroup is a constant multiple of the  energy allowed for a single user. By using a well-known ineqaulity called Birg\'e's inequality, we find a lower bound on the probability of error for each of these subgroups. This then yields a lower bound on the probability of error which proves the lemma.
}

Let $\cW$ denote the set of  the $(M_n+1)^{\el_n}$ messages  of all users. To prove Lemma~\ref{Lem_energy_bound}, we represent each $\bw \in \cW$   using an \edit{length-$\ell_n$} vector such that the $i$-th position of the vector is set to $j$ if user $i$ has message $j$. 
The Hamming distance $d_H$ between two messages $\bw=(w_1,\ldots,w_{\el_n})$ and $\bw'=(w'_1,\ldots,w'_{\el_n})$ is defined as the number of positions at which $\bw$ differs from $\bw'$, i.e.,
$d_H(\bw,\bw') \triangleq \left|\{i: w_i\neq w'_i \}\right|$.

We first group the set $\cW$ into $\el_n +1$ subgroups. 
Two messages $\bw, \bw' \in \cW$ belong to the same subgroup if they have the same number of zeros. Note that all the messages in a subgroup have the same probability, since the probability of a message $\bw$ is determined by the number of zeros in it.

Let $\MS$ denote the set of  messages $\bw \in \cW$ with $t$ non-zero entries, where $t=0, \ldots, \el_n$. 
Further let 
\begin{align}
\text{Pr}(\MS) \triangleq  \text{Pr}(\bW \in \MS)\notag
\end{align}
which can be evaluated as
\begin{equation}
\text{Pr}(\MS) = (1-\alpha_n)^{\el_n-t} \left( \frac{\alpha_n}{M_n}\right)^{t} |\MS|. \label{Eq_type_prob}
\end{equation}
We define
\begin{align}
P_e(\MS) \triangleq \frac{1}{|\MS|} \sum_{\bw\in \MS} P_e(\bw) \label{Eq_type_err_prob}
\end{align}
where $P_e(\bw)$ denotes the probability of error in decoding the set of messages $\bw=(w_1,\ldots,w_{\el_n})$.
It follows that
\begin{align}
P_{e}^{(n)} & =  \sum_{\bw \in \cW}  \text{Pr}(\bW =\bw) P_e(\bw) \notag \\
& =  \sum_{t=0}^{\el_n} \sum_{\bw\in \MS} (1-\alpha_n)^{\el_n-t} \left( \frac{\alpha_n}{M_n}\right)^{t} |\MS| \frac{1}{|\MS|}  P_e(\bw)\notag \\
& =  \sum_{t=0}^{\el_n}  \text{Pr}( \MS) \frac{1}{|\MS|} \sum_{w\in \MS} P_e(\bw) \notag \\
& =  \sum_{t=0}^{\el_n}  \text{Pr}(\MS) P_e(\MS) \notag \\
& \geq \sum_{t=1}^{\el_n}  \text{Pr}(\MS) P_e(\MS) \label{Eq_avg_prob_err3}
\end{align}
where we have used \eqref{Eq_type_prob} and the definition of $P_e(\MS)$ in \eqref{Eq_type_err_prob}.
To prove Lemma~\ref{Lem_energy_bound}, we next show that 
\begin{align}
P_e(\MS) \geq  1  -   \frac{ 256 E_n/N_0+\log 2}{\log \el_n}, \quad t=1,\ldots, \el_n. \label{Eq_prob_typ_lowr}
\end{align}
To this end, we partition each $\MS, t=1,\ldots, \el_n  $ into $D_t$ sets $\cS_d^t$.  For every $ t=1,\ldots, \el_n $, \edit{we then show that this partition satisfies}
\begin{align}
\frac{1}{|\cS_d^t|} \sum_{\bw \in \cS_d^t} P_e(\bw) \geq 1  -   \frac{ 256 E_n/N_0+\log 2}{\log \el_n}. \label{Eq_avg_prob_lowr}
\end{align}
This yields~\eqref{Eq_prob_typ_lowr} since
\begin{align}
P_e(\MS) & = \sum_{d=1}^{D_t} \frac{|\cS_d^t|}{|\MS|} \frac{1}{|\cS_d^t|} \sum_{\bw \in \cS_d^t} P_e(\bw).
\end{align}

Before we continue by defining the sets $\mathcal{S}_d^t$, we note that
\begin{align}
M_n \geq  2 \label{Eq_seqMn_assum}
\end{align}
since $M_n=1$ would contradict the \edit{lemma's first assumption}  that $\CR>0$.
We further have that
\begin{align}
\el_n \geq 5 \label{Eq_seqln_assum}
\end{align}
\edit{by the lemma's second assumption.}

We next define a partition of $\MS,  t=1,\ldots, \el_n  $  that satisfies the following:
\begin{align}
|\cS_d^t| \geq \el_n +1, \quad  d=1, \ldots, D_t \label{Eq_set_lowr}
\end{align}
and
\begin{align}
d_{H}(\bw,\bw') \leq 8, \quad \bw, \bw' \in  \cS_d^t. \label{Eq_distnce_uppr}
\end{align}
To this end, we consider the following four cases:

\subsubsection*{Case 1---$t=1$} For $t=1$, we do not partition the set, i.e., $\cS_1^1 = \cT_ {1}$. Thus, we have $|\cS_1^1| = \el_nM_n$. From~\eqref{Eq_seqMn_assum} and~\eqref{Eq_seqln_assum}, it follows that $|\cS_1^1| \geq \el_n +1$. Since any two messages $\bw, \bw' \in  \cT_ {1}$ 
have only one non-zero entry, we further have that $d_{H}(\bw,\bw') \leq 2$. Consequently, \eqref{Eq_set_lowr} and~\eqref{Eq_distnce_uppr} are satisfied.

\subsubsection*{Case 2---$t=2,\ldots,\el_n -2$}
In this case, we obtain a partition by finding a code $\cC_t$ in $\cT_ {t}$  that has minimum Hamming distance $5$, and for every $\bw\in \MS$, there exists at least one codeword in $\cC_t$ which is at most at Hamming distance 4 from it. 
Such a code exists because, if for some $\bw\in \MS$ all codewords were at Hamming distance 5 or more, then we could add $\bw$ to $\cC_t$ without affecting its minimum distance. 
Thus, for all $\bw \notin \cC_t$, there exists at least one index $j$ such that $d_H(\bw,\bc_t(j)) \leq 4$, where
$\bc_t(1),\ldots, \bc_t(|\cC_t|)$ denote the codewords of \edit{the} code $\cC_t$. With this code $\mathcal{C}_t$, we partition $\cT_t$ into the sets $\mathcal{S}_d^t$, $d=1,\ldots,D_t$ \edit{with $D_t = |\cC_t|$} the following procedure:
\begin{enumerate}
	\item For a given $d=1, \ldots,D_t$, we assign $\bc_t(d)$ to $\cS_d^t$ as well as all $\bw \in \MS$ that satisfy $d_H(\bw, \bc_t(d))\leq 2$. These assignments are unique since the code $\cC_t$ has minimum Hamming distance 5.
	\item  We then consider  all $\bw\in  \MS$ for which there is no codeword $\bc_t(1), \ldots, \bc_t(|\cC_t|)$ satisfying $d_H(\bw, \bc_t(d))\leq  2$ and assign them to the set \edit{$\cS_d^t$} with index $d = \min \{j=1,\ldots, D_t: d_H(\bw, \bc_t(j)) \leq 4 \}$. 
\end{enumerate}
Like this, we obtain a partition of $ \MS$.  
Since any two $\bw, \bw' \in  \cS_d^t$ are at most at a Hamming distance 4 from the codeword $\bc_t(d)$, we have that
$d_{H}(\bw,\bw') \leq 8$. Consequently, \eqref{Eq_distnce_uppr} is satisfied.

To show that~\eqref{Eq_set_lowr} is satisfied, too, we use the following fact:
\begin{align}
\text{For two natural numbers } a \text{ and } b, \text { if } a \geq 4 \text{ and } 2\leq b \leq a-2, \text{ then } b(a-b) \geq a. \label{Eq_prod_seq}
\end{align}
This fact follows since $b(a-b)$ is increasing \edit{in $b$} from $b=2$ to $b = \lfloor a/2\rfloor$ and is 
decreasing \edit{in $b$} from $b = \lfloor a/2\rfloor$ to $b=a-2$. So $b(a-b)$ is minimized at $b=2$ and $b=a-2$, where it has the value $2a-4$. For $a\geq 4$, this value is greater than or equal to $a$, hence the claim follows.

From~\eqref{Eq_prod_seq}, it follows that, if   $|\cS_d^t| \geq 1+ t(\el_n - t)$, then $|\cS_d^t|\geq 1+ \el_n$. It thus remains to show that  $|\cS_d^t| \geq 1+  t(\el_n - t)$.
To this end, for every codeword $\bc_t(d)$, consider all sequences in $\MS$ which differ exactly in one non-zero position and in one zero position from $\bc_t(d)$. There are $ t(\el_n - t)M_n$ such  sequences  in $\MS$, which can be lower-bounded as
\begin{align}
t(\el_n - t)M_n & \geq  t(\el_n - t) \notag \\
& \geq \el_n  \label{Eq_set_lowr2}
\end{align}
by~\eqref{Eq_seqMn_assum}, \eqref{Eq_seqln_assum}, and~\eqref{Eq_prod_seq}.
Since the codeword $\bc_t(d)$ also belongs to $S_d^t$, it follows from~\eqref{Eq_set_lowr2} that 
\begin{align}
| \cS_d^t| &\geq \el_n +1. \notag
\end{align}

\subsubsection*{Case 3---$t=\el_n -1$}
We obtain a partition by defining a code $\cC_t$ in $ \cT_{\el_n -1}$ 
that has 
the same properties as the code used for Case 2. We then use the same procedure as in Case 2 to assign messages in $\bw \in\cT_{\el_n -1}$ to the sets $\cS_d^t$, $d=1,\ldots,D_t$. This gives a partition of $\cT_{\el_n -1}$ where any two $\bw, \bw'\in\cS_d^t$ satisfy $d_H(\mathbf{w},\mathbf{w}')\leq 8$. Consequently, this partition satisfies~\eqref{Eq_distnce_uppr}. 

We next show that this partition also satisfies~\eqref{Eq_set_lowr}. To this end, for every codeword $\bc_t(d)$, consider all the sequences which differ exactly in two non-zero positions from $\bc_t(d)$.
There are $ \binom{\el_n-1}{2} (M_n-1)^2$  such sequences in $ \cT_{\el_n -1}$. Since $\cS_d^t$ also contains the codeword $\bc_t(d)$, we obtain that
\begin{align*}
| \cS_d^t|  & \geq \notag
\binom{\el_n-1}{2} (M_n-1)^2 + 1\\
& \geq \binom{\el_n-1}{2} +1   \\
& \geq  \el_n +1 
\end{align*}  
by~\eqref{Eq_seqMn_assum} and~\eqref{Eq_seqln_assum}.

\subsubsection*{Case 4---$t=\el_n$} We obtain a partition by defining a code $\mathcal{C}_t$ in $\cT_{\el_n-1}$ that has the same properties as the code used in Case 2. We then use the same procedure as in Case 2 to assign messages in $\mathbf{w}\in\cT_t$ to the sets $\mathcal{S}_d^t$, $d=1,\ldots,D_t$. This gives a partition of $\cT_t$ where any two $\mathbf{w}, \mathbf{w}'\in\mathcal{S}_d^t$ satisfy $d_H(\mathbf{w},\mathbf{w}')\leq 8$. Consequently, this partition satisfies~\eqref{Eq_distnce_uppr}.

We next show that this partition also satisfies~\eqref{Eq_set_lowr}. To this end, for every codeword $\bc_t(d)$, consider all sequences which are  at Hamming distance $1$ from $\bc_t(d)$. There are $\el_n(M_n-1)$ such sequences. Since $\cS_d^t$ also contains the codeword, we have 
\begin{align}
| \cS_d^t| & \geq 1+ \el_n(M_n-1) \notag \\
& \geq 1+\el_n \notag
\end{align}  
by~\eqref{Eq_seqMn_assum}. 

Having obtained a partition of $\cT_t$ that satisfies~\eqref{Eq_set_lowr} and~\eqref{Eq_distnce_uppr}, we next derive the lower bound~\eqref{Eq_avg_prob_lowr}.	To this end, we use  a stronger form of Fano's inequality known as Birg\'e's inequality.
\begin{lemma}[Birg\'e's inequality]
	\label{Lem_Berge}
	Let $(\cY, \cB)$ be a measurable space with a $\sigma$-field, and let $P_1,\ldots, P_N$  be probability measures defined on $\cB$. Further let $\cA_i$, $i=1, \ldots,N$  denote $N$ events defined on $\cY$, where $N\geq 2$.
	Then 
	\begin{align*}
	\frac{1}{N} \sum_{i=1}^{N} P_i(\cA_i) \leq \frac{\frac{1}{N^2} \sum_{i,j}^{}  D(P_i\|P_j)+\log 2}{\log (N-1)}.
	\end{align*}
\end{lemma}
\begin{proof}
	See~\cite{Yatracos88} and references therein.
\end{proof}

To apply Lemma~\ref{Lem_Berge} to the problem at hand, we set $N=|\cS_d^t|$ and $P_j = P_{\bY|\bX}(\cdot|\bx(j))$, where $\bx(j)$ denotes the set of codewords transmitted to convey the set of messages $j \in \cS_d^t$.
We further define $\cA_j$ as the subset of $\cY^n$ for which the decoder declares the set of messages $j\in\cS_d^t$. Then, the probability of
error in decoding messages $j\in\cS_d^t$ is given by $P_e(j) =1-P_j(\cA_j)$, and $\frac{1}{|\cS_d^t|} \sum_{j\in \cS_d^t} P_j(\cA_j)$ denotes the average probability of correctly decoding a \edit{set of messages} in $\cS_d^t$.

For two multivariate Gaussian distributions \mbox{${\bf Z}_1 \sim \cN(\boldsymbol {\mu_1 }, \frac{N_0}{2}I)$}
and ${\bf Z}_2 \sim \cN(\boldsymbol {\mu_2}, \frac{N_0}{2}I)$ (where $I$ denotes the identity matrix),
the relative entropy $D({\bf Z}_1\| { \bf Z}_2)$ is given  by $ \frac{ ||\boldsymbol {\mu_1 - \mu_2}||^2}{N_0}$. We next note that $P_{\bw} =  \cN(\overline{\bx}(\bw), \frac{N_0}{2}I)$ and $P_{\bw'} = \cN(\overline{\bx}(\bw'), \frac{N_0}{2}I)$, where $\overline{\bx}(j)$ denotes the sum of codewords contained in $\bx(j)$.
\edit{By construction}, any two messages $\bw, \bw' \in \cS_d^t$ are at a Hamming distance of at most 8. Without loss of generality, let us assume that $w_j = w'_j$ for $j=9, \ldots, \el_n$. Then
\begin{align}
\big\|\sum_{j=1}^{\el_n} \bx_j(w_j) -\sum_{i=1}^{\el_n} \bx_j(w'_j)\big\|^2 & = \big\|\sum_{i=1}^{8} \bx_j(w_j) - \bx_j(w'_j)\big\|^2 \notag \\
& \leq \big\|\sum_{j=1}^{8} |\bx_j(w_j) - \bx_j(w'_j)|\big\|^2 \notag\\
& \leq (8 \times 2\sqrt{E_n})^2 \notag \\
& = 256 E_n \notag
\end{align}
where we have used the triangle inequality and that the energy of a codeword for any user is upper-bounded by $E_n$. Thus, $D(P_{\bw}\| P_{\bw'}) \leq 256 E_n/N_0$.

It  follows from Birg\'e's inequality that
\begin{align}
\frac{1}{|\cS_d^t|} \sum_{\bw \in \cS_d^t} P_e(\bw) & \geq 1  -   \frac{ 256 E_n/N_0+\log 2}{\log (|\cS_d^t|-1)} \notag \\
&  \geq 1  -   \frac{ 256 E_n/N_0+\log 2}{\log \el_n } \label{Eq_prob_set_lowr}
\end{align}
where the last step holds because  $ |\cS_d^t|-1 \geq \el_n$. This proves~\eqref{Eq_avg_prob_lowr} and hence also~\eqref{Eq_prob_typ_lowr}.

Combining~\eqref{Eq_prob_typ_lowr} and~\eqref{Eq_avg_prob_err3}, we obtain
\begin{align*}
P_{e}^{(n)} & \geq \left( 1  -   \frac{ 256 E_n/N_0+\log 2}{\log \el_n } \right)  \sum_{i=1}^{\el_n}  \text{Pr}(\cT_{i}) \\
&  = \left( 1  -   \frac{ 256 E_n/N_0+\log 2}{\log \el_n } \right) (1-\text{Pr}(\cT_{0} )). 
\end{align*}
\edit{
	The probability $\textnormal{Pr}(\mathcal{T}_0)=\left((1-\alpha_n)^{\frac{1}{\alpha_n}}\right)^{k_n}$ is upper-bounded by $e^{-k_n}$ so, by the lemma's assumption $k_n=\Omega(1)$,
	\begin{equation*}
	\limsup_{n\to\infty} \textnormal{Pr}(\mathcal{T}_0) < 1.
	\end{equation*}
}
Consequently, $P_{e}^{(n)}$ \edit{may tend to zero as $n \to \infty$} only if 
\begin{align}
E_n  & = \Omega\left(\log \el_n \right).\notag
\end{align}
This proves Lemma~\ref{Lem_energy_bound}.

\section*{Acknowledgment}
The authors wish to thank the Associate Editor A.~Anastasopoulos and the anonymous referees for their valuable comments.

\end{appendices}

\bibliography{Bibliography.bib}
\bibliographystyle{IEEEtran}

\end{document}